\documentclass[twocolumn,prd,a4paper,amsfonts,amssymb,nobibnotes,nofootinbib,superscriptaddress,fleqn]{revtex4-2}%
\usepackage[dvipdfmx]{graphicx,xcolor}
\usepackage[section]{placeins}
\usepackage{afterpage}
\usepackage{amsfonts,amsmath,amssymb}
\usepackage{array}
\usepackage{bm}
\usepackage{caption}
\usepackage{CJK}
\usepackage{float}
\usepackage{mathrsfs}
\usepackage{subcaption}

\allowdisplaybreaks

\usepackage{comment}
\usepackage{braket}

\usepackage{amsthm}
\newtheorem{definition}{Definition}
\newtheorem{prop}{Proposition}

\newtheorem{example}{Example}
\newtheorem{rem}{Remark}

\begin{document}
\preprint{CHIBA-EP-274}
\title{
Quark confinement consistent with holography 
\\
due to hyperbolic magnetic monopoles and hyperbolic vortices 
\\
unifiedly reduced from symmetric instantons 
\\
}
\author{Kei-Ichi \surname{Kondo}}
\email{kondok@faculty.chiba-u.jp
\\
This is an extended version of a talk given at the workshop, Nonperturbative methods in QFTs--Kondo Condensation--, 25-26 March, 2025, Chiba University. 
}
\affiliation{Department of Physics, Graduate School of Science, Chiba University, \\
1-33 Yayoi-cho, Chiba, Chiba 263-8522, Japan
}
\affiliation{
Research and Education Center for Natural Sciences, Keio University, \\
4-1-1 Hiyoshi, Yokohama, Kanagawa 223-8521, Japan
}
\pacs{PACS number}

\begin{abstract}
We give a review on hyperbolic magnetic monopoles and hyperbolic vortices obtained in the unified way through the conformal equivalence  by the dimensional reduction from the symmetric instantons with various spatial symmetries in the four-dimensional Euclidean Yang-Mills theory. 
They are used to understand quark confinement in the sense of the area law of the Wilson loop average in a semi-classical picture from a unified treatment of Atiyah's hyperbolic magnetic monopole and Witten-Manton's hyperbolic vortex. 
In this way quark confinement is shown to be realized by the the non-perturbative vacuum disordered by these topological defects. 

For this purpose we start from the 4-dim. Euclidean Yang-Mills theory and require the conformal equivalence between the 4-dim. Euclidean space and the possible curved spacetimes with some compact dimensions.
This requirement forces us to restrict the gauge configurations of 4-dim.Yang-Mills instantons to those with some spatial symmetries (called symmetric instantons) which are identified with magnetic monopoles and vortices living in the lower-dimensional curved hyperbolic spacetime with constant negative curvature through the dimensional reduction. 
At the same time, this scheme caused by the dimensional reduction give a holographic description of hyperbolic magnetic monopole dominance on AdS3 in the rigorous way without any further assumptions, which does not hold in the flat Euclidean case. 
This unified treatment of two topological defects is shown to give the semi-classical picture for quark confinement in the sense of Wilson. 
We give the understanding of the result from the viewpoint of the gauge-covariant  Cho-Duan-Ge-Faddeev-Niemi decomposition for the gauge field. 

\end{abstract}

\maketitle

\tableofcontents{}

\section{Introduction}
\label{intro}



\subsection{Dual superconductivity picture for quark confinement} 


Quarks are supposed to be the most fundamental building blocks of the matter according to \textit{quark model} \cite{quark_model}.  
Howeverquarksrs have never been observed and cannot be extracted in its isolated form, which is the fact established in experiments called \textit{quark confinement}.
Quark confinement is recognized as one of the most important problems to be solved in elementary particle physics since the proposal of the quark model. 

To explain quark confinement, the \textit{dual superconductivity picture} was proposed long ago by Nambu, `t Hooft, Mandelstam, Polyakov \cite{dualsuper}, which is now one of the most promising mechanism for quark confinement. 
The dual superconductor picture is based on the \textit{electric-magnetic duality} of the ordinary superconductivity where the electric objects (electric charge, electric current and the electric field) are exchanged by the magnetic objects (magnetic charge, magnetic current and the magnetic field) and vice versa. 
If this picture is true, the chromo-electric flux originated by color charges are squeezed by the dual Meissner effect caused by the vacuum condensation of the magnetic charges, which should be compared with the ordinary superconductor where the magnetic flux is squeezed by the Meissner effect caused by the vacuum condensationon of the Cooper pairs of the electric charges (e.g., two electrons). 
See Fig.~\ref{fig:dual-super}.

We apply the dual superconductivity picture to \textit{a meson} which is composed of a quark and an antiquark according to quark model. 
It is known mainly by numerical simulations that the static potential $V(r)$ between a pair of a quark and an antiquark has a piece of the \textit{linear potential} proportional to their mutual distance $r$ (which is dominant in the long distance) in addition to the potential of the Coulomb type (which is dominant in the short distance):
\begin{equation}
 V(r) = - \frac{C}{r} + \sigma r .
\end{equation}
Consequently, we will need infinite energy to separate either a quark or an antiquark as one of the constituent of a meson, while a meson is observed in isolated form in experiment. 
The same explanation works also for a \textit{baryon} composed of three quarks or three antiquarks. 
Therefore, quarks and antiquark are considered to be confined in  \textit{hadrons}, i.e., mesons and baryons as the composite particles, and they cannot be observed in its isolated form. 

Thus the dual superconductivity picture is able to explain the experimental fact: quark confinement. 
In some models \cite{Polyakov75,Polyakov77,SW94} of quantum field theory, the dual superconductivity was proved to be the correct mechanism for quark confinement. 
To be precise, in general, quark confinement must be extended to \textit{color confinement} to explain confinement of all the other particles with colors such as gluons and light quarks. 
In this article, we focus on the issue of quark confinement.


\begin{figure}[htb] 
\begin{center}
\includegraphics[height=2.6cm]{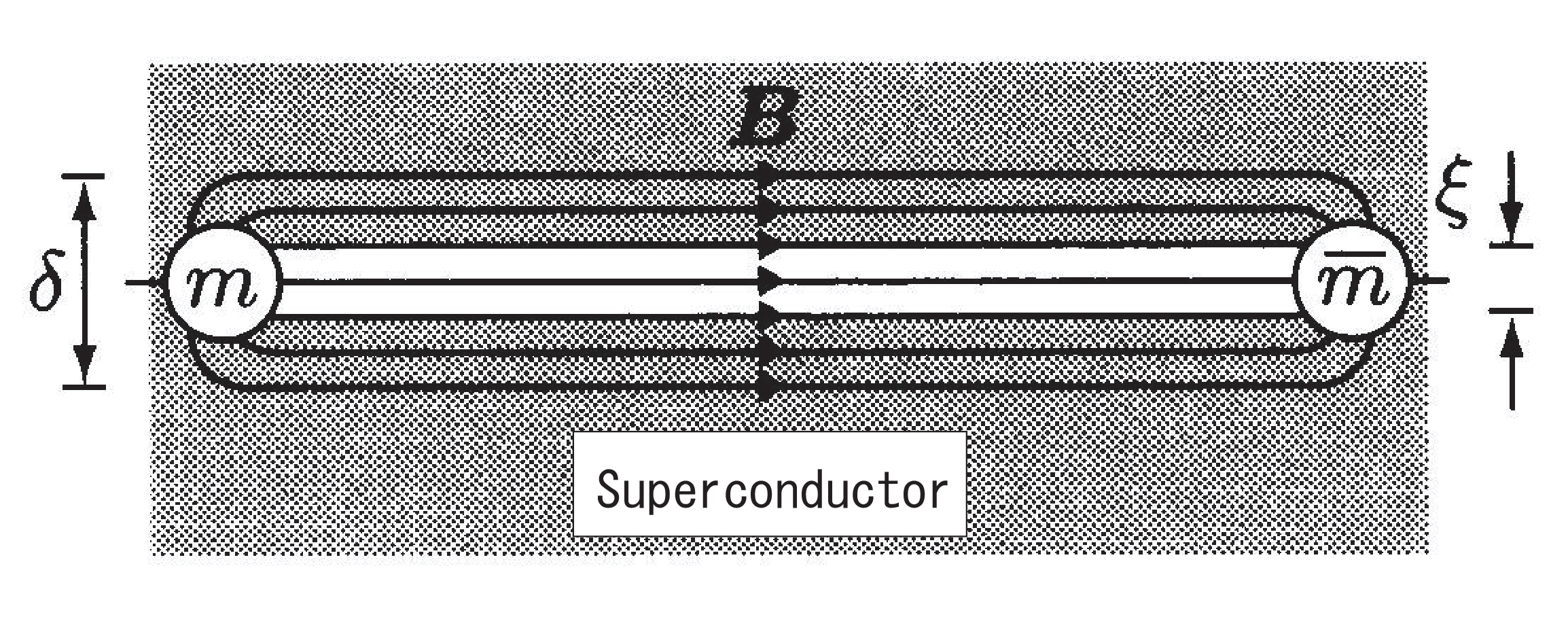}
\includegraphics[height=2.6cm]{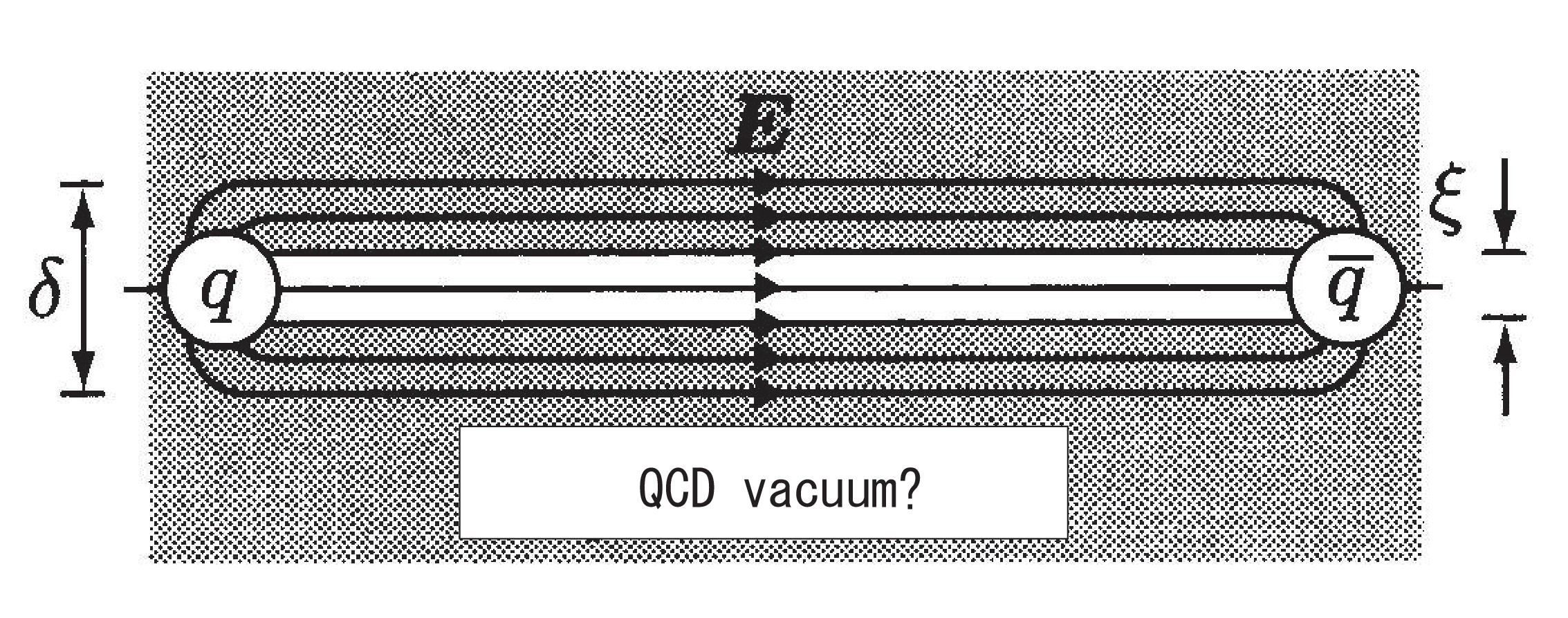}
\end{center}
\caption{
(Upper) Magnetic charges in a type II superconductor and the magnetic field they generate. 
(Lower) Electric charges in a dual superconductor and the electric field they generate.
}
\label{fig:dual-super}
\end{figure}



\subsection{Possible theories for ordinary and dual superconductors} 


To construct the field theory for the dual superconductivity, we recall the theoretical treatment of the ordinary superconductivity. 
The most fundamental theory for the ordinary superconductivity is the so-called BCS theory constructed by Bardeen, Cooper and Schrieffer \cite{BCS}. 
However, this theory is not written in the form of the field theory. 
The field theory for the superconductivity is given by the Ginzburg-Landau theory \cite{GL}, which is however still a non-relativisticic theory.  
The relativistic field theory model of the Ginzburg-Landau theory for the ordinary superconductivity is given by the Abelian U(1) gauge-scalar model, namely, the Maxwell-Higgs theory. 

The Maxwell-Higgs theory is written in terms of the Abelian U(1) gauge field (i.e., the Maxwell field) $A_\mu(x)$ and a complex scalar field $\phi(x)$. 
Here we need two kinds of fields to explain superconductivity, because one is the scalar field with a certain electric charge which is needed to cause the vacuum condensation and another is the gauge field which acquires the mass leading to the Meissner effect caused by the vacuum condensation. 
The squeezed magnetic flux in the superconductor (of the second type) is represented by the Nielsen-Olesen vortex \cite{NO}, a relativistic version of the Abrikosov vortex \cite{Abrikosov}, which is a topological  soliton characterized by the non-trivial first Homotopy class $\pi_1(S^1)=\mathbb{Z}$. 


$\bullet$ Superconductivity
\\
Magnetic charges generate 
the flux of magnetic field (understood as a vortex) which is  squeezed by the Meissner effect [Field affected: gauge field]
\\
This is due to the vacuum condensation of an electric object (namely, the Cooper pair 2e) [Field that creates condensation: scalar field]
\noindent
If electromagnetic duality holds (See Fig.~\ref{fig:dual-super})

$\bullet$ Dual superconductivity?
\\
Color charge generates 
the flux of color electric field (understood as a vortex) which is squeezed by the dual Meissner effect?
\\
This is due to the vacuum condensation of a magnetic object (magnetic monopole)?

Not only gauge field but also scalar field is required.
In fact, we know the Nielsen-Olesen vortex and the `tHooft-Polyakov magnetic monopole as such examples,
However, the Yang-Mills theory does not have scalar field.


If the dual theory describing the dual superconductivity is obtained by the electric-magnetic dual transformation, the theory must include the scalar field in addition to the gauge field. 
Such a dual scalar field with a certain color charge will have the vacuum condensation of magnetic objects to cause the dual Meissner effect and hence the color electric flux is squeezed to yield the linear potential between two color charges. 
In the non-Abelian gauge theory, the SU(2) gauge-scalar theory, the Yang-Mills-Higgs theory, there is a topological soliton called the `t Hooft-Polyakov magnetic monopole characterized by the non-trivial second Homotopy class $\pi_2(S^2)=\mathbb{Z}$. 
To explain the dual superconductivity, we need the scalar field.  

However, the ordinary formulation of the Yang-Mills theory for the strong interactions is given by the Yang-Mills field alone and has no scalar field, although the electro-weak interaction is described by the Yang-Mills-Higgs theory. 
Therefore, we have a question what plays the role of such a scalar field in the pure Yang-Mills theory.  
It is possible how to describe the color flux tube as a topological soliton in the Yang-Mills theory without the scalar field.  
An approach is the method based on change of variables, see e.g.  \cite{PR} for a review.


\subsection{Yang-Mills theory and topological solitons} 


We consider the Yang-Mills theory on the four-dimensional Euclidean space-time $\mathbb{E}^4(x^1,x^2,x^3,x^4=t)$, which is invariant under \textbf{conformal transformation}\index{conformal transformation}. 
In other words, it has \textbf{conformal invariance}\index{conformal invariance}.
The self-dual Yang-Mills equation on $\mathbb{E}^4$ is also conformally invariant.

It is known that the only topological solitons in Yang-Mills theory are instantons as the solutions of the self-dual Yang-Mills equation in four-dimensional Euclidean spacetime $\mathbb{E}^4$.
It is also known that various low-dimensional integrable equations can be obtained from the self-dual Yang-Mills equation $*\mathscr{F}=\pm \mathscr{F}$ in four-dimensional space $\mathbb{E}^4$ by various dimensional reductions.
Korteweg–De Vries (KdV) equation, Kadomtsev–Petviashvili (KP) equation, sine-Gordon equation, non-linear Schroedinger (NLS) equation, Liouville equation, etc. . .
See e.g., Mason and Woodhouse(1996)\cite{MW96}.

For example, if we perform the dimensional reduction to eliminate the Euclidean time $t$,
\begin{align}
 & \mathscr{A}_\mu(x^1,x^2,x^3,x^4=t) \nonumber\\
 \to & \mathscr{A}_\mu(x^1,x^2,x^3)=(\mathscr{A}_k(x^1,x^2,x^3),\Phi(x^1,x^2,x^3)) ,
\end{align}
the 4-dimensional self-dual equation reduces to a 3-dimensional magnetic monopole equation  where the time component of the original gauge field is identified with the scalar field. The solution is a 3-dimensional magnetic monopole. 

However, if we substitute the 3-dimensional static magnetic monopole solution obtained in this way into the action of the 4-dimensional Yang-Mills theory, the action diverges from the integration with respect to the time $t$ because this solution does not depend on the time. Therefore, the three-dimensional magnetic monopole configuration obtained in this way does not contribute to the path (functional) integral of the four-dimensional Yang-Mills theory, and therefore such magnetic monopoles cannot give the main contribution to quark confinement.
This is also the case for the other low-dimensional topological solitons such as vortices obtained by the dimensional reduction. 
One way to avoid this inconvenience is to appropriately compactify space-time ($\mathbb{R}^4 \to \mathbb{R}^3 \times S^1, \mathbb{R}^2 \times S^1 \times S^1=\mathbb{R}^2 \times T^2,\mathbb{R}^2 \times S^2,...$), but are there any criteria for choosing the appropriate compactification? 
See e.g., (Tanizaki and \"Unsal (2022)\cite{TU22}, Hayashi and Tanizaki (2024)\cite{HT24}, Hayashi, Misumi and Tanizaki (2024)\cite{HMT24}, Güvendik, Schaefer and \"Unsal (2024)\cite{GSU24}). 


\subsection{Our proposal for deriving dual superconductor picture} 


In this article we obey the following guiding principles:

\noindent
$\bullet$ conformal equivalence:

The four-dimensional Yang-Mills theory is conformally invariant. The self-dual equations are also conformally invariant. By choosing a spacetime that can be obtained by a conformal transformation, such kind of arbitrariness (associated with the choice of spacetimes to be compactified) is prevented.

\noindent
$\bullet$ spatially symmetric instantons:

By carefully selecting the spatial symmetry respected by instantons, it is possible to prevent the divergence of the action in the four-dimensional Yang-Mills theory even when using low-dimensional topological configurations, and to make it contribute to the four-dimensional path (functional) integral.

\noindent
$\bullet$ dimensional reductions:

When a certain spatial symmetry is required for the instanton configuration, the associated dimensional reduction occurs, which determines the low-dimensional topological configuration obtained by the dimensional reduction.

By obeying these three principles, starting from the symmetric instanton of the four-dimensional Yang-Mills theory, we can write down a dimensionally reduced effective theory that unifies the magnetic monopole and vortex in a conformally equivalent manner.

The magnetic monopole and vortex obtained in this way are equivalent and can be considered as configurations that provide a quasi-classical mechanism of confinement.

\noindent
$\bullet$ singular instantons:
\\
By constructing an instanton solution with a nontrivial holonomy around the singularity, we can construct an instanton with fractional topological charge. Starting from this, the above procedure can be carried out to obtain quark confinement in the sense of the area law of the vacuum expectation value of the Wilson loop operator.

This paper is organized as follows.

In section II, we introduce the conformal equivalence to specify the spacetime to decompose the 4-dimensional space into the hyperbolic space $\mathbb{H}^{4-\nu} $ times a compact space $S^{\nu}$: $\mathbb{H}^{4-\nu} \times S^{\nu}$. 

In section III, we explain the symmetric instanton with a spatial symmetry is  associated with the dimensional reduction of the spacetime given in the previous section.

In section IV, by using the conformal equivalence and the symmetric instanton, we can obtain the unified magnetic monopole and the vortex which are equivalent to each other.

In section V, we show that the Bogomolnyi equation for the 
 magnetic monopole in the 3-dimensional hyperbolic space $\mathbb{H}^3$ is equivalent to the vortex equation in the 2-dimensional hyperbolic plane $\mathbb{H}^2$. 

In section VI, we give some models which describe the hyperbolic space $\mathbb{H}^3$.

In section VII, we give explicit analytical solutions of the hyperbolic vortex equation and they indeed the topological soliton with non-trivial winding numbers. 

In section VIII, we give explicit analytical solutions of the Bogomolnyi equation for the hyperbolic magnetic monopole and show that they indeed the topological soliton with non-trivial topological charge, namely, the winding numbers. 

In section IX,  we explain the hyperbolic magnetic monopole obeys the holographic principle. 

In section X, we give the solution based on the superpotential based on the Ansatz for the 4-dimensional instantons due to the dimensional reduction. The solutions give the same results as those given by directly solving the monopole equation and vortex equation.

In section XI,  we give a symmetric but singular instantons which have non-integral topological charge. 

In section XII, we calculate the expectation value of the Wilson loop operator to show the area law, which is a criterion for quark confinement due to Wilson. 
Using the holography principle, the calculation of the non-Abelian Wilson loop operator in the original Yang-Mills theory reduces to the calculation of the average of the Abelian Wilson loop in the 2-dimensional conformal boundary.

In section XII, we give another explanation how the Witten transformation works from the viewpoint of the gauge covariant decomposition of the gauge field called the Cho-Duan-Ge-Faddeev-Niemi decomposition. 

This final section is devoted to conclusions and discussions. 
In Appendix A, we give a summary of the Cho-Duan-Ge-Faddeev-Niemi decomposition. 
In Appendix B, we give a summary of the ADHM construction. 
In Appendix C, we give a derivation of the circle-symmetric instanton, namely, the hyperbolic magnetic monopole based on the ADHM construction.

\section{Conformal equivalence}

We start from the flat 4-dimensional Euclidean space $\mathbb{E}^4$ with the Cartesian coordinates $(x^1,x^2,x^3,x^4)=(x,y,z, t)$.
The metric of $\mathbb{E}^4$ is given by using the Cartesian coordinates $(x^4,x^1,x^2,x^3)$:
\begin{align}
(ds)^2(\mathbb{E}^4) = (dx^4)^2 + (dx^1)^2 +(dx^2)^2 +(dx^3)^2 .
\end{align}

\noindent
(I) 
We introduce the coordinates $(\rho,\varphi)$ in the 2-dimensional space $(x^1,x^2)$ to rewrite the metric  in the cylindrical coordinates $(\rho,\varphi,x^3,x^4)$:
\begin{align}
(ds)^2(\mathbb{E}^4) 
=  (d\rho)^2 + \rho^2(d\varphi)^2 + (dx^3)^2 + (dx^4)^2 .
\label{R4toH3}
\end{align}
We adopt $\rho^2$ as a \textbf{conformal factor}\index{conformal factor} to further rewrite the metric:
\begin{align}
(ds)^2(\mathbb{E}^4) = \rho^2\left[\frac{(dx^3)^2 + (dx^4)^2 + (d\rho)^2}{\rho^2} + (d\varphi)^2\right].
\label{R4toH3b}
\end{align}
Therefore, we obtain a conformal equivalence:
\begin{align}
&\mathbb{R}^4 = \mathbb{R}^3 \times \mathbb{R}^1
\nonumber\\
&\begin{array}{cccccccc}
 \rightarrow&\mathbb{R}^4  &\backslash& \mathbb{R}^2 &\simeq&  \mathbb{H}^3& \times &S^1\\
&\rotatebox{90}{$\in$}&&\rotatebox{90}{$\in$}&&\rotatebox{90}{$\in$}&&\rotatebox{90}{$\in$}\\
&(x^1,x^2,x^3,x^4)&&(x^3,x^4)&&(\rho,x^3,x^4)&&\varphi
\end{array}
\end{align}
$\mathbb{H}^3(\rho,x^3,x^4)$ is a \textbf{hyperbolic 3-space} with coordinates $(\rho, x^3,x^4)$: $x^3,x^4 \in (-\infty,+\infty)$ and $\rho \in (0,\infty)$, and has the metric $g_{\mu\nu}=\rho^{-2}\delta_{\mu\nu}$ and the negative constant curvature $-\frac{1}{2}$. 
This is the \textbf{upper half space model}\index{upper half space model} with $\rho>0$. 
Here $\rho = 0$ is a singularity, therefore the corresponding two-dimensional space, i.e., the $(x^3,x^4)$ plane with $\rho = 0$ must be excluded from $\mathbb{R}^4$.

 $\mathbb{R}^4 \backslash \mathbb{R}^2$ means excluding $(x^3,x^4) \in \mathbb{R}^2$ from $\mathbb{R}^4$.
See Fig.\ref{conformal_equiv1}.
$S^1(\varphi)$ is a 1-dimensional unit sphere, i.e., a unit circle with the coordinate $\varphi \in [0,2\pi)$. 
The special orthogonal group $SO(2)$ acts on $S^1(\varphi)$ in the standard way.

\begin{figure}[htb]
\begin{center}
\includegraphics[scale=0.50]{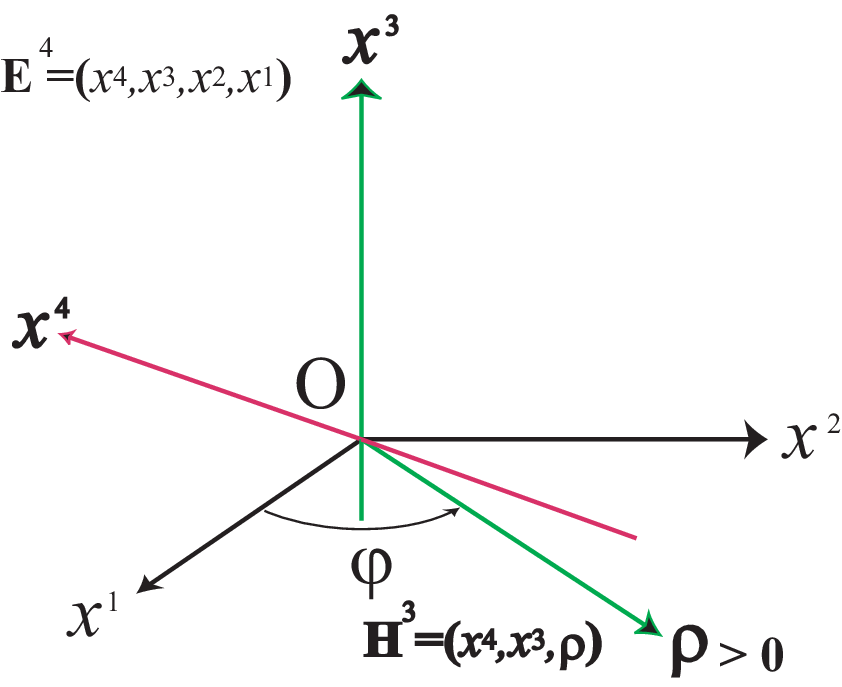}
\end{center}
\caption{
4-dim. Euclidean space $\mathbb{E}^4(x^1,x^2,x^3,x^4)$ versus 3-dim. hyperbolic space $\mathbb{H}^3(\rho,x^3,x^4)$.
}
\label{conformal_equiv1}
\end{figure}

It is also possible to arbitrarily change the curvature of $\mathbb{H}^3(\rho,x^3,x^4)$ by using the conformal factor $\frac{\rho^2}{\rho_0^2}$:
\begin{equation}
(ds)^2(\mathbb{E}^4) = \frac{\rho^2}{\rho_0^2}\left[\rho_0^2\frac{(dx^4)^2 + (dx^3)^2 + (d\rho)^2}{\rho^2} + \rho_0^2(d\varphi)^2\right] .
\label{R4toH3b}
\end{equation}
That is, the conformal equivalence is rewritten as
\begin{align}
\mathbb{R}^4 \backslash \mathbb{R}^2 \simeq \mathbb{H}^3(\rho_0) \times S^1(\rho_0) .
\end{align}
Here, $\mathbb{H}^3(\rho_0)$ denotes a three-dimensional hyperbolic space with negative constant curvature $-\frac{1}{2\rho_0^2}$, and $S^1(\rho_0)$ denotes a circle with radius $\rho_0$.
In the limit of infinite radius $\rho_0 \to \infty$ of $S^1(\rho_0)$, $S^1(\rho_0)$ approaches $\mathbb{R}$, and the curvature of $\mathbb{H}^3(\rho_0)$ approaches zero, and $\mathbb{H}^3(\rho_0)$ reduces to the flat Euclidean space $\mathbb{E}^3$.


Let us consider another example of conformal equivalence.

\noindent
(II) 
We introduce the polar coordinates $(r,\theta,\varphi)$ for the 3-dimensional space $(x^1,x^2,x^3)$ to rewrite the metric as 
\begin{align}
(ds)^2(\mathbb{E}^4) =& (dx^4)^2 + (dr)^2 +r^2((d\theta)^2 +\sin^2\theta(d\varphi)^2) \nonumber\\
& (r := \sqrt{x_1^2 + x_2^2 + x_3^2}) .
\end{align}
Then, we adopt $r^2$ as a \textbf{conformal factor}\index{conformal factor} to rewrite the metric 
\begin{equation}
(ds)^2(\mathbb{E}^4) =r^2\left[ \frac{(dx^4)^2 + (dr)^2}{r^2} +(d\theta)^2 +\sin^2\theta(d\varphi)^2\right] .
\label{R4toH2}
\end{equation}
Therefore, we obtain the \textbf{conformal equivalence}\index{conformal equivalence}:
\begin{align}
&\mathbb{R}^4 = \mathbb{R}^2 \times \mathbb{R}^2 \nonumber\\
&\begin{array}{cccccccc}
\rightarrow&\mathbb{R}^4  &\backslash& \mathbb{R}^1 &\simeq&  \mathbb{H}^2&\times &S^2\\
&\rotatebox{90}{$\in$}&&\rotatebox{90}{$\in$}&&\rotatebox{90}{$\in$}&&\rotatebox{90}{$\in$}\\
&(x^4,x^1,x^2,x^3)&&x^4&&(x^4,r)&&(\theta,\varphi)
\end{array}
\end{align}

\begin{figure}[htb]
\begin{center}
\includegraphics[scale=0.50]{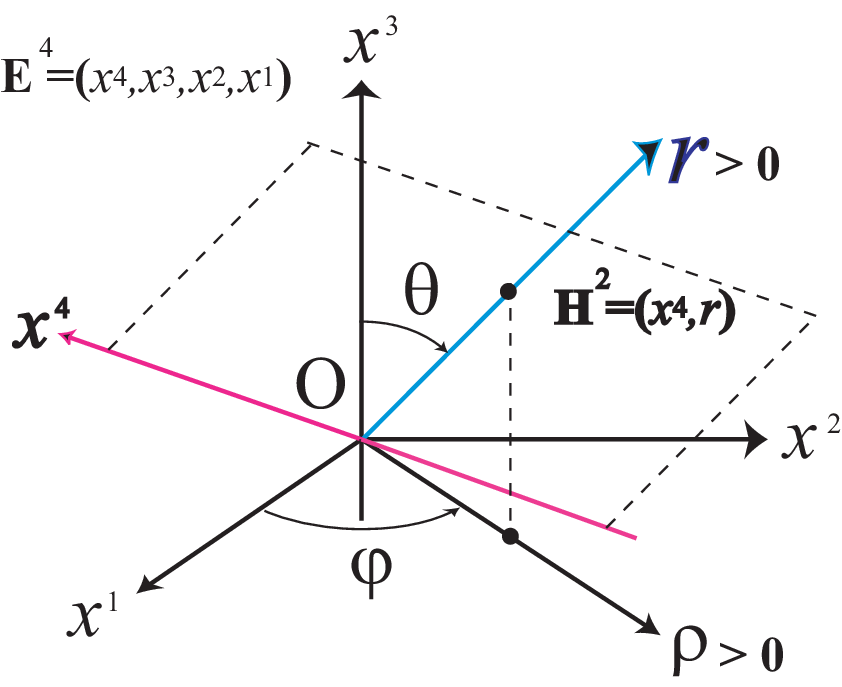}
\end{center}
\caption{
4-dim. Euclidean space $\mathbb{E}^4(x^1,x^2,x^3,x^4)$ versus 2-dim. hyperbolic space $\mathbb{H}^2(r,x^4)$.
}
\label{conformal_equiv2}
\end{figure}

$\mathbb{H}^2(x^4,r)$ is a \textbf{hyperbolic plane}\index{hyperbolic plane} with coordinates $(x^4,r)$: $x^4 \in (-\infty,\infty)$, $r \in(0,\infty)$, and has the metric $g_{\mu\nu}=r^{-2}\delta_{\mu\nu}$ and the negative constant curvature $(-\frac{1}{2})$. 
This is the \textbf{upper half plane model}\index{upper half plane model} with $r >0$. 
Here $r=0$ is a singularity, therefore the one-dimensional space $\mathbb{R}^1$ with $r = 0$, i.e., the $x^4$-axis must be excluded from $\mathbb{R}^4$.

$S^2(\theta,\varphi)$ is a two-dimensional unit sphere with coordinates $(\theta, \varphi)$: $\theta \in [0,\pi)$, $\varphi \in [0,2\pi)$ and has a positive constant curvature (2). 
The special orthogonal $SO(3)$ acts on $S^2(\theta,\varphi)$ in the standard way.

It is also possible to arbitrarily change the curvature of $\mathbb{H}^2(x^4,r)$ using the conformal factor $\frac{r^2}{r_0^2}$:
\begin{equation}
(ds)^2(\mathbb{E}^4) =\frac{r^2}{r_0^2}\left\{\frac{(dx^4)^2 + (dr)^2}{r^2/r_0^2} +r_0^2[(d\theta)^2 +\sin^2\theta(d\varphi)^2]\right\} .
\label{R4toH2b}
\end{equation}
That is, we have the conformal equivalence:
\begin{align}
\mathbb{R}^4 \backslash \mathbb{R}^1 \simeq \mathbb{H}^2(r_0)\times S^2(r_0) .
\end{align}
Here, $r_0$ is the radius of the 2-sphere $S^2$, $S^2(r_0)$ is the 2-sphere of radius $r_0$, and $\mathbb{H}^2(r_0)$ is the hyperbolic plane with negative constant curvature $(-\frac{1}{2r_0^2})$.
In particular, in the limit of $r_0 \to \infty$, the curvature approaches zero, $S^2(r_0)$ approaches $\mathbb{R}^2$, and $\mathbb{H}^2(r_0)$ reduces to the flat Euclidean $\mathbb{E}^2$.

\begin{rem}

The standard metric on $\mathbb{H}^3\times S^1$ is conformal to the standard metric on $S^4 \ \backslash \ S^2$, i.e., we have \textbf{conformal equivalence}\index{conformal equivalence}:
\begin{align}
S^4 \ \backslash \ S^2 \equiv \mathbb{H}^3\times S^1
\end{align}
$\mathbb{H}^3$ is obtained by quotienting with the $SO(2)$ action.
(To avoid singularities, $SO(2)$ must act freely. Therefore, we need to remove the special $S^2$ orbitals.)
The curvature of $\mathbb{H}^3$ is related to the length of the circle ($S^1$), and when the circumference is normalized to $2\pi$, $\mathbb{H}^3$ has curvature $(-1)$.

Consider $S^4$ as a compactified $\mathbb{R}^4$ with an additional point at infinity. Since the self-dual Yang-Mills equations are conformally invariant, an instanton on $S^4$ is equivalent to an instanton on $\mathbb{R}^4$ with the appropriate boundary conditions. If the Cartesian coordinates of $\mathbb{R}^4$ are $x^\mu(\mu = 1, \cdots, 4)$, then the metric is
\begin{equation}
(ds)^2(\mathbb{E}^4) = (dx^1)^2 + (dx^2)^2 + (dx^3)^2 + (dx^4)^2
\end{equation}
where
\begin{align}
x^1 + ix^2 = \rho e^{i\varphi} \ \ (\rho \geq 0 \ , \ 0 \leq \varphi <2\pi)
\end{align}
then the circular action on $\mathbb{R}^4$ can be defined by the standard rotation of $\varphi$. 
The set of fixed points is the $x^3- x^4$ plane ($x^1 = x^2 = 0$, i.e., $\rho = 0$), which by compactification becomes the two-dimensional sphere $S^2$. By removing this plane $\mathbb{R}^2$ from $\mathbb{R}^4$ and taking the quotient by the circular action, we obtain $\mathbb{H}^3$:
\begin{equation}
(ds)^2(\mathbb{E}^4) = (d\rho)^2 + \rho^2(d\varphi)^2 + (dx^3)^2 + (dx^4)^2 \ \ (\rho>0)
\end{equation}
This is then conformally equivalent:
\begin{equation}
(ds)^2(\mathbb{H}^3\times S^1) = \frac{1}{\rho^2}((dx^3)^2 + (dx^4)^2 + (d\rho)^2) + (d\varphi)^2
\end{equation}
This is a product metric on $\mathbb{H}^3\times S^1$. By taking the quotient by $SO(2)$, we obtain the metric $(ds)^2(\mathbb{H}^3)$ of the upper half-space model $\mathbb{H}^3$. The removed plane $\mathbb{R}^2$, together with the point at infinity, can be interpreted as the boundary $\partial\mathbb{H}^3 \cong S^2$ of $\mathbb{H}^3$.

\end{rem}

\section{Symmetric instantons and dimensional reductions}

We start from the $D=4$ Euclidean Yang-Mills theory on $\mathbb{E}^4(x^1,x^2,x^3,x^4)$ with the Cartesian coordinates $x^1,x^2,x^3,x^4$:
\begin{align}
&\mathscr{L} = \frac12 {\rm tr}(\mathscr{F}_{\mu\nu}(x) \mathscr{F}_{\mu\nu}(x)) , \ x = (x^1,x^2,x^3,x^4) ,
\nonumber\\
&\mathscr{F}_{\mu\nu}(x) := \partial_\mu \mathscr{A}_\nu(x) - \partial_\nu \mathscr{A}_\mu(x)  - ig [\mathscr{A}_\mu(x), \mathscr{A}_\nu(x)] , \nonumber\\
&\mathscr{A}_\mu(x) := \mathscr{A}_\mu^A(x) \frac{\sigma_A}{2} ,
\end{align}
with the metric 
\begin{equation}
(ds)^2(\mathbb{E}^4) =  (dx^1)^2 + (dx^2)^2 + (dx^3)^2 + (dx^4)^2 .
\end{equation}

We reconsider the solution of the self-dual Yang-Mills equation in the four-dimensional Euclidean space $\mathbb{E}^4(x^1,x^2,x^3,x^4)$ from the viewpoint of symmetry.
The self-dual Yang-Mills equation is given by
\begin{align}
&  *\mathscr{F}_{\mu\nu}(x^1,x^2,x^3,x^4) =\mathscr{F}_{\mu\nu}(x^1,x^2,x^3,x^4) \nonumber\\
 \Leftrightarrow & 
  \frac{1}{2} \epsilon_{\rho\sigma\mu\nu} \mathscr{F}_{\rho\sigma}(x^1,x^2,x^3,x^4) =\mathscr{F}_{\mu\nu}(x^1,x^2,x^3,x^4) .
\end{align}

\noindent
(0)
First, we consider a solution  for the gauge field that has the translational invariance. For example, a solution that has \textbf{translational invariance in the time $t=x^4$} is equivalent to considering a solution that does not depend on the Euclidean time $x^4$.
(Recall cyclic coordinates in analytical mechanics.)
\begin{align}
 &(\mathscr{A}_{1}(\bm{x},t), \mathscr{A}_{2}(\bm{x},t), \mathscr{A}_{3}(\bm{x},t), \mathscr{A}_{4}(\bm{x},t)) \nonumber\\
 \to 
 &(\mathscr{A}_{1}(\bm{x}), \mathscr{A}_{2}(\bm{x}), \mathscr{A}_{3}(\bm{x}), \Phi(\bm{x})) .
\end{align}
The $x^4$-independent solution of the self-dual equation 
 reduces to the solution of \textbf{Bogomolny equation}\index{Bogomolny equation} on $\mathbb{E}^3$:
\begin{equation}
 (*\mathscr{F})_{\ell 4}(x^1,x^2,x^3) = \mathscr{D}_\ell \Phi(x^1,x^2,x^3) , \ (\ell=1,2,3). 
\end{equation}
In fact, for the self-dual equation for $\mu, \nu=\ell , 4$
\begin{align}
 \frac{1}{2} \epsilon_{jk\ell 4} \mathscr{F}_{jk}(x^1,x^2,x^3) = \mathscr{F}_{\ell 4}(x^1,x^2,x^3) ,
\end{align}
the right-hand side reads 
\begin{align}
 & \mathscr{F}_{\ell 4}(x^1,x^2,x^3) \nonumber\\
 =& \partial_\ell \mathscr{A}_4(x^1,x^2,x^3) - \partial_4 \mathscr{A}_\ell (x^1,x^2,x^3)   \nonumber\\
& -ig[ \mathscr{A}_\ell (x^1,x^2,x^3), \mathscr{A}_4(x^1,x^2,x^3)] 
\nonumber\\
 =& \partial_\ell \mathscr{A}_4(x^1,x^2,x^3) -ig[ \mathscr{A}_\ell (x^1,x^2,x^3), \mathscr{A}_4(x^1,x^2,x^3)] \nonumber\\
 =& \mathscr{D}_\ell \Phi(x^1,x^2,x^3), \ \Phi(x^1,x^2,x^3) := \mathscr{A}_4(x^1,x^2,x^3) ,
\end{align}
where we have used $\partial_4 \mathscr{A}_\ell (x^1,x^2,x^3)=0$.

The solution of the Bogomolny equation is called the \textbf{Prasad-Sommerfield magnetic monopole}.
However, this solution leads to a divergent four-dimensional action:
\begin{equation}
S=\int_{-\infty}^{\infty}dx^4 \left[ \int dx^1 dx^2 dx^3  \mathscr{L}(x^1,x^2,x^3) \right] =\infty ,
\end{equation} 
even if $\int dx^1 dx^2 dx^3  \mathscr{L}(x^1,x^2,x^3)<\infty$ 
because of the $t$-independence.
Therefore, the PS magnetic monopole does not contribute to the path integral, because  
\begin{equation}
\exp( -S/\hbar) =0 .
\end{equation} 
Thus, the PS magnetic monopole does not play any role in the mechanism for quark confinement.

\noindent
(I)
Next, we consider solutions with \textbf{spatial rotation  symmetry} $S^1\simeq SO(2)$.


The $S^1$ symmetric instanton solution on $\mathbb{R}^4 \backslash \mathbb{R}^2$ that does not depend on the rotation angle $\varphi$ reduces to a magnetic monopole solution on $\mathbb{H}^3$.
Any solution of the Bogomolny equation 
 on $\mathbb{H}^3$ is a $\varphi$-independent solution of the self-dual equation on $\mathbb{R}^4 \backslash \mathbb{R}^2$: 
\begin{align}
 (*\mathscr{F})_{\ell \varphi}(\rho,x^3,x^4) =& \frac{1}{\rho} \mathscr{D}_\ell \Phi(\rho,x^3,x^4), \ (\rho,x^3,x^4) \in \mathbb{H}^3 \nonumber\\
 \Phi(\rho,x^3,x^4) :=& \mathscr{A}_\varphi(\rho,x^3,x^4) . 
\end{align}
Since $S^1$ is compact (unlike $\mathbb{R}^1$), any solution of the Bogomolny equation giving a finite three-dimensional action on $\mathbb{H}^3$ gives a configuration with a finite four-dimensional action if $\int_{0}^{\infty}d\rho \ \rho \int_{-\infty}^{\infty}dx^3 \int_{-\infty}^{\infty}dx^4 \mathscr{L}(\rho,x^3,x^4) <\infty$:
\begin{equation}
S
= \int_{0}^{2\pi}d\varphi \left[ \int_{0}^{\infty}d\rho \ \rho \int_{-\infty}^{\infty}dx^3 \int_{-\infty}^{\infty}dx^4 \mathscr{L}(\rho,x^3,x^4) \right] <\infty .
\end{equation}
Therefore, \textbf{$S^1\simeq SO(2)$ symmetric (or circle symmetric) instantons on $\mathbb{E}^4$ can be reinterpreted as hyperbolic magnetic monopoles on $\mathbb{H}^3$, giving a configuration with a finite four-dimensional action} (Atiyah(1984)\cite{Atiyah84}).
This means that the hyperbolic magnetic monopoles can contribute to the path integral, because  
\begin{equation}
\exp( -S/\hbar) \not= 0 .
\end{equation} 
Thus, \textbf{the hyperbolic magnetic monopoles can play the important role for quark confinement}.

\noindent
(II)
Furthermore, we consider another solution with \textbf{spatial rotation symmetry} $SO(3)$.
The $SO(3)$ (spherically) symmetric instanton on $\mathbb{R}^4 \backslash \mathbb{R}^1$ that does not depend on the rotation angles $\theta,\varphi$ reduce to a vortex on $\mathbb{H}^2(r,x^4)$.
Any solution of the vortex equation on $\mathbb{H}^2(r,x^4)$ is a $\theta,\varphi$-independent solution of the self-dual equation on $\mathbb{R}^4 \backslash \mathbb{R}^1$: 
\begin{align}
\left\{\,
\begin{aligned}
& \partial_4a_r - \partial_ra_4 = \frac{1}{r^2}(1 - \phi_1^2 - \phi_2^2), \\
&\partial_4\phi_1 + a_4\phi_2 = \partial_r\phi_2 - a_r\phi_1, \\
&\partial_4\phi_2 - a_4\phi_1 = -(\partial_r\phi_1 + a_r\phi_2).
\end{aligned} 
\right.
\label{vortex_eq1}
\end{align}
Since $S^2(\theta,\varphi)$ is compact (unlike $\mathbb{R}^2$), any solution of the vortex equation giving a finite two-dimensional  action on $\mathbb{H}^2(r,x^4)$ gives a configuration with a four-dimensional action if $\int_{0}^{\infty}dr \ r^2 \int_{-\infty}^{\infty}dx^4 \mathscr{L}(r,x^4) <\infty$:
\begin{equation}
S
= \int_{0}^{\pi}d\theta \sin \theta \int_{0}^{2\pi}d\varphi \left[ \int_{0}^{\infty}dr \ r^2 \int_{-\infty}^{\infty}dx^4 \mathscr{L}(r,x^4) \right] <\infty .
\end{equation}
Therefore, $SO(3)$ \textbf{spherically symmetric instantons on $\mathbb{E}^4$ can be reinterpreted as vortices on $\mathbb{H}^2$, giving a configuration with a finite four-dimensional action}.
Therefore,  the hyperbolic vortices can contribute to the path integral $\exp( -S/\hbar) \not= 0$ and \textbf{the hyperbolic vortices can play the important role for quark confinement}.

\noindent
The case (I) was first pointed out by Atiyah (1984)\cite{Atiyah84}, (1987)\cite{Atiyah84b}. 
Concrete magnetic monopole solutions on $\mathbb{H}^3$ were explicitly constructed by Chakrabarti (1986)\cite{Chakrabarti86} and Nash (1986)\cite{Nash86}.

\noindent
The case (II) was first used by Witten (1977)\cite{Witten} 
 to find multi-instanton solutions of four-dimensional Yang-Mills theory, which is confirmed by Manton(1978)\cite{Manton78}.
This view was established by Forgacs and Manton (1980)\cite{FM80}. 
It is also written in Manton and Sutcliffe, p.424, section 4.3 (2004)\cite{MS04}.

The Witten ansatz is the most general form for fields that are invariant under transformations that combine space rotations (spacetime symmetries) and gauge transformations (internal symmetries).
This Ansatz has four-dimensional \textbf{cylindrical symmetry}\index{cylindrical symmetry}, i.e., SO(3) rotational symmetry around the t-axis.
The hyperbolic vortex on $\mathbb{H}^2$ consists of a complex scalar field $\phi = \phi_1(r,x^4) + i\phi_2(r,x^4)$ and a gauge potential $a = a_r(r,x^4)dr + a_4(r,x^4)dx^4$. These hyperbolic vortex fields are functions of $x^4$ and $r$.

This \textbf{symmetry reduction}\index{symmetry reduction} reduces the cylindrically symmetric sector of the four-dimensional Yang-Mills theory on $\mathbb{E}^4$ to a $U(1)$ gauge scalar theory in two dimensional $\mathbb{H}^2$. 



\section{Unifying magnetic monopole and vortices}

The metric of the $D=4$ Euclidean space $\mathbb{E}^4$ with the coordinate $(x^1,x^2,x^3,x^4)$ is given by
\begin{equation}
(ds)^2(\mathbb{E}^4) = (dx^4)^2 + (dx^1)^2 + (dx^2)^2 + (dx^3)^2 .
\end{equation}
If the cylindrical coordinates $(\rho,\varphi,x^3,x^4)$: $\rho \in (0,\infty), \varphi \in [0,2\pi)$ are introduced with $x^1 = \rho\cos\varphi, x^2 = \rho\sin\varphi$, the metric can be rewritten as
\begin{equation}
(ds)^2(\mathbb{E}^4) = (dx^4)^2 + (dx^3)^2 + (d\rho)^2 + \rho^2(d\varphi)^2.
\end{equation}
If we require that the metric is independent of the coordinate $\varphi$, then the reduced metric is conformally equivalent to $\mathbb{H}^3$:
\begin{align}
(ds)^2(\mathbb{H}^3) = \frac{1}{\rho^2}((dx^4)^2 + (dx^3)^2 + (d\rho)^2). 
\label{eq:2c}
\end{align}
See Fig.~\ref{conformal_equiv_all}.

\begin{figure}[htb]
\begin{center}
\includegraphics[scale=0.45]{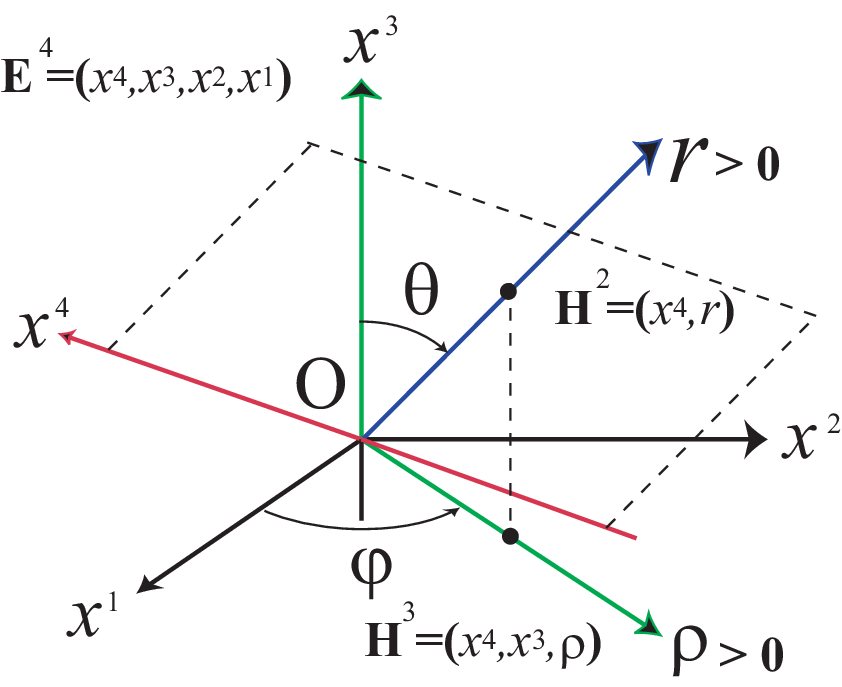}
\end{center}
\caption{
$D=3$ Euclidean space $\mathbb{E}^4(x^1,x^2,x^3,x^4)$ versus $D=3$  hyperbolic space $\mathbb{H}^3(\rho,x^3,x^4)$ with $\rho:=\sqrt{(x^1)^2+(x^2)^2} > 0$
and $D=2$ hyperbolic plane $\mathbb{H}^2(r,x^4)$ with $r:=\sqrt{(x^1)^2+(x^2)^2+(x^3)^2}> 0$.
}
\label{conformal_equiv_all}
\end{figure}

Next, if the polar coordinates $(r,\theta, \varphi,x^4)$: $r \in (0,\infty)$, $\theta \in (0,\pi)$, $\varphi \in [0,2\pi)$ with $x^3 = r\cos\theta, \rho = r\sin\theta$ is introduced, the metric is rewritten into
\begin{equation}
(ds)^2(\mathbb{E}^4) = (dx^4)^2 + (dr)^2 + r^2((d\theta)^2 + \sin^2\theta(d\varphi)^2).
\end{equation}
Here $r := \sqrt{\rho^2 + (x^3)^2}= \sqrt{(x^1)^2+ (x^2)^2+ (x^3)^2}$. 
If we require that the metric is independent of the coordinates $\theta$ and $\varphi$, then the reduced metric is conformally equivalent to $\mathbb{H}^2$:
\begin{align}
(ds)^2(\mathbb{H}^2) = \frac{1}{r^2}((dx^4)^2 + (dr)^2) .
\label{eq:3c}
\end{align}
Looking at the relationship between the two metrics \eqref{eq:2c} and \eqref{eq:3c} (see Figure \ref{conformal_equiv_all}), if we restrict to $\theta = \frac{\pi}{2}$, i.e., $x^3 = 0$, $r = \rho$,
\begin{align}
\left.(ds)^2(\mathbb{H}^2)\right|_{\theta = \frac{\pi}{2}} = \frac{1}{\rho^2}((dx^4)^2 + (d\rho)^2) .
\end{align}
This is the slice $(x^3 = 0)$ of $\mathbb{H}^3$. (The equatorial slice of the unit sphere model of $\mathbb{H}^3$ is the unit disk carrying the hyperbolic metric.)

On the other hand, if we restrict to $\theta = 0$, i.e., $\rho = 0$, then $x^3 = r$ and 
\begin{align}
\left.(ds)^2(\mathbb{H}^2)\right|_{\theta = 0} = \frac{1}{(x^3)^2}((dx^4)^2 + (dx^3)^2)
\end{align}
This is the boundary $\partial\mathbb{H}^3$ of $\mathbb{H}^3$ (known as the hemispherical model of $\mathbb{H}^2$). Since $x^3$ can take both signs, $\partial\mathbb{H}^3$ is two copies of the hyperbolic plane glued together along the $x^4$ axis.

Since \textbf{the self-dual instanton equation is conformally invariant}, the solution is invariant under conformal scaling of the background metric as described above.
In this paper, it is shown that magnetic monopoles in $\mathbb{H}^3$ and vortices in $\mathbb{H}^2$ can be constructed from instantons on $\mathbb{E}^4$ by dimensional reduction. 

\noindent
Case (I): Instantons with $S^1\simeq SO(2) \simeq U(1)$ symmetry are dimensionally reduced to magnetic monopoles on $\mathbb{H}^3$.
On the other hand,
\noindent
Case (II): Instantons with $SO(3) \simeq SU(2)$ symmetry are dimensionally reduced to vortices on $\mathbb{H}^2$.
See Table~\ref{Table:reduction}. 

Case (II): We can consider the instantons with $SO(4)$ symmetry, which is however the spacetime symmetry beyond the spatial symmetry for the 4-dimensional spacetime.  
It should be remarked that $SO(4)$ is larger than the gauge group $SU(2)$ and hence it cannot be compensated by $SU(2)$ gauge transformation. Therefore, this case does not correspond to the symmetric instantons defined in the above. 
See Remark~\ref{rem:caseIII}.


In this way, hyperbolic magnetic monopoles and hyperbolic vortexes are bidirectionally related under conformal equivalence.
A hyperbolic magnetic monopole can be obtained by lifting a hyperbolic vortex to a four-dimensional instanton, then performing dimensional reduction that preserves $SO(2)$ symmetry.
Conversely, a hyperbolic vortex can be obtained by lifting a hyperbolic magnetic monopole to a four-dimensional instanton, then performing dimensional reduction that preserves $SO(3)$ symmetry.
See Fig.~\ref{conformal_unification}. 

\begin{figure}[tb]
\begin{center}
\includegraphics[scale=0.3]{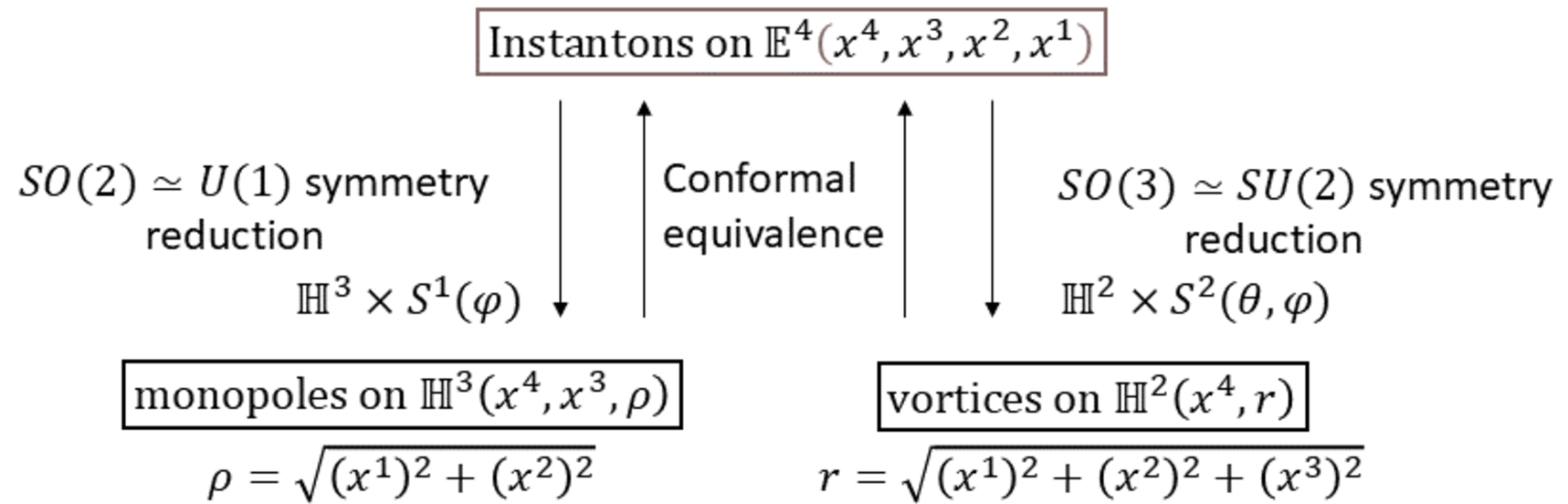}
\end{center}
\caption{
Conformal equivalence, symmetric instanton, dimensional reduction.
}
\label{conformal_unification}
\end{figure}

\begin{table}[h]
\begin{center}
\begin{tabular}{c|c|c|c|c}
\noalign{\hrule height 0.03cm}
reduction & original & conformal & symmetry & topology  \\
 & spacetime & equivalence & group & objects  \\
\hline
\hline
case (I) & $\mathbb{R}^4-\mathbb{R}^2$ &  $\mathbb{H}^3 \times S^1$ & $SO(2)$ & $magnetic$ 
\\
 & $(x^1,x^2,x^3,x^4)$ & $(\rho,x^3,x^4) \times (\varphi)$ & $~$ & $monopole$ \\
\hline
case (II) & $\mathbb{R}^4-\mathbb{R}^1$ &  $\mathbb{H}^2 \times S^2$ & $SO(3)$ & $vortex$ 
\\
 & $(x^1,x^2,x^3,x^4)$ & $(r,x^4) \times (\theta,\varphi)$ & $~$ & $~$  \\
\hline
\hline
case (III) & $\mathbb{R}^4-\mathbb{R}^0$ &  $\mathbb{H}^1 \times S^3$ & $SO(4)$ & $kink$ 
\\
 & $(x^1,x^2,x^3,x^4)$ & $(|x|) \times (\omega,\theta,\varphi)$ & $~$ & $~$  
\\
\noalign{\hrule height 0.03cm}
\end{tabular}
\end{center}
\caption{
Dimensional reductions according to conformal equivalence with compact spaces with the symmetry groups and the associated with the topological objects. 
}
\label{Table:reduction}
\end{table}

In what follows, we will construct a specific one based on the former.

\begin{definition}[Rotationally symmetric gauge field]
When a gauge field $\mathscr{A}(x)$ with a non-Abelian gauge group $G$ undergoes a spatial rotation in $\mathbb{R}^d(d>2)$, and the gauge field can be made invariant by simultaneously performing an appropriate gauge transformation, the gauge field $\mathscr{A}(x)$ is called \textbf{rotationally symmetric}\index{rotationally symmetric}. 
We say that the rotation keeps the field invariant, apart from the gauge transformation. 
When a space rotation $R$ is performed on the spatial component $\mathscr{A}_j(\bm{x})$ of the gauge field, if the corresponding gauge transformation $U_R(\bm{x})$ satisfies
\begin{align}
R_{kj}\mathscr{A}_k(R\bm{x}) =& U_R(\bm{x})\mathscr{A}_j(\bm{x})U_R^{-1}(\bm{x}) \nonumber\\
&+ iU_R(\bm{x})\partial_jU_R^{-1}(\bm{x}) ,
\label{eq:5}
\end{align}
then $\mathscr{A}_j$ is invariant under the space rotation. 
In other words, the space rotation $R$ has the same effect on the gauge field as the gauge transformation $U_R$. 
Or equivalently, if we combine $R$ and $U_R^{-1}$, the gauge field remains invariant.

If the gauge field $\mathscr{A}$ is combined with a scalar field $\Phi$ that transforms as a fundamental representation of $G$ under a gauge transformation of the group $G$, if the relation holds:
\begin{equation}
\Phi(R\bm{x}) = U_R(\bm{x})\Phi(\bm{x}) , 
\label{eq:6}
\end{equation}
$\Phi$ is invariant under rotations.

\end{definition}

\begin{prop}[Witten transformation (Witten Ansatz)]
The transformation with the $SO(3)$ spatial rotation symmetry from the $D=4$ $SU(2)$ Yang-Mills field to the dimensionally reduced $D=2$ field is given by the Witten transformation (which was originally called the Witten Ansatz):
\begin{align}
\mathscr{A}_4(x) &= \frac{\sigma_A}{2}\frac{x^A}{r}a_t(r,t) , 
\nonumber\\
\mathscr{A}_j(x) &=\frac{\sigma_A}{2}  \left\{\frac{x^A}{r}\frac{x^j}{r}a_r(r,t) 
+ \frac{\delta_j^Ar^2 - x^Ax^j}{r^3}\phi_1(r,t) 
\right. \nonumber\\
&\left. \quad\quad\quad + \epsilon_{jAk}\frac{x^k}{r^2}[1 + \phi_2(r,t)]\right\},
\nonumber\\ & r:=\sqrt{(x^1)^2+(x^2)^2+(x^3)^2} , \ (r,t) \in \mathbb{H}^2 .
\end{align}
The Witten transformation is the most general form (which is unique up to the gauge transformation) for the field that is invariant under a combination of $SO(3)$ spatial rotation (spacetime symmetry) and $SU(2)$ gauge transformation (internal symmetry).
The dimensionally reduced field has four-dimensional cylindrical symmetry, i.e., SO(3) rotational symmetry around the t-axis.

The hyperbolic vortex on $\mathbb{H}^2$ consists of a complex scalar field $\phi = \phi_1(r,x^4) + i\phi_2(r,x^4)$ and a gauge potential $a = a_r(r,x^4)dr + a_4(r,x^4)dx^4$. These hyperbolic vortex fields are functions of $x^4$ and $r$.

\end{prop}

\begin{proof}
This result was first discovered by Witten (1977)\cite{Witten} in the course for finding multi-instanton solutions of four-dimensional Yang-Mills theory. Therefore, it was called the Witten Ansatz.
This result was confirmed by Manton(1978)\cite{Manton78}.
It was proved that the Witten transformation is exactly derived from the viewpoint of the symmetric gauge field  by Forgacs and Manton (1980) \cite{FM80}.
The form is unique up to the degrees of freedom of gauge transformation.
It is also described in Manton and Sutcliffe, p.424, section 4.3 (2004) \cite{MS04}.


\end{proof}

\begin{rem} 
The argument for the magnetic monopole similar to Forgacs and Manton for the vortex was done in Manton(1978)\cite{Manton78b}. 
\end{rem} 

\begin{definition}[polar component of the Yang-Mills field]
In the cylindrical coordinate $(\rho, \varphi, x^3, x^4)$,
we define the $\rho$-component $\mathscr{A}_\rho$ and the $\varphi$-component $\mathscr{A}_\varphi$ of the Yang-Mills field $\mathscr{A}$ by
\begin{align}
& \mathscr{A}_\rho(x):=\frac{1}{\rho}(x^1\mathscr{A}_1(x) + x^2\mathscr{A}_2(x)) , \nonumber\\
& \mathscr{A}_{\varphi}(x):= -x^2\mathscr{A}_{1}(x) +x^1\mathscr{A}_{2}(x) .
\end{align}

\end{definition}

From the above scenario according to the conformal equivalence shown in Fig.~\ref{conformal_unification}, we can obtain a direct relationship between the hyperbolic magnetic monopole field on $\mathbb{H}^3$ and the hyperbolic vortex field on $\mathbb{H}^2$.

\begin{prop}[hyperbolic magnetic monopole field on $\mathbb{H}^3$ and hyperbolic vortex field on $\mathbb{H}^2$ ]\label{prop:monopole-vortex}
By applying the gauge transformation 
\begin{equation}
U_\varphi=\exp \left( i\varphi\frac{\sigma_3}{2} \right) \in SU(2)  \ \left( \varphi:=\arctan \frac{x^2}{x^1} \in [0, 2\pi) \right)
\end{equation} 
to the instanton gauge potential $\mathscr{A}_\mu(x^1,x^2,x^3,x^4)$, which corresponds to a rotation of angle $\varphi$ around the $x_3$ axis,
\begin{align}
& \mathscr{A}_\mu(x^1,x^2, x^3,x^4) &\nonumber\\
\rightarrow   \quad 
& U_\varphi \mathscr{A}_\mu(x^1,x^2, x^3,x^4) U_\varphi^\dagger + iU_\varphi\partial_\mu U_\varphi^\dagger
\nonumber\\ 
 =:& \mathscr{A}_\mu^G(\rho, x^3,x^4) .
\end{align}
We can make $\mathscr{A}_\mu(x^1,x^2,x^3,x^4)$ independent of $\varphi$, and obtain an $S^1$-symmetric instanton $\mathscr{A}_\mu^G(\rho,x^3,x^4)$ ($\rho := \sqrt{(x^1)^2+(x^2)^2}$).
In this case, $\mathscr{A}_\varphi^G(\rho,x^3,x^4)$ is identified with the hyperbolic magnetic monopole field $\Phi(\rho,x^3,x^4)$ on $\mathbb{H}^3$: 
\begin{align}
\mathscr{A}_\varphi^G(\rho,x^3,x^4) = \Phi(\rho,x^3,x^4).
\end{align}
Then the gauge field  $\mathscr{A}_\rho^G(\rho,x^3,x^4)$, $\mathscr{A}_3^G(\rho,x^3,x^4)$, $\mathscr{A}_4^G(\rho,x^3,x^4)$ and the scalar field $\Phi(\rho,x^3,x^4)$ of a hyperbolic magnetic monopole on $\mathbb{H}^3$ are written in terms of the gauge field $a_t=a_t(t,r)$, $a_r=a_r(t,r)$ and the scalar field $\phi_1=\phi_1(t,r)$, $\phi_2=\phi_2(t,r)$ of a hyperbolic vortex on $\mathbb{H}^2$: 
\begin{align}
&\mathscr{A}_4^G(\rho, x^3,x^4) =\frac{1}{2}\left\{\frac{1}{r}(\sigma_1\rho + \sigma_3x_3)\right\}a_t, 
\label{eq:G4}
\\
 &\mathscr{A}_3^G(\rho, x^3,x^4) =\frac{1}{2}\left\{\frac{x_3}{r^2}(\sigma_1\rho + \sigma_3x_3)a_r \right. 
 \nonumber\\
 &\left. \quad\quad\quad\quad\quad\quad+ \frac{\rho}{r^3}(-\sigma_1x_3 +\sigma_3\rho)\phi_1 - \frac{\rho}{r^2}\sigma_2(1 + \phi_2)\right\}, 
\label{eq:G5}
\\
 &\mathscr{A}_\rho^G(\rho, x^3,x^4) =\frac{1}{2}\left\{\frac{\rho}{r^2}(\sigma_1\rho + \sigma_3x_3)a_r \right. \nonumber\\
 &\left. \quad\quad\quad\quad\quad\quad+ \frac{x_3}{r^3}(\sigma_1x_3 -\sigma_3\rho)\phi_1 + \frac{x_3}{r^2}\sigma_2(1 + \phi_2)\right\}, 
\label{eq:G6}
\\
&\Phi(\rho, x^3,x^4) \equiv \mathscr{A}_\varphi^G(\rho, x^3,x^4) 
 \nonumber\\
 &= \frac{1}{2}\left\{\frac{\rho}{r}\sigma_2\phi_1 + \frac{\rho}{r^2}(-\sigma_1x_3 +\sigma_3\rho)(1 + \phi_2) -\sigma_3\right\} , 
\label{eq:G7}
\end{align}
where $\sigma_A (A=1,2,3)$ are the Pauli matrices and 
\begin{align}
 & r:= \sqrt{(x^1)^2+(x^2)^2+(x^3)^2}=\sqrt{\rho^2+(x^3)^2}, \nonumber\\ 
& \rho := \sqrt{(x^1)^2+(x^2)^2} .
\end{align}

\end{prop}

\begin{proof}
This result was obtained by Maldonado (2017) \cite{Maldonado17} by straightforward but a bit tedious calculations. 
Here the assignment of the components are different from his result for our  convenience. 
\end{proof}

\begin{figure}[htb]
\begin{center}
\includegraphics[scale=0.5]{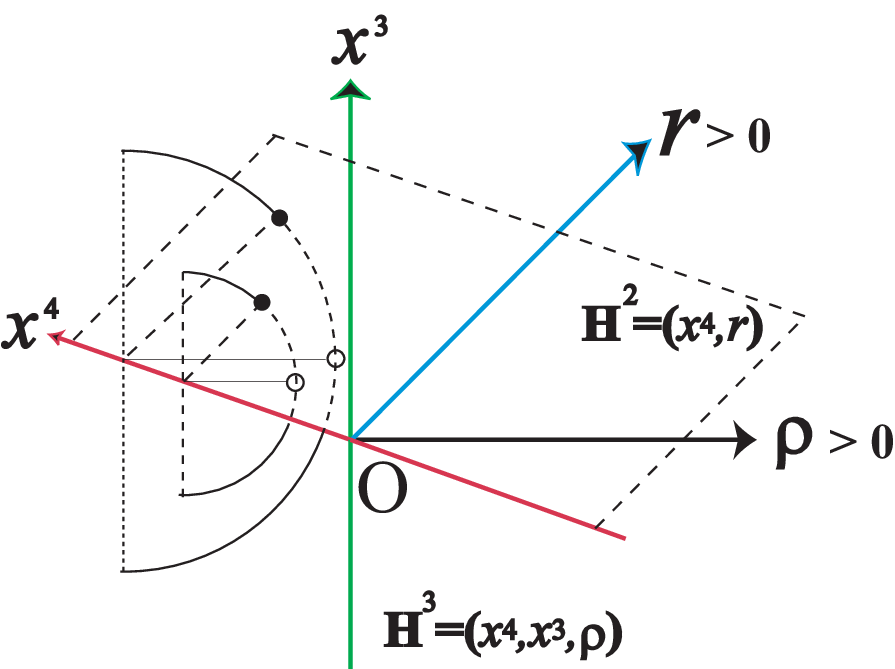}
\end{center}
\caption{
The relationship between hyperbolic vortices (black circles) on $\mathbb{H}^2$ and hyperbolic magnetic monopoles (white circles) on $\mathbb{H}^3$.
}
\label{vortex_monopole_relation}
\end{figure}

From the above formula, we can obtain a direct relationship between hyperbolic vortices and hyperbolic magnetic monopoles.

\begin{prop}[Direct relationship between hyperbolic vortices and hyperbolic magnetic monopoles]
The relationship for the norm between the $su(2)$-valued hyperbolic magnetic monopole field $\Phi(\rho, x^3,x^4)=\mathscr{A}_\varphi^G(\rho, x^3,x^4)$ and the complex-valued hyperbolic vortex field $\phi(x^4,r)=\phi_1(x^4,r)+i\phi_2(x^4,r)$ is given as 
\begin{align}
& ||\Phi(x^4,x^3,\rho)||^2 = \frac{\rho^2|\phi(x^4,r)|^2 + (x^3)^2}{4r^2} , \nonumber\\ & \ ( r:=\sqrt{\rho^2 + (x^3)^2}) .
\end{align}
$||\Phi||$ has the correct boundary value: $||\Phi|| \rightarrow v = \frac{1}{2}$ $(\rho\rightarrow0)$.

The zero point of $\Phi$ exists at the position where $x^3 = 0$ and $\phi(x^4,r)=0$. At the equatorial plane $x^3=0$, $||\Phi||$ and $|\phi|$ are proportional:
\begin{align}
  ||\Phi(x^4,x^3=0,\rho)|| = \frac12  |\phi(x^4,r=\rho)| .
\end{align}
See Fig.~\ref{vortex_monopole_relation}. 
From this, we can interpret the hyperbolic magnetic monopole as an embedded hyperbolic vortex.

\end{prop}


\begin{proof}
This result was obtained by Maldonado (2017) \cite{Maldonado17}.
By using the representation of the Pauli matrices $\sigma_1 = \left(\begin{smallmatrix}0&1\\1&0\end{smallmatrix}\right)$, $\sigma_2 = \left(\begin{smallmatrix}0&-i\\i&0\end{smallmatrix}\right)$, $\sigma_3 = \left(\begin{smallmatrix}1&0\\0&-1\end{smallmatrix}\right)$, the scalar field $\Phi$ is written as
\begin{equation}
\Phi = \frac{1}{2}
\begin{pmatrix}
-1+ \frac{\rho^2}{r^2}(1 + \phi_2)&-i\frac{\rho}{r}\phi_1 - \frac{\rho}{r^2}x_3(1 + \phi_2) \\
i\frac{\rho}{r}\phi_1 - \frac{\rho}{r^2}x_3(1 + \phi_2)&1- \frac{\rho^2}{r^2}(1 + \phi_2) .
\end{pmatrix}
\end{equation}
Then, the norm of $\Phi$ can be calculated as follows:
\begin{align}
||\Phi^2|| :=& \frac{1}{2}\operatorname{tr}(\Phi \Phi) 
=\frac{1}{4}\left\{1-\frac{\rho^2}{r^2} +\frac{\rho^2}{r^2}(\phi_1^2 +\phi_2^2)\right\} 
\nonumber\\
=&\frac{x_3^2 + \rho^2|\phi|^2}{4r^2} .
\end{align}

\end{proof}

\begin{rem} 
If the zero point of the hyperbolic vortex field is $(x_0^4,r_0)$, then $r = r_0$ defines a geodesic in the upper half-space. That is, the geodesic is a semicircle with end points $(x^4,x^3)=(x_0^4,\pm r_0)$ on the boundary.

If we define the hyperbolic magnetic monopole field as a function of the hyperbolic distance $d_H$ from the zero point of the hyperbolic magnetic monopole field measured along this geodesic line,
\begin{align}
\left. ||\Phi|\right|^2_{\phi = 0} = \frac{(x^3)^2}{4r^2} = \frac{1}{4}\tanh^2(d_H) .
\end{align}
This is consistent with the radial profile function of one hyperbolic magnetic monopole.

\end{rem}

\section{Magnetic monopole equation and vortex equation}

The Bogomolnyi equation on $\mathbb{H}^3$ implies a vortex equation on $\mathbb{H}^2$.

\begin{prop}[Magnetic monopoles on $\mathbb{H}^3$ and vortices on $\mathbb{H}^2$]
Bogomolny equation on $\mathbb{H}^3(x^4,x^3,\rho)$ obtained by the dimensional reduction from the $S^1$ symmetric Yang-Mills field on $\mathbb{E}^4$
\begin{align}
\mathscr{F}_{43} = \frac{1}{\rho}\epsilon_{43\rho}\mathscr{D}_\ell\Phi
\end{align}
implies the vortex equation on $\mathbb{H}^2(x^4,r)$:
\begin{align}
\left\{\,
\begin{aligned}
&F_{4r} := \partial_4a_r - \partial_ra_4 = \frac{1}{r^2}(1 - \phi_1^2 - \phi_2^2), \\
&\partial_4\phi_1 + a_4\phi_2 = \partial_r\phi_2 - a_r\phi_1, \\
&\partial_4\phi_2 - a_4\phi_1 = -(\partial_r\phi_1 + a_r\phi_2),
\end{aligned} 
\right.
\label{vortex_eq1}
\end{align}
which is written using the complex number as
\begin{align}
\left\{\,
\begin{aligned}
&F_{4r}-\frac{1}{r^2}(1 - |\phi|^2)=0, \\
&D_4\phi + iD_r\phi = 0.
\end{aligned}
\right.
\end{align}
Here the covariant derivative are defined by
\begin{align}
D_4:=\partial_4 - ia_4, D_r:=\partial_r - ia_r.
\end{align}

\end{prop}

\begin{proof}
\noindent
(1) 
We calculate $\mathscr{F}_{34} = \partial_3\mathscr{A}_{4}^G - \partial_4\mathscr{A}_{3}^G + i[\mathscr{A}_{3}^G, \mathscr{A}_{4}^G]$ using (\ref{eq:G4}) and (\ref{eq:G5}).
First, we calculate the derivative. Since the fields $a_4, a_r, \phi_1, \phi_2$ are all functions of $x_4$ and $r$, we have from $\rho := \sqrt{x_1^2 + x_2^2}, r := \sqrt{x_1^2+x_2^2+x_3^2} = \sqrt{\rho^2 + x_3^2}$:
\begin{align}
\partial_3 = \frac{\partial}{\partial x_3} = \frac{\partial r}{\partial x_3}\frac{\partial}{\partial r} = \frac{x_3}{r}\frac{\partial}{\partial r} = \frac{x_3}{r}\partial_r ,
\end{align}
which is used to show 
\begin{align}
\partial_3\mathscr{A}_{4}^G 
=&\left(\frac{1}{2}\frac{-x_3}{r^3}a_4 + \frac{1}{2}\frac{x_3}{r^2}\partial_ra_4\right)(\rho\sigma_1 + x_3\sigma_3) 
\nonumber\\& + \frac{1}{2}\frac{1}{r}a_4\sigma_3 ,
\nonumber\\ 
\partial_4\mathscr{A}_{3}^G =&\frac{1}{2}\frac{x_3}{r^2}\partial_4a_r(\rho\sigma_1 + x_3\sigma_3) \nonumber\\& + \frac{1}{2}\frac{\rho}{r^3}\partial_4\phi_1(-x_3\sigma_1 + \rho\sigma_3) 
\nonumber\\& -\frac{1}{2}\frac{\rho}{r^2}\partial_4\phi_2\sigma_2. 
\label{eq:11}
\end{align}

Next, we compute the commutator $[\mathscr{A}_{3}^G, \mathscr{A}_{4}^G]$.
\begin{align}
&[\mathscr{A}_{3}^G, \mathscr{A}_{4}^G] \nonumber\\
=&\frac{i}{2}\left\{\frac{\rho}{r^2}a_4\phi_1\sigma_2+\frac{\rho}{r^3}a_4(1 + \phi_2)(-x_3\sigma_1 + \rho\sigma_3)\right\}.\label{eq:12}
\end{align}
From \eqref{eq:11} and \eqref{eq:12},
we have
\begin{align}
\mathscr{F}_{34} =&\frac{1}{2}\left\{\frac{x_3}{r^2}(\partial_r a_4 -\partial_4 a_r)(\rho\sigma_1 + x_3\sigma_3)\right. \nonumber\\
&\left.+\frac{\rho}{r^3}(\partial_4 \phi_1 +a_4\phi_2)(x_3\sigma_1 - \rho\sigma_3) \right. \nonumber\\
&\left.+\frac{\rho}{r^2}(\partial_4 \phi_2 -a_4\phi_1)\sigma_2\right\}.
\label{F1}
\end{align}

\noindent
(2) Furthermore, we calculate the covariant derivative $\mathscr{D}_\rho\Phi = \partial_\rho\Phi + i[\mathscr{A}_\rho^G, \Phi]$ using (\ref{eq:G6}) and (\ref{eq:G7}).

First, we calculate the derivative $\partial_\rho \Phi$ using 
\begin{align}
\partial_\rho \equiv \frac{\partial}{\partial \rho} = \frac{\partial r}{\partial \rho}\frac{\partial}{\partial r} = \frac{\rho}{r}\frac{\partial}{\partial r} = \frac{\rho}{r}\partial_r .
\end{align}
as
\begin{align}
&\partial_\rho \Phi \nonumber\\
=& \frac{1}{2}\left\{\left(\frac{1}{r}\phi_1 + \rho\frac{\rho}{r}\frac{-1}{r^2}\phi_1 + \frac{\rho}{r}\frac{\rho}{r}\partial_r\phi_1\right)\sigma_2\right. \nonumber\\
&\left. +\left[\frac{1}{r^2}(1 + \phi_2) + \rho \frac{\rho}{r}\frac{-2}{r^3}(1 + \phi_2) + \frac{\rho}{r^2}\frac{\rho}{r}\partial_r\phi_2\right]
\right. \nonumber\\
&\left.  \quad \times (-x_3\sigma_1 + \rho\sigma_3) \right. 
\nonumber\\
&\left. +\frac{\rho}{r^2}(1 + \phi_2)\sigma_3 \right\} . 
\label{eq:13}
\end{align}
Next, we compute the commutator $[\mathscr{A}_{\rho}^G, \Phi]$.
\begin{align}
&[\mathscr{A}_{\rho}^G, \Phi] \nonumber\\
=&\frac{2}{4}\left\{\frac{\rho^2}{r^3}a_r\phi_1(-x_3\sigma_1 + \rho\sigma_3)- \frac{\rho^2}{r^4}a_r(1 + \phi_2)r^2\sigma_2\right. \nonumber\\
&\left. +\frac{\rho}{r^2}a_r\rho\sigma_2 +\frac{x_3\rho}{r^4}\phi_1^2(\rho\sigma_1 + x_3\sigma_3) 
-\frac{x_3}{r^3}\phi_1x_3\sigma_2\right. \nonumber\\
&\left. 
+ \frac{\rho x_3}{r^4}(1 + \phi_2)^2(\rho\sigma_1 + x_3\sigma_3) - \frac{x_3}{r^2}(1 + \phi_2)\sigma_1\right\}. \label{eq:14}
\end{align}
From \eqref{eq:13} and \eqref{eq:14}
we have
\begin{align}
\mathscr{D}_\rho\Phi =& \frac{1}{2}\left\{\frac{\rho^2}{r^3}(\partial_r\phi_2 - a_r\phi_1)(-x_3\sigma_1 + \rho\sigma_3)
\right. \nonumber\\
&\left. +\frac{\rho^2}{r^2}(\partial_r\phi_1 + a_r\phi_2)\sigma_2
\right. \nonumber\\
&\left.+\frac{\rho x_3}{r^4}(1 - \phi_1^2 - \phi_2^2)(\rho\sigma_1 + x_3\sigma_3)
\right\}.
\label{DF1}\end{align}

\noindent
(3) If the results of (\ref{F1}) and (\ref{DF1}) are substituted into 
\begin{align}
\mathscr{F}_{43} = \frac{1}{\rho}\epsilon_{43\rho}\mathscr{D}_\rho\Phi \  (\epsilon_{43\rho} = -1),
\end{align}
we have (\ref{vortex_eq1}).
The two equations can be combined into one using complex numbers:
\begin{align}
0 =& (D_4 + iD_r)\phi \nonumber\\
=&(\partial_4 - ia_4)\phi + i(\partial_r - ia_r)\phi \nonumber\\
=&(\partial_4 - ia_4)(\phi_1 + i\phi_2) + i(\partial_r - ia_r)(\phi_1 + i\phi_2) \nonumber\\
=&(\partial_4 \phi_1 + a_4\phi_2) + i(\partial_4\phi_2 - a_4\phi_1) \nonumber\\
&+(-\partial_r\phi_2 + a_r\phi_1) + i(\partial_r\phi_1 + a_r\phi_2) \nonumber\\
\Leftrightarrow&
\left\{\,
\begin{aligned}
&\partial_4\phi_1 + a_4\phi_2 = \partial_r\phi_2 - a_r\phi_1\\
&\partial_4\phi_2 - a_4\phi_1 = -(\partial_r\phi_1 + a_r\phi_2).
\end{aligned}
\right.
\end{align}

\end{proof}

\section{Hyperbolic Space}

We can use different coordinates to express the hyperbolic space depending on the purpose.

\begin{definition}[Hyperbolic space: upper half space model, ball model]
The hyperbolic space $\mathbb{H}^3$ is expressed by the following metric using coordinates $(x^3, x^4, \rho)(\rho>0)$ which is called the \textbf{upper half space model}\index{upper half space model}:
\begin{align}
(ds)^2(\mathbb{H}^3) = \frac{(dx^3)^2 + (dx^4)^2 +(d\rho)^2}{\rho^2}.
\end{align}
Alternatively, the unit \textbf{ball model}\index{ball model} of the hyperbolic space $\mathbb{H}^3$ uses the metric of $\mathbb{H}^3$ given by 
\begin{align}
(ds)^2(\mathbb{H}^3) =& 4\frac{(dX^1)^2 + (dX^2)^2 +(dX^3)^2}{(1 - R^2)^2} \ , \nonumber\\
R:=&\sqrt{(X^1)^2 + (X^2)^2 +(X^3)^2} < 1 .
\end{align}
which is modified to a more general case
\begin{align}
(ds)^2(\mathbb{H}^3) = 4R_0^{-2} \frac{(dX^1)^2 + (dX^2)^2 + (dX^3)^2}{(1-R^2/R_0^2)^2} , \ R<R_0 .
\end{align}
In this coordinate system, the metric is rotational.
The coordinates $(X^1, X^2, X^3)$ of the ball model in the hyperbolic space $\mathbb{H}^3$ are given from the coordinates $(x^4, x^3, \rho)$ of the upper half-space model:
\begin{align}
& X^1 + iX^2 = 2\frac{x^4 + ix^3}{(x^3)^2 + (x^4)^2 +(\rho + 1)^2} , \nonumber\\
& X^3 = \frac{(x^3)^2 + (x^4)^2 + (\rho^2 - 1)}{(x^3)^2 + (x^4)^2 +(\rho + 1)^2}, \\
& R := \sqrt{(X^1)^2 + (X^2)^2 + (X^3)^2} \nonumber\\
&= \sqrt{ \frac{(x^3)^2 +(x^4)^2+ (\rho - 1)^2}{(x^3)^2 + (x^4)^2 +(\rho + 1)^2} } \to 1 \ (\rho \to 0). 
\end{align}
The inverse transformation is given by
\begin{equation}
x^4 + ix^3 = 2 \frac{X^1 + iX^2}{1 + R^2 - 2X^3} , \ \rho = \frac{1 - R^2}{1 + R^2 -2X^3}
\end{equation}

\end{definition}

\begin{definition}[Another model of hyperbolic space]
As another model, if the geodesic distance $\xi$ from the origin is introduced:
\begin{align}
R = \tanh \frac{\xi}{2},
\end{align}
then using the 3D polar coordinates $(R,\tilde{\theta},\tilde{\varphi})$ or $(\xi,\tilde{\theta},\tilde{\varphi})$:$X^1 = R\sin\tilde{\theta}\cos\tilde{\varphi}$, $X^2 = R\sin\tilde{\theta}\sin\tilde{\varphi}$, $X^3= R\cos\tilde{\theta}$, the metric reads (Manton and Sutcliffe(2014)\cite{MS14})
\begin{align}
(ds)^2(\mathbb{H}^3) &= 4\frac{(dR)^2 + R^2[(d\tilde\theta)^2 +\sin^2\theta(d\tilde\varphi)^2]}{(1 - R^2)^2} \nonumber\\
&=(d\xi)^2 + \sinh^2\xi((d\tilde\theta)^2 +\sin^2\tilde\theta(d\tilde\varphi)^2).
\end{align}

\end{definition}

\begin{rem}
The metric is written using the complex number $z$ as 
\begin{align}
(ds)^2(\mathbb{H}^3) &= \frac{4}{\kappa^2(1 - R^2)^2}\left((dR)^2 + R^2\frac{4dzdz^*}{(1 + |z|^2)^2}\right) \nonumber\\
&=(d\ell)^2 + \frac{\sinh^2(\kappa \ell)}{\kappa^2}\frac{4dzdz^*}{(1+|z|^2)^2}
\end{align}
Here $\ell$ is the hyperbolic distance defined by
\begin{align}
R = \tanh \left(\frac{\kappa}{2}\ell \right)
\end{align}
In the zero curvature limit $\kappa \to 0$, the metric becomes that of flat Euclidean space $\mathbb{R}^3$.

\end{rem}

\begin{rem}
Introducing toroidal coordinates $(\xi, \tilde{\theta}, \tilde{\varphi}, \chi)$ on $\mathbb{R}^4$:
\begin{equation}
x^\mu = \frac{(\sinh\xi\sin\tilde{\theta}\cos\tilde{\varphi}, \sinh\xi\sin\tilde{\theta}\sin\tilde{\varphi}, \sinh\xi\cos\tilde{\theta}, \sin\chi)}{\cosh\xi + \cos\chi} .
\end{equation}
The flat metric reads
\begin{align}
(ds)^2(\mathbb{E}^4) &= \frac{(ds)^2(\mathbb{H}^3) + (d\chi)^2}{(\cosh\xi + \cos\chi)^2} \ , \nonumber\\ (ds)^2(\mathbb{H}^3) &= (d\xi)^2 + \sinh^2\xi((d\tilde{\theta})^2 + \sin^2\theta(d\tilde{\varphi})^2),
\end{align}
which is conformal equivalent to
\begin{align}
(ds)^2(\mathbb{E}^4) \cong (ds)^2(\mathbb{H}^3)+(d\chi)^2
\end{align}
And the $SO(2)$ orbit is a two-dimensional sphere $(\tilde{\theta}, \tilde{\varphi})$ with constant $\xi$.
Since the quotient of $\mathbb{R}^4$ by the circular action is conformally equivalent to $\mathbb{H}^3$, the $S^1$-invariant gauge field of $\mathbb{R}^4$ gives rise to the gauge field on $\mathbb{H}^3$ and the adjoint Higgs field (the component of the gauge field along the circle) according to the standard idea of the dimensional reduction. 
The self-dual Yang-Mills equation reduces to the Bogomolny equation on $\mathbb{H}^3$.

\end{rem}

\section{Hyperbolic magnetic monopoles solutions}

The Lagrangian density and the action for the theory obtained by the dimensional reduction describing the hyperbolic magnetic monopole on  $\mathbb{H}^3$ is given by 
\begin{align}
S_3 =& \int_{\mathbb{H}^3}d^3x \sqrt{g} \mathscr{L}_3 \ (d^3x := dx^3dx^4d\rho) , \nonumber\\
\mathscr{L}_3 =& \frac{1}{2}g^{\mu\nu}g^{\nu\beta}\operatorname{tr}(\mathscr{F}_{\mu\nu}\mathscr{F}_{\alpha\beta}) + g^{\mu\nu}\operatorname{tr}(\mathscr{D}_\mu\Phi\mathscr{D}_\nu\Phi) ,
\end{align}
where the metric $g_{\mu\nu}$, its inverse $g^{\mu\nu}$ and determinant $ g$ are given by
\begin{align}
g_{\mu\nu} = \frac{1}{\rho^2}\delta_{\mu\nu} , \ g^{\mu\nu} = \rho^2\delta^{\mu\nu} , \ g:=\det(g_{\mu\nu}) = \frac{1}{\rho^6} .
\end{align}
Therefore, the Lagrangian density reads 
\begin{align}
\mathscr{L}_3 =& \rho^3 \left[\rho\frac{1}{2}\operatorname{tr}(\mathscr{F}_{\mu\nu} \mathscr{F}_{\mu\nu} ) + \frac{1}{\rho}\operatorname{tr}\{(\mathscr{D}_\mu\Phi) (\mathscr{D}_\mu\Phi) \}\right] 
 .
\end{align}

For the reviews of magnetic monopoles in gauge theories, see e.g., Goddard, Nuyts and Olive(1977)\cite{GNO77}, Goddard and Olive(1978)\cite{GO78}.
For a recent view on the hyperbolic magnetic monopole, see Lang(2025)\cite{Lang25} and also Lang(2024)\cite{Lang24}.

\subsection{The Bogomol'nyi equation and the hyperbolic magnetic monopole}

The Bogomolnyi equation on $\mathbb{H}^3$ implies the vortex equation on $\mathbb{H}^2$.

\begin{definition}[Bogomol'nyi equation]

The hyperbolic magnetic monopole is a solution of \textbf{Bogomol'nyi equation}\index{Bogomol'nyi equation}(Bogomol'nyi(1976)\cite{Bogomolnyi}):
\begin{align}
*_{\mathbb{H}^3} \mathscr{F}=\mathscr{D}[\mathscr{A}]\Phi  \Leftrightarrow \mathscr{F}=*_{\mathbb{H}^3} \mathscr{D}[\mathscr{A}]\Phi .
\end{align}
Here, $\mathscr{F}$ is the field strength of the $su(2)$-valued gauge field $\mathscr{A}$, $\mathscr{D}\Phi$ is the covariant derivative of the $su(2)$-valued adjoint scalar field $\Phi$, and $*$ is the Hodge duality defined using the metric of hyperbolic space $\mathbb{H}^3$. 

As a boundary condition, let us assume that the norm (magnitude) of the $su(2)$-valued scalar field $\Phi=\Phi^AT_A$ defined by
\begin{align}
||\Phi||:=\sqrt{\frac{1}{2}\operatorname{Tr}(\Phi^2)} = \frac12 \sqrt{ \Phi^A \Phi^A}
\end{align}
takes a constant positive value $v$ on the boundary $\partial \mathbb{H}^3$, i.e, at infinity $\infty$:
\begin{align}
||\Phi||_\infty = v \ (>0) .
\end{align}

\end{definition}

\begin{prop}[Solution of the Bogomol'nyi equation]
Using the local coordinates $(X_1,X_2,X_3)$ in the ball model $(R:=\sqrt{X_1^2+X_2^2+X_3^2})$, the Bogomol'nyi equation $\mathscr{F} = *d_A\Phi$ on $\mathbb{H}^3$ can be written :
\begin{align}
\mathscr{D}_\ell[\mathscr{A}]\Phi^A = \frac{1}{2\sqrt{f}}\epsilon_{jk\ell}\mathscr{F}_{jk}^A , \ f :=\left (1 - \frac{R^2}{4}\right)^{-2} ,  
\end{align}
where $R \in [0, R_0=2)$. 
We adopt the Ansatz for the solution  with unknown functions $P$ and $Q$:
\begin{align}
\mathscr{A}_{j}^A = \epsilon_{jAk}\frac{X^k}{R^2}[P(R) - 1] , \ \mathscr{A}_{\varphi}^A = \Phi^A = \frac{X^A}{R}Q(R) ,
\end{align}
where $P$ and $Q$ depend only on the radial coordinate $R$ of the ball model.
Then the Bogomol'nyi equation is reduced to a pair of first order differential equations:
\begin{align}
& \frac{dP(R)}{dR} = \sqrt{f}P(R)Q(R) , \nonumber\\
& \sqrt{f}R^2\frac{dQ(R)}{dR} = P(R)^2 - 1 .
\end{align}
The regular solution at the origin [and infinity (the boundary)] is given by
\begin{align}
&P = \frac{C\sinh \xi}{\sinh(C\xi )} , \quad Q  =  \coth \xi - C \coth(C\xi) , \ \nonumber\\
&\xi := 2\tanh^{-1}\frac{R}{2} , \ C = 2v + 1=2||\Phi||_\infty+1 ,
\end{align}
where $C$ is determined by the boundary condition $C = 2v + 1$ for $v:=||\Phi||_\infty$.
Therefore, the norm $||\Phi||$ of the scalar field $\Phi$ can be written with its asymptotic value $v$ with $v = \frac{1}{2}(C - 1)$:
\begin{align}
 ||\Phi|| =\frac12 |Q| =\frac{1}{2}\left[C \coth (C\xi) - \coth \xi \right] > 0 .
\end{align}
In the vicinity of the origin $\xi = 0$, the terms with the poles cancel, and $||\Phi||$ is proportional to $\xi$.

\end{prop}

\begin{proof}
This result was obtained by Nash (1986)\cite{Nash86}. 

\noindent
(1) 
The derivative reads
\begin{align}
\partial_j\mathscr{A}_k^A 
=&\epsilon_{kAj}\frac{P - 1}{R^2} - 2\epsilon_{kA\ell}\frac{X^\ell X^j}{R^4}(P - 1) 
\nonumber\\
&+ \epsilon_{kA\ell}\frac{X^\ell X^j}{R^3}P' ,
\end{align}
where we have used
$
\partial_j :=\frac{\partial}{\partial X^j}, 
$
and
$
\partial_j R =\frac{X^j}{R} .
$
Therefore, we have
\begin{align}
&\partial_j \mathscr{A}_k^A-\partial_k \mathscr{A}_j^A \nonumber\\
= &2\epsilon_{jkA}\frac{P - 1}{R^2} -2(\epsilon_{kA\ell}X^j - \epsilon_{jA\ell}X^k)\frac{X^\ell}{R^4}(P - 1) \nonumber\\
&+(\epsilon_{kA\ell}X^j - \epsilon_{jA\ell}X^k)\frac{X^\ell}{R^3}P' ,
\end{align}
which yields
\begin{align}
& \epsilon^{jkn}(\partial_j\mathscr{A}_k^A - \partial_k\mathscr{A}_j^A)  \nonumber\\
=&4\delta^n_A\frac{P - 1}{R^2} + 2(\delta^{nA}R^2 - X^AX^n)\left(\frac{P'}{R^3} - 2\frac{P - 1}{R^4}\right),
\end{align}
where we have used
$
\epsilon^{ABC}\epsilon_{ADE} = \delta^{BD}\delta^{CE} - \delta^{BE}\delta^{CD}, 
$
$
\epsilon^{ABC}\epsilon_{ABE} 
= 2\delta^{CE}.
$
Next, we have
\begin{align}
\epsilon^{jkn}\epsilon^{ABC}\mathscr{A}_j^B\mathscr{A}_k^C 
= 2X^nX^A\frac{(P - 1)^2}{R^4}.
\end{align}
Thus we obtain
\begin{align}
&\epsilon^{jkn}\mathscr{F}_{jk}^A 
= \epsilon^{jkn}(\partial_j\mathscr{A}_k^A -\partial_k\mathscr{A}_j^A + g\epsilon^{ABC}\mathscr{A}_j^B\mathscr{A}_k^C) \nonumber\\
=&2\delta^{nA}\frac{P'}{R} + 2X^nX^A\left(2\frac{P - 1}{R^4} - \frac{P'}{R^3} + \frac{(P - 1)^2}{R^4}\right). \label{eq:1}
\end{align}
On the other hand,
for $\mathscr{D}_n\Phi^A = \partial_n\Phi^A + g\epsilon^{ABC}\mathscr{A}_n^B\Phi^C$, 
\begin{align}
\partial_n\Phi^A =& \partial_n\left(\frac{X^A}{R}Q\right) \nonumber\\
=& \frac{\delta^{nA}}{R}Q + X^A\frac{-1}{R^2}\frac{X^n}{R}Q + \frac{X^A}{R}\frac{X^n}{R}Q' , 
\\
\epsilon^{ABC}\mathscr{A}_n^B\Phi^C 
=&(\delta^{An}R^2 - X^AX^n)\frac{(P - 1)Q}{R^3} .
\end{align}
Thus we obtain
\begin{align}
\mathscr{D}_n\Phi^A =& \delta^{nA}\left(\frac{P-1}{R^1} + \frac{1}{R}\right)Q 
\nonumber\\ &+ X^nX^A\left[-\frac{Q}{R^3} + \frac{Q'}{R^2} - \frac{(P - 1)Q}{R^3}\right]. \label{eq:2}
\end{align}
Finally, for the equation
\begin{align}
\mathscr{D}_n\Phi^A = \frac{1}{2\sqrt{f}}\epsilon^{jkn}\mathscr{F}_{jk}^A \Leftrightarrow \frac{1}{2}\epsilon^{jkn}\mathscr{F}_{jk}^A = \sqrt{f}\mathscr{D}_n\Phi^A
\end{align}
to hold, 
the coefficient functions for the tensor $\delta^{nA}$ and $X^nX^A$ must agree on both sides of the equation.  
The coefficients of $\delta^{An}$ lead to 
\begin{align}
\frac{P'}{R} = \sqrt{f}\left(\frac{P - 1}{R} + \frac{1}{R}\right)Q \Rightarrow 
P' = \sqrt{f}PQ .
\label{eq:3}
\end{align}
The coefficients of $X^AX^n$ lead to
\begin{align}
&2\frac{P - 1}{R^4} - \frac{P'}{R^3} + \frac{(P-1)^2}{R^4} \nonumber\\
=& \sqrt{f}\left(-\frac{Q}{R^3} + \frac{Q'}{R^2} -\frac{(P-1)Q}{R^3}\right) \nonumber\\
\Rightarrow&  \quad \frac{P^2 - 1}{R^4} - \frac{P'}{R^3} = \sqrt{f}\left(\frac{Q'}{R^2} - \frac{PQ}{R^3}\right) \nonumber\\
\Rightarrow& \quad 
P^2 - 1 = \sqrt{f}R^2Q' .
\label{eq:4}
\end{align}

\noindent
(2) 
Furthermore, by introducing new variables $\xi$ and $T$:
\begin{align}
& T := \ln P(R), \ \xi := 2\tanh^{-1}\frac{R}{2} \nonumber\\
\Leftrightarrow &  \quad P(R)=e^{T}, \ R=2\tanh \frac{\xi}{2},
\end{align}
the equations are rewritten into 
\begin{align}
\frac{dT}{d\xi}=Q, \ \frac{d^2T}{d\xi^2} = \sinh^{-2}\xi(e^{2T} - 1) .
\end{align}
In fact, they are confirmed as follows.
\begin{align}
\sqrt{f}:=\frac{1}{1-\frac{R^2}{4}}=\frac{1}{1 - \tanh^2\frac{\xi}{2}} = \cosh^2\frac{\xi}{2} .
\end{align}
This can be used to rewrite the differentiation:
\begin{align}
 \frac{d}{dR} =& \frac{d\xi}{dR}\frac{d}{d\xi} = \frac{1}{1 - \frac{R^2}{4}}\frac{d}{d\xi} = \frac{1}{1 - \tanh^2\frac{\xi}{2}}\frac{d}{d\xi} \nonumber\\
=& \cosh^2\frac{\xi}{2}\frac{d}{d\xi} 
= \sqrt{f}\frac{d}{d\xi} .
\end{align}
where we have used
\begin{align}
\cosh^2\frac{\xi}{2} - \sinh^2\frac{\xi}{2} = 1 \Rightarrow 1 - \tanh^2\frac{\xi}{2} = \frac{1}{\cosh^2\frac{\xi}{2}} .
\end{align}
If we use the variable $\xi$ instead of $R$, then the equations are rewritten as
\begin{align}
 & \frac{dP}{dR} = \sqrt{f}PQ \nonumber\\
 \Rightarrow  & \quad Q = \frac{1}{\sqrt{f}}\frac{1}{P}\frac{dP}{dR} = \frac{1}{\sqrt{f}}\frac{d\ln P}{dR} =\frac{d}{d\xi}\ln P, 
\\
& \sqrt{f}R^2\frac{dQ}{dR} = P^2 - 1  
\Rightarrow fR^2\frac{dQ}{d\xi} = P^2 - 1 \nonumber\\
\Rightarrow&\quad \frac{dQ}{d\xi} = \frac{P^2-1}{fR^2} = \sinh^{-2}\xi(P^2 - 1) ,
\end{align}
where we have used
\begin{align}
fR^2 = \cosh^4\frac{\xi}{2}\cdot 4\tanh^2\frac{\xi}{2} = 4\cosh^2\frac{\xi}{2}\sinh^2\frac{\xi}{2} = \sinh^2 \xi .
\end{align}

\noindent
(3) 
The general solution with two constants $C$ and $\alpha$ to this second order differential equation is given by:
\begin{align}
P = \frac{C\sinh \xi}{\sinh(C\xi + \alpha)} .
\end{align}
In fact, we can check that it is indeed the solution:
\begin{align}
Q &= \frac{dT}{d\xi} = \frac{d}{d\xi}\ln P = \frac{d}{d\xi}[\ln \sinh \xi - \ln \sinh(C\xi + \alpha)] \nonumber\\
&= \frac{\cosh \xi}{\sinh \xi} - \frac{C \cosh(C\xi + \alpha)}{\sinh(C\xi + \alpha)} \nonumber\\
&=  \coth \xi - C \coth(C\xi + \alpha).
\end{align}
By differentiating this function, we can verify that, using $(\coth x)^\prime=-1/\sinh^2 x$, $Q$ is indeed a solution:
\begin{align}
\frac{dQ}{d\xi}=\frac{d^2T}{d\xi^2}
=&  \frac{C^2}{\sinh^2(C\xi+ \alpha)} - \frac{1}{\sinh^2\xi} \nonumber\\
 =& \sinh^{-2}\xi(P^2 - 1)
\end{align}
To obtain a regular solution at $R=0$, we need $P \to 1$ as $R \to 0$, therefore we set $\alpha = 0$.
It is easy to see that the boundary condition $R \to 1$ as $\rho \to 1$ is also satisfied by this solution.

\end{proof}

\begin{example}[Manton and Sutcliffe (2014)\cite{MS14}]
A hyperbolic 1-monopole can be obtained from an $S^1$ symmetric instanton if and only if $v =||\Phi||_\infty= \frac{1}{2}$ is a half integer.
In this case, $P$ and $Q$ are rational functions of $R$.

When $C = 2v +1$ is an integer, $||\Phi||$ is a rational function of the radial coordinate $R$ of the rational metric of the ball model:
\begin{align}
||\Phi|| &= \frac12 \left[ C\frac{e^{C\xi} +e^{-C\xi}}{e^{C\xi} -e^{-C\xi}} - \frac{e^{\xi} +e^{-\xi}}{e^{\xi} -e^{-\xi}} \right] \nonumber\\
&= \frac12 \left[ C\frac{(e^{\xi})^C +(e^{-\xi})^C}{(e^{\xi})^C -(e^{-\xi})^C} - \frac{e^{\xi} +e^{-\xi}}{e^{\xi} -e^{-\xi}} \right] .
\end{align}
In fact, using $e^\xi = \frac{1 + R}{1 - R}$ and $e^{-\xi} = \frac{1 - R}{1 + R}$ following from $R=\tanh\frac{\xi}{2}$, we can write $||\Phi||$ as a rational function:
\begin{align}
& C =2 , \ v = \frac{1}{2} \ (I =N): \nonumber\\
&\quad\quad\quad ||\Phi|| = \frac{R}{1 + R^2}  \to \frac12 \ ( R \to 1) , \nonumber\\
& C =3 , \  v = 1 \ (I =2N): \nonumber\\
&\quad\quad\quad ||\Phi|| = \frac{8R(1 + R^2)}{(3 + R^2)(1 + 3R^2)}  \to 1 \ ( R \to 1) , \nonumber\\
& C =4 ,\ v = \frac{3}{2} \ (I =3N) : \nonumber\\
&\quad\quad ||\Phi|| = \frac{R(5 + 14R^2 + 5R^4)}{(1 + R^2)(1 + 6R^2 + R^4)} \to \frac32 \ ( R \to 1)  .
\end{align}
We can confirm the linear behavior $|\Phi|\propto R$ of $||\Phi||$ near $R = 0$ and the boundary value $||\Phi||_\infty = v$ at $R = 1$ ($\xi=\infty$).

\end{example} 

\begin{figure}[htb]
\begin{center}
\includegraphics[scale=0.6]{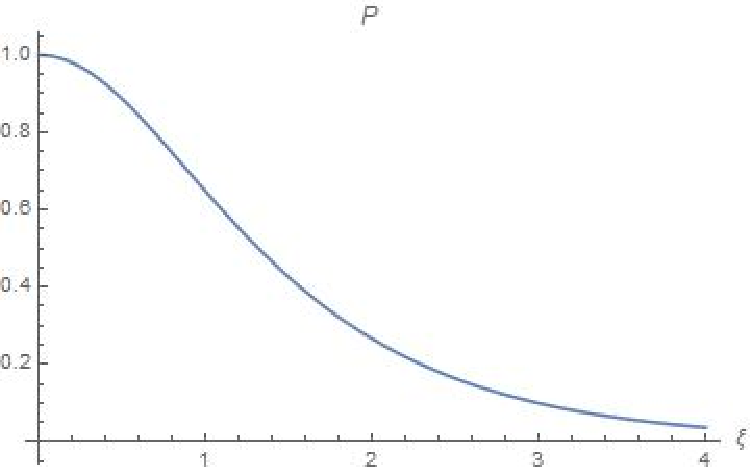}
\quad\quad
\includegraphics[scale=0.6]{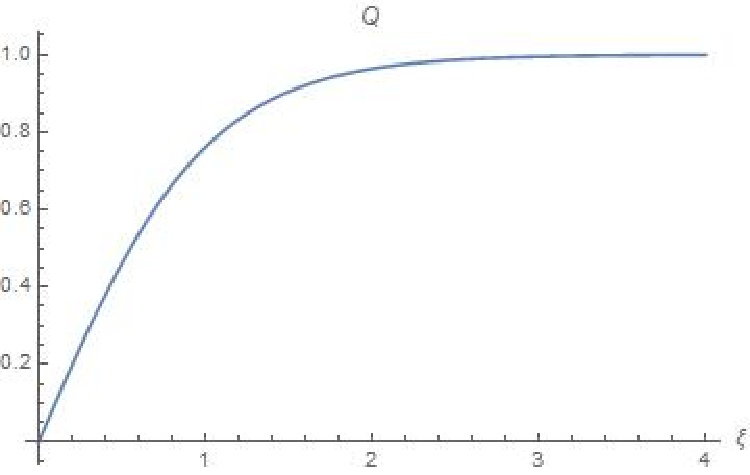}
\end{center}
\caption{
The profile functions $P$ and $Q$ of a hyperbolic magnetic monopole as a function of $\xi$ for $v=\frac12$ and $C=2$.
}
\label{H3_monopole_P_Q}
\end{figure}

\begin{figure}[htb]
\begin{center}
\includegraphics[scale=0.7]{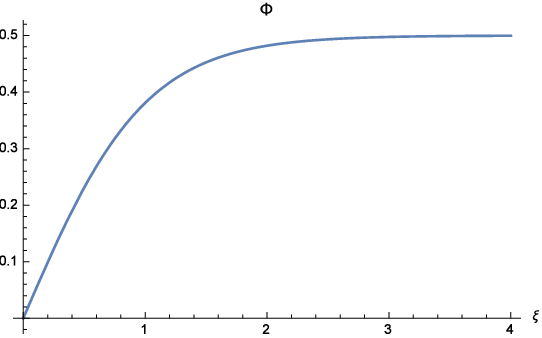}
\end{center}
\caption{
The norm of the scalar field $||\Phi||$ of a hyperbolic magnetic monopole as a function of $\xi$ for $v=\frac12$ and $C=2$.
}
\label{H3_monopole}
\end{figure}

\begin{rem}
The Bogomolny equation for a magnetic monopole in flat space has the same form as (1), but the Hodge dual is that of the 3-dimensional Euclidean space $\mathbb{R}^3$. In flat space, the boundary value $v = ||\Phi||_\infty$ introduces a length (inverse) scale, and replacing $v$ by $\tilde{v}(>v)$ gives a magnetic monopole scaled down by a factor $\tilde{v}/v$. On the other hand, hyperbolic space has a built-in length scale, and the value of $v$ affects the magnetic monopole solution in a nontrivial way.

Hyperbolic magnetic monopoles exist for any positive value of $v(>0)$, but they can be interpreted as instantons of the $SU(2)$ Yang-Mills theory in $\mathbb{R}^4$ that are invariant under the $S^1$ action only when $2v \equiv 2\rho$ is an integer. Recall that instantons are solutions of the conformally invariant self-dual Yang-Mills equation $*\mathscr{F} = \mathscr{F}$ in $\mathbb{R}^4$. 
\end{rem}

\begin{rem}
There are two length scales for monopoles in hyperbolic space: the scale of the curvature of hyperbolic space ($1/\kappa$) and the core size of a monopole ($1/v$). The ratio $v/\kappa$ of these length scales is important. As first pointed out by Atiyah, when $2v/\kappa\in\mathbb{Z}$, a hyperbolic magnetic monopole of charge $N$ is equivalent to an $S^1$-invariant self-dual Yang-Mills instanton on $\mathbb{R}^4$ (with instanton number $2vN/\kappa$). As a result, the study of hyperbolic magnetic monopoles is simplified when the asymptotic value $v$ of the length of the scalar field $\Phi$  is discrete relative to the curvature of hyperbolic space. Since only the ratio matters, we can choose to fix either $\kappa$ or $v$ without loss of generality. In what follows, we choose to fix $\kappa$ to $\kappa = 1$. The flat space limit is then equivalent to the $v \to \infty$ limit.

\end{rem}

\begin{prop}[Euclidean magnetic monopoles as the zero curvature limit of the the hyperbolic magnetic monopole]
In the zero curvature limit $\kappa \to 0$ of the hyperbolic space $\mathbb{H}^3$, the hyperbolic magnetic monopole converges to the corresponding Euclidean magnetic monopole.
\end{prop}

\begin{proof}
This was first predicted by Atiyah (1984)\cite{Atiyah84,Atiyah84b} and later rigorously proven by Jarvis and Norbury (1997)\cite{JN97}.
The following presentation is based on Nash (1986)\cite{Nash86}.

If $\mathbb{H}^3(a)$ is a hyperbolic space with scalar curvature $-\frac{6}{a^2}$ and $S^1(a)$ is a circle with radius $a$, the Bogomolnyi equation on $\mathbb{H}^3(a)\times S^1(a)$ is given by
\begin{align}
\mathscr{D}_\ell\Phi^A = \frac{a}{2\sqrt{f_a}}\epsilon_{jk\ell}\mathscr{F}_{jk}^A \ , \ f_a = \left(1 - \frac{R^2}{4a^2}\right)^{-2} .
\end{align}
This equation can be solved exactly using a similar Ansatz. The Higgs field $\Phi$ is written using $\xi = 2\tanh^{-1}\left(\frac{R}{2a}\right)$
\begin{align}
\Phi^A = \frac{1}{a}\frac{dT}{d\xi}\frac{X^A}{R} = \frac{1}{a}\left[C \coth(C\xi)-\coth \xi \right]\frac{X^A}{R} .
\end{align}
If we fix $R$ and consider the $a \to \infty$ limit, $R = 2a\tanh(\frac{\xi}{2}) \to 0$ $(a \to \infty)$, so if we replace $\xi$ with $\xi_a$, we get $\frac{R}{2a} = \frac{s_a}{2}$, 
\begin{align}
& \lim_{a \to \infty}\Phi^A \nonumber\\
=& \lim_{a \to \infty}\frac{1}{a}\left[ (a + 1)\coth((1+a)\frac{R}{2a}) -\coth (\frac{R}{2a}) \right]\frac{X^A}{R} \nonumber\\
=&\left( \frac{1}{\tanh R} -\frac{1}{R} \right)\frac{X^A}{R} .
\end{align}
This corresponds to the Higgs field of the Prasad-Sommerfield monopole (Prasad-Sommerfield(1975)\cite{PS75}). Similar rigorous results hold for $\mathscr{A}_j^A$. The $a \to \infty$ limit of such a hyperbolic magnetic monopole of $\mathbb{H}^3(a)$ certainly reproduces the usual Euclidean monopole.

\end{proof}

\begin{figure}[htb]
\begin{center}
\includegraphics[scale=0.7]{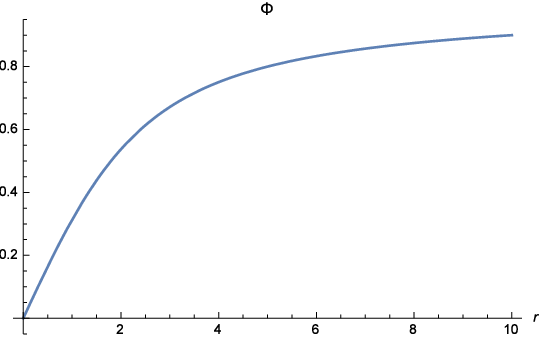}
\end{center}
\vskip -0.5cm
\caption{
The norm of the scalar field $||\Phi||$ of a Prasad-Sommerfield magnetic monopole as a function of $r$.
}
\label{PS_monopole}
\end{figure}

\begin{rem}
Let's look at the difference between the hyperbolic magnetic monopole and the normal Euclidean monopole. In a hyperbolic space with curvature $-\kappa^2$,
\begin{align}
||\Phi|| = \frac12 \left[C\kappa\coth(C\kappa\xi) -\kappa\coth(\kappa\xi) \right]
\end{align}
For the Euclidean monopole with $v = 1/2$, if the limit $\kappa \to 0$ and $C \to \infty$ is taken so that $C\kappa = 2$ is fixed, $\xi$ becomes the usual radial coordinate $r$ of $\mathbb{R}^3$,
\begin{align}
||\Phi|| = \coth2\xi - \frac{1}{2\xi} \ \ (\xi = r)
\end{align}
However, this is not a rational function either as a function of $r=\xi$ or as a function of $e^\xi$. It is also known that the Euclidean monopole with a large magnetic charge is not a rational function.

\end{rem}

\begin{rem}[Jarvis and Norbury (1997)\cite{JN97}]
It has been shown that in the infinite curvature limit $\kappa \to \infty$, a rational map associated with the hyperbolic magnetic monopole appears.
\end{rem}

\begin{rem}[Manton and Sutcliffe(2014)\cite{MS14}]
The hyperbolic magnetic monopoles considered here are $S^1$-invariant instantons, and most of them are also invariant under some subgroup $K$ of $SO(3)$. This leads us to investigate the geometry of the $SO(2)$ action that commutes with $SO(3)$ and the relationship between the $S^1$-invariant instantons and hyperbolic magnetic monopoles.

Examples of large magnetic charges have been constructed that also have the \textbf{platonic}\index{platonic}($K\subset SO(3)$ symmetry of $\mathbb{H}^3$. For example, there are monopoles with tetrahedral symmetry of magnetic charge 3, cubic symmetry of magnetic charge 4, octahedral symmetry of magnetic charge 5, dodecahedral symmetry of magnetic charge 7, and icosahedral symmetry of magnetic charge 17.
For details, see  
Manton and Sutcliffe (2014)\cite{MS14}.

\end{rem}

\subsection{Energy and magnetic charge of hyperbolic magnetic monopoles}

\begin{definition}[Energy of a hyperbolic magnetic monopole]
The ``energy'' (three-dimensional action) $E_3$ of a hyperbolic magnetic monopole is given by:
\begin{align}
E_3 =& \int_{\mathbb{H}^3}d^3x \sqrt{g}\epsilon \nonumber\\
=& \int_{\mathbb{H}^3}d^3x \sqrt{g}\left[\frac{1}{2}\operatorname{tr}(\mathscr{F}_{jk}\mathscr{F}^{jk}) + \operatorname{tr}(\mathscr{D}_j\Phi\mathscr{D}^j\Phi)\right] .
\end{align}
Using the Bogomolnyi equation, the energy density $\epsilon$ of a monopole can be obtained by applying the Laplace-Beltrami operator to the square of the magnitude of the scalar field $\Phi$:
[R. S. Ward (1981)\cite{Ward81}]
\begin{align}
\epsilon = \frac{1}{\sqrt{g}}\partial_\mu(\sqrt{g}g^{\mu\nu}\partial_\nu|\Phi|^2) .
\end{align}
Here, the metric $g_{\mu\nu}$ is that of the ball model.

For a spherically symmetric 1-magnetic monopole, $\epsilon$ is obtained as
\begin{align}
\epsilon = \frac{3}{2}\left(\frac{1 - R^2}{1 + R^2}\right)^4 .
\end{align}

\end{definition}

\begin{rem} 
It is known that a monopole with a large magnetic charge $Q_m>1$ cannot be written in terms of $R$ alone, and is therefore not spherically symmetric.
\end{rem}

\begin{prop}[Bogomol'nyi bound on the hyperbolic magnetic monopole]
The standard Bogomol'nyi argument gives us the lower bound on the ``energy'' (3-dimensional action):
\begin{align}
E_3 \geqq 
  \int_{\partial\mathbb{H}^3} \operatorname{tr}(\Phi\mathscr{F})
:=  \int_{\partial\mathbb{H}^3} d^2S_j \ 2\operatorname{tr}(\Phi\mathscr{B}_j) 
= 4\pi vQ_m ,
\end{align}
where $\mathscr{B}_j$ is the magnetic field $\mathscr{B}_j = \frac{1}{2}\epsilon_{jk\ell} \mathscr{F}_{k\ell}$. 
Here $v$ is the norm of the scalar field at the infinity:
\footnote{
In the literatures in mathematical physics $v$ is also called the mass.  However, it is not the mass of the magnetic monopole, rather it is the mass for the gauge field resulting from the Brout-Englert-Higgs mechanism in the tree level.  
}
\begin{align}
v:=||\Phi(x)||_\infty = \sqrt{\frac{1}{2}\operatorname{Tr}(\Phi(x)^2)} , \ (x \in \partial\mathbb{H}^3) .
\end{align}
The magnetic charge $Q_m$ of a magnetic monopole is defined using the flux of the magnetic field $\mathscr{B}$ through the boundary $\partial\mathbb{H}^3$ at infinity:
\begin{align}
Q_m := \frac{1}{4\pi} \int_{\partial\mathbb{H}^3}\frac{\operatorname{tr}(\Phi\mathscr{F})}{||\Phi||}
= \frac{1}{4\pi v }\int_{\partial\mathbb{H}^3} \operatorname{tr}(\Phi\mathscr{F}) \in \mathbb{Z} .
\end{align}
The moduli space of solutions to the Bogomol'nyi equation is 4$N$-dimensional, which corresponds to the positions and $U(1)$-phases of each of the $N$ monopoles: $(3+1)\times N=4N$.

\end{prop}

\begin{proof}
The ``energy'' can be transformed as follows:
\begin{align}
E_3 =& \int d^3x \ \sqrt{g}\mathscr{L} \ (d^3x := dx^3dx^4d\rho)\nonumber\\
=& \int d^3x \ \left[\rho\frac{1}{2}\operatorname{tr}(\mathscr{F}_{\mu\nu}^2) + \frac{1}{\rho}\operatorname{tr}\{(\mathscr{D}_\mu\Phi)^2\}\right] \nonumber\\
=&\int d^3x \ \left[\rho\operatorname{tr}(\mathscr{E}_j^2) + \rho\operatorname{tr}(\mathscr{B}_j^2) + \frac{1}{\rho}\operatorname{tr}\{(\mathscr{D}_0\Phi)^2\} \right. \nonumber\\
&\left. \quad\quad\quad\quad + \frac{1}{\rho}\operatorname{tr}\{(\mathscr{D}_j\Phi)^2\}\right] \nonumber\\
=&\int d^3x \ \left[\rho\operatorname{tr}\left\{\left(\mathscr{B}_j \mp \frac{1}{\rho}\mathscr{D}_j\Phi\right)^2\right\} + \rho\operatorname{tr}(\mathscr{E}_j^2) \right. \nonumber\\
& \left.  \quad\quad\quad\quad + \frac{1}{\rho}\operatorname{tr}\{(\mathscr{D}_0\Phi)^2\}\right] \nonumber\\
&\pm 2\int d^3x \ \operatorname{tr}(\mathscr{B}_j\mathscr{D}_j\Phi) .
\end{align}
By using the Bianchi identity: 
\begin{align}
\mathscr{D}_j\mathscr{B}_j = \frac{1}{2}\epsilon_{jk\ell}\mathscr{D}_j\mathscr{F}_{k\ell} = \mathscr{D}_j^*\mathscr{F}_j = 0 ,
\end{align}
the last term reads
\begin{align}
& \int d^3x \ \operatorname{tr}(\mathscr{B}_j\mathscr{D}_j\Phi) \nonumber\\ 
=& \int d^3x \ [\partial_j\operatorname{tr}(\mathscr{B}_j\Phi) - \operatorname{tr}(\mathscr{D}_j\mathscr{B}_j\Phi)] \nonumber\\
=&\int d^2S_j \ \operatorname{tr}(\mathscr{B}_j\Phi):=\frac{v}{2}Q_m .
\end{align}
Therefore, we obtain the following inequality
\begin{align}
E_3 =& \int d^3x \ \left[\rho\operatorname{tr}\left\{ \left(\mathscr{B}_j \mp \frac{1}{\rho}\mathscr{D}_j\Phi\right)^2 \right\} + \rho\operatorname{tr}(\mathscr{E}_j^2) 
\right. \nonumber\\
&\left.  \quad\quad\quad\quad + \frac{1}{\rho}\operatorname{tr}\{(\mathscr{D}_0\Phi)^2\}\right] \pm vQ_m  
 \geqq \pm  vQ_m.
\end{align}
where we defined 
\begin{align}
Q_m = \frac{2}{v}\int_{\partial\mathbb{H}^3} d^2S_j \ \operatorname{tr}(\Phi\mathscr{B}_j) := 2 \int_{\partial\mathbb{H}^3} \ \frac{\operatorname{tr}(\Phi\mathscr{F})}{||\Phi||}.
\end{align}
The equality is true if and only if the Bogomolnyi equation holds
\begin{align}
\mathscr{B}_j = \pm \frac{1}{\rho}\mathscr{D}_j\Phi,
\end{align}
and $E_3$ has a minimum value:
\begin{align}
E_3 = vQ_m =  \int_{\partial\mathbb{H}^3} d^2S_j \ 2\operatorname{tr}(\Phi\mathscr{B}_j) .
\end{align}

\end{proof}

\begin{prop}[Relationship between instanton number, magnetic charge, and asymptotic value of scalar field]
The instanton number $I\in\mathbb{Z}$, the magnetic charge $Q_m\in\mathbb{Z}$ of the hyperbolic magnetic monopole, and the asymptotic value $v:=||\Phi||_\infty \in \mathbb{Z}$ of the norm $||\Phi||$ of the scalar field $\Phi$ are related as
\begin{align}
I = 2vQ_m  \Leftrightarrow c_2 =2 ||\Phi||_\infty \ c_1 \ (I, v, Q_m\in\mathbb{Z})
\end{align}

For $v=\frac{1}{2}$, in particular, the instanton number $I$ and the magnetic charge $Q_m$ coincide:
\begin{align}
I = Q_m .
\end{align}
For $v \neq \frac{1}{2}$, then the instanton number $I$ and the magnetic charge $Q_m$ do not coincide:
\begin{align}
& I = 2vQ_m \Leftrightarrow Q_m = \frac{1}{2v}I < I,
\nonumber\\
& v=1 \Rightarrow Q_m=\frac{I}{2}, \ v=\frac32 \Rightarrow Q_m=\frac{I}{3}, ...
\end{align}

\end{prop}

\begin{proof}
As shown by Atiyah(1984)\cite{Atiyah84,Atiyah84b}, the instanton number $I$ is proportional to the magnetic charge $Q_m$:
\begin{align}
I = 2vQ_m .
\end{align}
Here the proportionality factor $v$ is equal to the boundary value $||\Phi||_\infty$ of the scalar field $\Phi$. This arises from the way we lift the circular $S^1$ action onto a bundle carrying $SU(2)$ instantons on a fixed $S^2 \cong \partial\mathbb{H}^3$ of $\mathbb{R}^4$ under the circular $S^1$ action.

The ``energy'' $E_3$ of a hyperbolic magnetic monopole can be written in terms of the asymptotic value $||\Phi||_\infty$ of the norm $||\Phi||$ of the scalar field $\Phi$ and the magnetic charge $Q_m$:
\begin{align}
E_3 = 4\pi||\Phi||_\infty Q_m, \ Q_m=c_1(L) .
\end{align}
The four-dimensional action integral $S_4$ is obtained by integrating with respect to $\varphi$ as follows:
\begin{align}
S_4 = \int_0^{2\pi} d\varphi E_3 = 2\pi E_3.
\label{S4a}
\end{align}
For instantons, on the other hand, we know that
\begin{align}
S_4 = 8\pi^2c_2(P)
\label{S4b}
\end{align}
holds.
Setting (\ref{S4a}) and (\ref{S4b}) equal, we obtain
\begin{align}
& c_2(P) = \frac{S_4}{8\pi^2} = \frac{E_3}{4\pi} = ||\Phi||_\infty Q_m = ||\Phi||_\infty c_1 , \nonumber\\ 
& c_1(L) = Q_m .
\end{align}

\end{proof}

\begin{rem} 
Here, $I= c_2(P)$ is the second Chern class of the $SU(2)$ principal bundle $P$ on $S^4$, and $Q_m = c_1(L)$ is the first Chern class of the line bundle $L$ on the axis $\mathbb{R}^2$.

\end{rem}

\begin{rem}
Even if $2v$ is not an integer: $2v \notin \mathbb{Z}$, hyperbolic magnetic monopoles exist, but they do not correspond to $S^1$-invariant symmetric instantons.
They correspond to singular instantons with singularities of branch type on $S^2$ (with nontrivial holonomy around $S^2$) as shown in Appendix B and C.

\end{rem}

\begin{example}[Hyperbolic magnetic monopole with $Q_m = 1$]
The magnetic charge $Q_m$ of a 1-hyperbolic magnetic monopole is given by
\begin{align}
Q_m &= \frac{1}{2\pi}\int_{\partial\mathbb{H}^3}\operatorname{tr}(\hat{\Phi}d\hat{\Phi}\wedge d\hat{\Phi}) \ , \ \hat{\Phi}:= \frac{\Phi}{||\Phi||} \nonumber\\
&= \frac{1}{16\pi}\int_{\partial\mathbb{H}^3}\operatorname{tr}(\mathbf{n}d\mathbf{n} \wedge d\mathbf{n}) = 1 = c_1.
\label{eq:1}
\end{align}
Here, $\mathbf{n}$ is written with the angular coordinates $\alpha, \beta$ on $\partial\mathbb{H}^3 \simeq S^2$ as
\begin{align}
\mathbf{n} = \sin\alpha\cos\beta\sigma_x +\sin\alpha\sin\beta\sigma_y +\cos\alpha\sigma_z.
\end{align}
The ``energy'' $E_3$ is rewritten into
\begin{align}
E_3
= 4\pi||\Phi||_\infty = 4\pi c_1 .
\label{eq:2a}
\end{align}
The asymptotic value $||\Phi||_\infty$ is equal to the value of $||\Phi||$ on the boundary $\partial\mathbb{H}^3 \simeq S^2$:
\begin{align}
||\Phi||_\infty := \lim_{\substack{s \to \infty \\ (R \to 1)}}||\Phi|| = ||\Phi(x \in \partial\mathbb{H}^3)|| = C -1=1.
\end{align}

\end{example}


\begin{example}[Vortex embedded in a hyperbolic magnetic monopole]
The first Chern number $c_1$ can be found by performing the integral over the boundary $\partial\mathbb{H}^3 \simeq S^2$:
\begin{align}
c_1 = \frac{1}{4\pi}\int_{\partial\mathbb{H}^3}\frac{\operatorname{tr}(\Phi\mathscr{F})}{||\Phi||} .
\end{align}
Hence, the first Chern number $c_1$ reduces to the boundary integral of the flux of the gauge field:
\begin{align}
c_1 = \frac{1}{4\pi}\int_{\partial\mathbb{H}^3}dx^4 \wedge dx^3 \ F_{4r} \ , \ F_{4r}:=\partial_4a_r - \partial_ra_t.
\end{align}
For coordinate $r>0$, since $x^3\in\mathbb{R}$, the integration must be done over two copies of $\mathbb{H}^2 \ni (x^4, r>0)$, i.e. $\partial\mathbb{H}^3$ is a double copy of $\mathbb{H}^2$:
\begin{align}
\partial\mathbb{H}^3 \ni (x^4, x^3, \rho=0)\simeq 2(x^4, r) \ \ (r>0)
\end{align}
Hence, we have
\begin{align}
c_1 = \frac{1}{2\pi}\int_{\mathbb{H}^2}dx^4 \wedge dr \ 
F_{4r} = N_v
\end{align}
This is equal to the vortex charge $N_v$.

\end{example}

\begin{rem} 
For the magnetic monopole with the charge 2, see Ward(1981)\cite{Ward81} and Atiyah and Ward(1977)\cite{AW77}. 
In the very large limit of the magnetic charge $Q_m \to \infty$, the hyperbolic magnetic monopole  is simplified, which is called the \textbf{magnetic bag}.  
See S. Bolognesi(2006)\cite{Bolognesi06}, Bolognesi, Cockburn and Sutcliffe(2014)\cite{BCS14}, Bolognesi, Harland and Sutcliffe(2015)\cite{BHS15}.

\end{rem}

\subsection{JNR data for instantons and dimensional reduction}

The multi-instanton, i.e., $N$-instantons are found via the \textbf{Jackiw-Nohl-Rebb (JNR) Ansatz}\index{Jackiw-Nohl-Rebb (JNR) Ansatz} (Jackiw-Nohl-Rebb(1977)\cite{JNR77}:
\begin{align}
\mathscr{A}_\mu^A(x) = \bar{\eta}_{\mu\nu}^A\partial_\nu\ln\Xi ,
\end{align}
with the superpotential $\Xi$:
\begin{align}
\Xi(x) = \sum_{j = 0}^N\frac{\lambda_j^2}{|x - a_j|^2} .
\end{align}

The dimensional reduction can be performed at the level of the JNR data. The $S^1$-invariant JNR data gives a subset of hyperbolic magnetic monopoles. The poles $a_j$ $(j=0,...,N)$ of $\Xi$ must lie on a plane $a_j \in (x^3,x^4)$ (a fixed set of the circular $S^1$ action) in $\mathbb{E}^4$. This plane is the boundary $\partial \mathbb{H}^3$ of $\mathbb{H}^3$. By counting parameters, it can be seen that all hyperbolic  magnetic monopoles in $N\leq 3$ can be generated in this way because $(2+1)(N+1) \ge (3+1)N$. 

To reduce a  magnetic monopole to a vortex, the poles must lie on a fixed set of $SO(3)$ actions, i.e., on a line in $\mathbb{E}^4$. It was shown by Manton (1978)\cite{Manton78} that the JNR Ansatz generates all hyperbolic vortices.

The scalar field $\Phi$ of the magnetic monopole is expressed by the JNR superpotential
\begin{equation}
\Xi = \sum_{j =0}^N\frac{\lambda_j^2}{|x^4 + ix^3 - \gamma_j|^2 + \rho^2} , \ \gamma_j \in \mathbb{C}, \ (\rho, x^3, x^4) \in \mathbb{H}^3
\label{eq:3a}
\end{equation}
in the form
\begin{align}
||\Phi||^2 = \frac{\rho^2}{4\Xi^2}\left[\left(\frac{\partial \Xi}{\partial x^4}\right)^2 + \left(\frac{\partial \Xi}{\partial x^3}\right)^2 + \left(\frac{\Xi}{\rho} + \frac{\partial \Xi}{\partial \rho}\right)^2\right] .
\label{eq:4}
\end{align}
If all poles are placed on the real $x^4$ axis, $\Xi$ becomes:
\begin{align}
\Xi = \sum_{j = 0}^N\frac{\lambda_j^2}{(x^4 - \gamma_j)^2 + r^2} , \ \gamma_j \in \mathbb{R}.
\label{eq:5}
\end{align}

On the other hand, the scalar field $\phi$ of the vortex obeys the following equation:
\begin{align}
|\phi|^2 =& \frac{r^2}{\Xi^2}\left[\left(\frac{\partial \Xi}{\partial x^4}\right)^2 +\left(\frac{\Xi}{r} + \frac{\partial \Xi}{\partial r}\right)^2\right]
\nonumber\\
=& -r^2(\partial_4^2 + \partial_r^2)\log(r\Xi) .
\end{align}
We can choose $\mathscr{A}_j^G$ to satisfy the Coulomb gauge $\partial_j\mathscr{A}_j^G = 0$:
\begin{align}
a_4 = \frac{\phi_2 + 1}{r} \ , \ a_r = \frac{\phi_1}{r} .
\end{align}
This fixes the phase of $\phi$. Then
\begin{align}
a_4 = -\partial_r\log\Xi \ , \ a_r = \partial_4\log\Xi .
\end{align}
Using \eqref{eq:5} in \eqref{eq:4} and changing the variables, we reproduce the following relations.
\begin{align}
||\Phi||^2 = \frac{1}{4r^2}[\rho^2|\phi|^2 + (x^3)^2] .
\end{align}

\begin{rem}[Dimensional reduction: case (III)]\label{rem:caseIII}
Further dimensional reduction can be performed.
$SO(4)$ invariant instantons reduce to one-dimensional field theory.
Using the new radial coordinate $R := \sqrt{r^2 +(x^4)^2}$,
\begin{align}
\varphi^2 = \frac{R^2}{\Xi^2}\left(\frac{\Xi}{R} + \frac{d\Xi}{dR}\right)^2 \ , \ \Xi = \Xi(R) .
\end{align}
Here, $\Xi$ is a function of $R$ only. Combining this with previous results, the corresponding vortex scalar field reads
\begin{align}
|\phi|^2 = \frac{r^2\varphi^2 + (x^4)^2}{r^2 + (x^4)^2} .
\end{align}
Substituting this into the Taubes equation:
\begin{align}
\Delta\log|\phi|^2 + 2(1 - |\phi|^2) = 0,
\end{align}
we obtain the Bogomolny equation for the $\varphi^4$ kink:
\begin{align}
\frac{d\varphi}{d\log R} = 1 - \varphi^2 .
\end{align}
In other words, by embedding the $\varphi^4$ kink in $\mathbb{H}^2$, we obtain a hyperbolic vortex of charge 1, see 
Maldonado (2017)\cite{Maldonado17}.
\end{rem}

\section{Hyperbolic vortex solutions}

\begin{prop}[The dimesionally reduced $D=2$ $U(1)$ gauge-scalar model]

By this spatial $SO(3)$ \textbf{symmetry reduction}\index{symmetry reduction}, the cylindrical (spherically symmetric) sector of the $SU(2)$ Yang-Mills theory on the four-dimensional Euclidean space $\mathbb{E}^4$ 
\begin{align}
 S_{\rm{YM}} &= \int d^4 x \ \mathscr{L}_{\rm{YM}} , \nonumber\\
  \mathscr{L}_{\rm{YM}} &= \frac{1}{4} \mathscr{F}_{\mu \nu}^A \mathscr{F}_{\mu \nu}^A + \frac{\vartheta}{32 \pi^2} \mathscr{F}_{\mu \nu}^A {}^* \mathscr{F}_{\mu \nu}^A,
\end{align}
is reduced to the $U(1)$ gauge scalar theory on the two-dimensional hyperbolic space $\mathbb{H}^2$:
\begin{align}
\Rightarrow 
S_{\rm{YM}} = &4 \pi \int_{- \infty}^\infty dt \int_{0}^\infty d r \ \mathscr{L}_{\rm{GS}} , \nonumber\\
 \mathscr{L}_{\rm{GS}} 
 = &\frac{1}{4} r^2 F_{\mu \nu} F_{\mu \nu}+ (D_\mu \phi)^* D_\mu \phi \nonumber\\
 &+ \frac{1}{2 r^2}(|\phi|^2 - 1)^2 + \frac{\vartheta}{16 \pi^2} \varepsilon_{\mu \nu} F_{\mu \nu}.
 \label{YM-L}
\end{align}
Here we defined the covariant derivative $D_\mu=\partial_\mu - ia_\mu$ and introduced the complex scalar field $\phi = \phi_1 + i \phi_2$ and used $D_\mu \varphi_a D_\mu \varphi_a = (D_\mu \phi)^* D_\mu \phi$.
We have ignored total differential terms.
In the theory with dimensional reduction, the gauge group becomes $U(1)$.

This is indeed a theory defined on $\mathbb{H}^2$ with the metric $g_{\mu\nu}= r^{- 2} \delta_{\mu\nu}$ ($g^{\mu\nu}= r^{2} \delta^{\mu\nu}$):
\begin{align}
 S_{\rm{YM}} = &\int_{- \infty}^\infty dt \int_{0}^\infty d r \sqrt{g} \ \mathscr{L}_{\rm{GS}} , \nonumber\\
 \mathscr{L}_{\rm{GS}} = & \frac14  g^{\mu\alpha}g^{\nu\beta} F_{\mu \nu} F_{\alpha \beta} +  g^{\mu\nu} (D_\mu \phi)^* D_\nu \phi  
\nonumber\\
 & + \frac{1}{2}(|\phi|^2 - 1)^2 + \frac{\vartheta}{16 \pi^2} \varepsilon_{\mu \nu} F_{\mu \nu},
\label{YM-L2}
\end{align}
where $g:=\det (g_{\mu\nu})=r^{-2}$.
\end{prop}

\begin{proof}
This result was first discovered by Witten (1977)\cite{Witten} in the course to find multi-instanton solutions of the four-dimensional $SU(2)$ Yang-Mills theory.

\end{proof}

\begin{prop}[Bogomolnyi bound]
Consider the U(1) gauge-scalar model with the parameters $q,\lambda,v$ on $\mathbb{H}^2$:
\begin{align}
\frac{S_{GS}}{4\pi} =& \int^\infty_{-\infty}dt\int^\infty_{0}dr \ \left[\frac{r^2}{4}F_{\mu\nu}^2 + |D_\mu\phi|^2 \right. \nonumber\\
& \left. + \frac{\lambda^2}{2r^2}(|\phi|^2 - v^2)^2\right] ,
\end{align}
where we defined the covariant derivative $D_\mu=\partial_\mu - iqa_\mu$.
When $\lambda = q$ ($v$ is arbitrary),  the action integral has the lower bound:
\begin{align}
\frac{S_{GS}}{4\pi} =& \int dt \int dr \ \left\{ |D_0\phi \mp iD_1\phi|^2 
\right. \nonumber\\
& \left. \quad\quad\quad\quad\quad + \frac{r^2}{2}(F_{01} \mp \frac{q}{r^2}(| \phi|^2 - v^2))^2\right\} \nonumber\\
&\mp qv^2\int dt \int dr \ F_{01} \nonumber\\
\geqq& \mp qv^2\int dt \int dr \ F_{01}.
\end{align}
Here the equality holds, namely, the action is given by
\begin{align}
\frac{S_{GS}}{4\pi} = qv^2\int dt \int dr \ F_{01} = 2\pi qv^2 N_v,
\end{align}
if and only the following pair of  vortex equations are satisfied:
\begin{align}
\left\{\,
\begin{aligned}
&D_0\phi \mp iD_1\phi = 0,\\
&F_{01}=\pm\frac{q}{r^2}(|\phi|^2 - v^2).
\end{aligned}
\right.
\label{vortex_Bogolnyi_eq}
\end{align}
Here $N_v$ is the vortex number defined by
\begin{align}
N_v := \frac{1}{2\pi}\int^{+\infty}_{-\infty}dt \ \int^{\infty}_{0}dr \ F_{01} .
\end{align}

\end{prop}

\begin{proof}
The action is rewritten as 
\begin{align}
& \frac{S_{GS}}{4\pi}  
= \int dt dr \ \left[\frac{r^2}{4}F_{\mu\nu}^2 + |D_\mu\phi|^2 \right. \nonumber\\
& \left. \quad\quad\quad\quad\quad\quad\quad + \frac{\lambda^2}{2r^2}(|\phi|^2 - v^2)^2\right] \nonumber\\
=&\int dt dr \ \left[\frac{r^2}{2}F_{01}^2 + |D_0\phi|^2 +|D_1\phi|^2 \right. \nonumber\\
& \left. \quad\quad\quad\quad+ \frac{\lambda^2}{2r^2}(|\phi|^2 - v^2)^2\right] \nonumber\\
=&\int dt dr \ \left[|D_0\phi \mp iD_1\phi|^2 \right. \nonumber\\
& \left. \quad\quad\quad\quad + \frac{r^2}{2}\left[F_{01} \mp \frac{\lambda}{r^2}(|\phi|^2 - v^2)\right]^2\right] \nonumber\\
&\pm i\int dt dr \ [(D_0\phi)^*(D_1\phi) - (D_1\phi)^*(D_0\phi) \nonumber\\
&\quad\quad\quad\quad\quad\quad - i\lambda F_{01}(|\phi|^2 - v^2)].
\end{align}
The last term reads
\begin{align}
&(D_0\phi)^*(D_1\phi) - (D_1\phi)^*(D_0\phi) \nonumber\\
=&(\partial_0\phi^* + iqa_0\phi^*)(\partial_1\phi - iqa_1\phi) \nonumber\\
&- (\partial_1\phi^* + iqa_1\phi^*)(\partial_0\phi - iqa_0\phi) \nonumber\\
=&\partial_0\phi^*\partial_1\phi -\partial_1\phi^*\partial_0\phi +iqa_0\phi^*\partial_1\phi +  iqa_0\partial_1\phi^*\phi \nonumber\\
&-iqa_1\partial_0\phi^*\phi - iqa_1\phi^*\partial_0\phi \nonumber\\
=&\partial_0(\phi^*\partial_1\phi) - \partial_1(\phi^*\partial_0\phi) + \partial_1(iqa_0\phi^*\phi) \nonumber\\
&- iq\partial_1a_0\phi^*\phi - \partial_0(iqa_1\phi^*\phi) + iq\partial_0a_1\phi^*\phi \nonumber\\
=&iq(\partial_0a_1 - \partial_1a_0)|\phi|^2 + \text{boundary term},
\end{align}
which means the cancellation occurs when $\lambda = q$:
\begin{align}
&\pm i\int dt \int dr \ [(D_0\phi)^*(D_1\phi) - (D_0\phi)(D_1\phi)^*   \nonumber\\
&\quad\quad\quad\quad - i\lambda F_{01}(|\phi|^2 - v^2)] \nonumber\\
=&\pm i\int dt dr \ iqF_{01}v^2 \ (\Leftarrow \lambda = q).
\end{align}
Here we neglected the boundary term as giving the vanishing contribution:
\begin{align}
& \int dt dr \ [\partial_0(\phi^*\partial_1\phi) - \partial_1(\phi^*\partial_0\phi) 
\nonumber\\
& \quad\quad\quad\quad- \partial_0(iqa_1|\phi|^2) + \partial_1(iqa_0|\phi|^2)] = 0.
\end{align}

\end{proof}

The theory (\ref{YM-L2}) corresponds to the case of $\lambda = q=v=1$. 
These equations (\ref{vortex_Bogolnyi_eq}) are the Bogomolnyi equation for the vortex:
\begin{align}
&D_0\phi + iD_1\phi = 0,  
\label{eq:15a}\\
&F_{01}+ \frac{\Omega}{2}(|\phi|^2-1) = 0, \ \Omega := \frac{1}{r^2} .
\label{eq:16a}
\end{align}
Equations \eqref{eq:15a}, \eqref{eq:16a} are self-dual equations written using Witten Ansatz $\mathscr{F}_{\mu\nu} = {}^*\mathscr{F}_{\mu\nu} \Leftrightarrow \mathscr{F}_{0j}^A =  \frac{1}{2}\epsilon_{jk\ell}\mathscr{F}_{k\ell}^A$:
\begin{align}
&\partial_0(\phi_1 + i\phi_2) - ia_0(\phi_1 + i\phi_2)
\nonumber\\
&+ i\partial_1(\phi_1 + i\phi_2) + a_1(\phi_1 + i\phi_2) = 0 \nonumber\\
\Leftrightarrow & \quad
\left\{\,
\begin{aligned}
&\partial_0\phi_1 + a_0\phi_2 = \partial_1\phi_2 - a_1\phi_1 \\
&\partial_1\phi_1 + a_1\phi_2 = -(\partial_0\phi_2 -a_0\phi_1).
\end{aligned}
\right. \\
&F_{01} - \frac{\Omega}{2}(1 - |\phi|^2) = 0 \nonumber\\
\Leftrightarrow &\partial_0a_1 - \partial_1a_0 = \frac{1}{r^2}(1 - \phi_1^2 - \phi_2^2).
\end{align}

\begin{definition}[vortex equation]
The complex scalar field $\phi$ is rewritten in terms of $h$ and $\chi$,
\begin{align}
\phi =&\phi_1+ i \phi_2 = e^{\frac{1}{2}h+i\chi} = e^{\frac{1}{2}h}e^{i\chi} \nonumber\\
=& |\phi|e^{i\chi} , \ e^{\frac{1}{2}h} = |\phi|.
\end{align}
In other words, if $h$ is defined by
\begin{align}
h:=\log |\phi|^2 ,
\end{align}
then $h$ is gauge invariant and finite except for the zero point where $\phi = 0$. Also, $h$ is $h=0$ on the circle at infinity (corresponding to $|\phi| = 1$).

Rewriting \eqref{eq:15} gives us
\begin{align}
&\phi^{-1}\partial_0\phi - ia_0 + i\phi^{-1}\partial_1\phi + a_1 = 0 \nonumber\\
\Leftrightarrow & \partial_0\left(\frac{1}{2}h + i\chi \right) - ia_0 + i\partial_1\left(\frac{1}{2}h + i\chi\right) + a_1 = 0 \label{eq:17}
\end{align}
Therefore, the gauge fields $a_0, a_1$ can be written in terms of $h$ and $\chi$:
\begin{align}
a_0 = \frac{1}{2}\partial_1 h + \partial_0 \chi \ , \ a_1 = -\frac{1}{2}\partial_0 h + \partial_1\chi .
\end{align}
Therefore, $F_{01}$ can be written as
\begin{align}
F_{01}:=\partial_0a_1 - \partial_1a_0 = -\frac{1}{2}(\partial_0^2 + \partial_1^2)h = -\frac{1}{2}\nabla^2h.
\end{align}
and \eqref{eq:16a} reduces to an equation involving only $h$:
\begin{align}
\nabla^2h + \Omega(1 - e^h) = 0 \Leftrightarrow \nabla^2\log|\phi|^2 + \Omega(1-|\phi|^2) = 0.
\end{align}
This is gauge invariant because it does not depend on $\chi$ and is an equation involving only $h$. where $\nabla^2$ is the standard Laplacian (Laplace-Beltrami operator). This equation is called the \textbf{Taubes equation}(Taubes equation)\index{Taubes equation@Taubes equation}. 

These equations are valid except for the zeros of $\phi$, i.e., the logarithmic singularities of $h$. This singularity is captured by including a source for the Dirac delta function:
\begin{align}
\nabla^2h + \Omega - \Omega e^h = 4\pi\sum_{r=1}^N \delta^2(\bm{X} - \bm{a}_r) .
\label{eq:17}
\end{align}
Here $\{\bm{a}_r\}$ is the location vector of the zeros of $\phi$.

\end{definition}

\begin{prop}[Liouville equation]

To solve this equation, we introduce a new unknown function $\rho$ as \begin{align}
& h = 2\rho + 2\log\frac{1}{2}(1 - |z|^2) \nonumber\\
&
\Leftrightarrow \rho = \frac{1}{2}h - \log\frac{1}{2}(1 - |z|^2) \ (|z|<1)
\end{align}
The equation \eqref{eq:17} for $h$ is rewritten into the \textbf{Liouville equation}\index{Liouville equation}:
\begin{align}
\nabla^2\rho - e^{2\rho} = 2\pi\sum_{r=1}^N\delta^2(z-\rho_r) .
\label{eq:18}
\end{align}
This is \textbf{integrable}\index{integrable}, i.e. it can be solved exactly, and the solution is found:
\begin{align}
\rho(z) = -\log\frac{1}{2}(1 - |f|^2) + \frac{1}{2}\log\left|\frac{df}{dz}\right|^2.
\label{eq:19}
\end{align}
Here, $f(z)$ is any complex analytic function.

\end{prop}

\begin{proof}
See section 7.8 of Manton and Sutcliffe (2004)\cite{MS04}.

\end{proof}



\begin{figure}[htb]
\begin{center}
\includegraphics[scale=0.2]{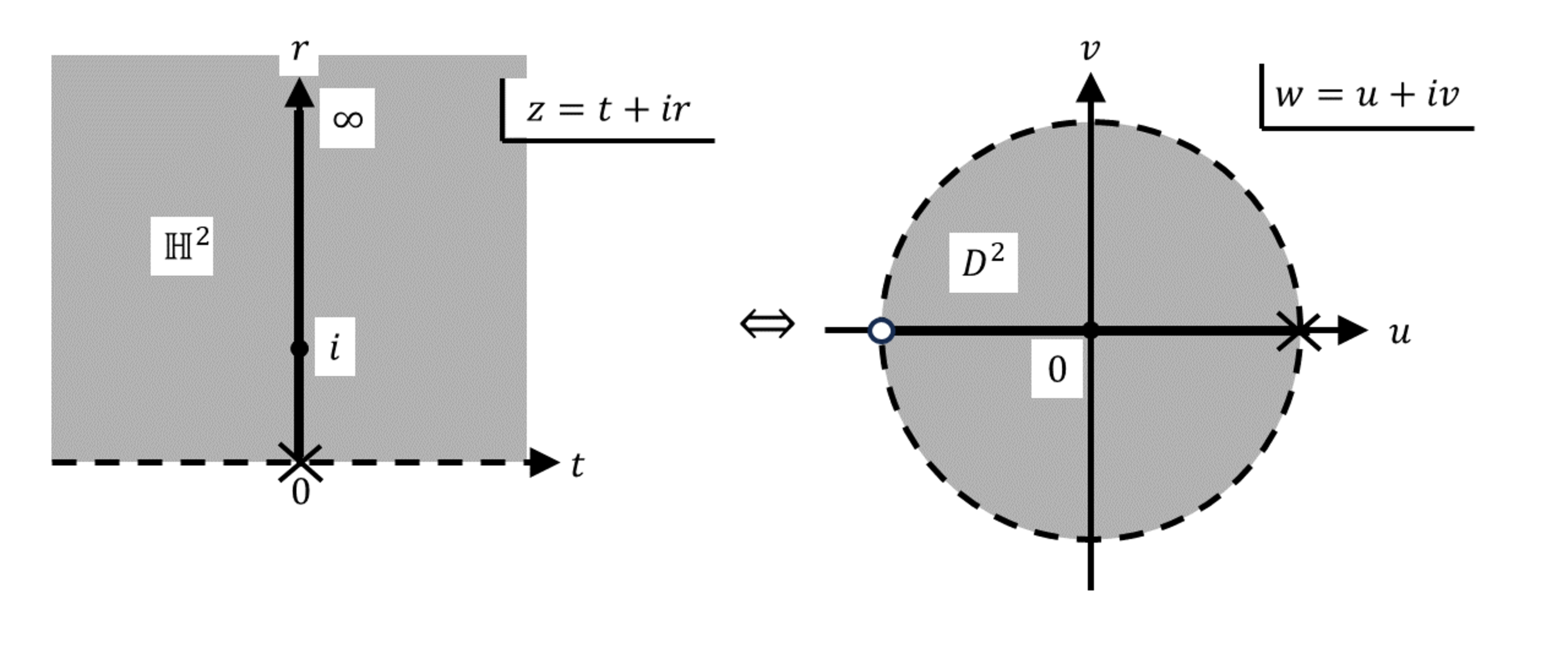}
\end{center}
\vskip -0.5cm
\caption{
The relationship between the complex number $z=t+ir \in \mathbb{H}^2$ in the upper half plane and the complex number $w=u+iv \in \mathbb{D}^2$ in the unit disk.
}
\label{H2_conf_transf}
\end{figure}

In what follows, we move on to another complex variable $w$. 
Let $z$ be the complex coordinate of the upper half-plane $\mathbb{H}^2$, and $w$ be the complex coordinate of the unit disk $D^2$. $z$ and $w$ are related by a conformal transformation or \textbf{Cayley transformation}\index{Cayley transformation}:
\begin{align}
w = \frac{i -z}{i + z} \Leftrightarrow z = i\frac{1 -w}{1 + w} .
\label{eq:9ccc}
\end{align}
See Figure\ref{H2_conf_transf}.  
This mapping maps the upper half-plane $\mathbb{H}^2$ to the interior of the unit disk $D^2$. It maps the real axis of $\mathbb{H}^2$ to the unit circle of $D^2$. In particular, it maps the point $z = i$ of $\mathbb{H}^2$ to the origin $w = 0$ of $D^2$, and the origin $z = 0$ of $\mathbb{H}^2$ to $w = 1$ of $D^2$. It maps the point at infinity $z = \infty$ to $w = -1$. It maps the positive imaginary axis of $z$ to the interval $(-1,1)$ on the real axis of $w$. This disk $D^2$ is called the \textbf{Poincar\'{e} disk}\index{Poincar\'{e} disk}.
The metric on $D^2$ is given by:
\begin{align}
(ds)^2 = \frac{4dwd\bar{w}}{(1 - |w|^2)^2}.
\end{align}
The vortex equation is invariant under this coordinate transformation \eqref{eq:9ccc}. The equation $\bar{\partial}\phi = 0$ holds because the transformation \eqref{eq:9ccc} is regular. The equation $*_hF_A = 1 - |\phi|^2$ is invariant because it is a relation between scalar quantities that are invariant under coordinate transformations.

\begin{prop}[Solutions of the vortex equation]

Therefore, the scalar field $\phi$ is found as follows:
\begin{align}
|\phi| &= e^{\frac{1}{2}h} = e^{\rho + \log\frac{1}{2}(1-|w|^2)} = \frac{1}{2}(1 - |w|^2)e^\rho \nonumber\\
&=\frac{1}{2}(1 - |w|^2)e^{-\log\frac{1}{2}(1 - |f|^2) + \log\left|\frac{df}{dw}\right|} \\
\Rightarrow& |\phi(w)| = \frac{1 - |w|^2}{1 - |f(w)|^2}\left|\frac{df(w)}{dw}\right|.
\end{align}
Note that $\phi$ is not gauge invariant, while $|\phi|$ is gauge-invariant. 
If the phase of $\phi$ is chosen appropriately,
\begin{align}
\boxed{\phi(w) = \frac{1 - |w|^2}{1 - |f(w)|^2}\frac{df(w)}{dw}}.
\end{align}
\eqref{eq:15} can be written as $D_{\bar{w}}\phi = 0$, therefore the gauge field $a_\mu$ can be found as 
\begin{align}
&D_{\bar{w}}\phi = 0 \Leftrightarrow \partial_{\bar{w}}\phi = ia_{\bar{w}}\phi \Leftrightarrow a_{\bar{w}} = -i\phi^{-1} \partial_{\bar{w}}\phi \\
\Leftrightarrow&\boxed{a_{\bar{w}} = -i\partial_{\bar{w}}\log \phi = -i\partial_{\bar{w}}\log\left(\frac{1 - |w|^2}{1 - |f(w)|^2}\right)}.
\end{align}
Here $\partial_{\bar{w}}f(w) = 0$ is used since $f(w)$ is an analytic function. The center of the vortex $\{\bm{a}_r\}$ is the point $w = a_r$, where $\frac{df(w)}{dw}=0$, or $\phi(w) = 0$.

Next, we need to ensure that $|\phi(w)| = 1$ on the boundary of the disk when $|w| = 1$, and that $\phi$ has no singularities inside the disk when $|w|< 1$. This requires that $|f(w)| = 1$ on the boundary and $|f(w)| < 1$ inside. These conditions are satisfied by choosing $f$ to have the form \textbf{Blaschke product}\index{Blaschke@Blaschke product}:
\begin{align}
\boxed{f(w) = \prod_{j=1}^{N+1}\left(\frac{w - c_j}{1 - c_j^*w}\right) \quad (|c_j|<1, \forall j)}.
\end{align}
Each factor of this product has absolute value less than 1 inside the unit disk and 1 on the boundary, and its phase is strictly increasing along the boundary. The same holds for $f(w)$. Therefore, $|\phi(w)|$ has no singularities in the unit disk, and since $\frac{df}{dw}$ is nonzero on the boundary, we see that $|\phi(w)| = 1$ there, and the radial derivative of $|\phi(w)|$ is also zero. The solution does not depend on $(N + 1)$ complex parameters, but only on $N$, because there exists a 1-parameter family of Möbius transformations of $f$ that only generate gauge transformations of $\phi$, thus keeping the zeros of $\phi$ fixed.

\end{prop}

\begin{proof}
See section 7.14.3 of Manton and Sutcliffe (2004)\cite{MS04}.

\end{proof}

\begin{example}[Manton and Sutcliffe (2004)\cite{MS04} section 7.14.3]
A simple example of a solution is the case where $w = 0$ has $N$ coincident zeroes, and with $c_j = 0(\forall j)$: $f(w) = w^{N+1}$. Therefore,
\begin{align}
\phi(w) =& \frac{1 - |w|^2}{1 - |w|^{2N + 2}}(N + 1)w^N \nonumber\\
=& \frac{(N+1)w^N}{|w|^{2N} +|w|^{2N-2} + \cdots +|w|^{2} + 1}.
\end{align}
This $\phi(w)$ has $N$ multiple zeros at $w = 0$, satisfies the boundary conditions, and has the winding number $N$ along the boundary. Any transformation of this solution will also have the winding number $N$ as long as all parameters $c_j$ do not cross the unit circle, i.e., it will be $N$ vortexes. 

In particular, for $N = 1$, we have 
\begin{align}
 f(w) = w^2, \ \phi(w)  = \frac{2w}{1+|w|^2} , 
\end{align}
The gauge field is obtained as
\begin{align}
 a_{\bar{w}}  =& a_4+ia_r= -i(\partial_0+i\partial_r) \log \phi \nonumber\\
 =& -i\partial_{\bar{w}}\log \phi = i\partial_{\bar{w}}\log(1+w\bar{w}) = i\frac{w}{1+|w|^2}.
\end{align}
\end{example}

The more details of the hyperbolic 1-vortex solution (note that $w$ is dimensionless) are as follows. 
The complex scalar field of the hyperbolic 1-vortex solution is written as
\begin{align}
\phi = \frac{2w}{1 + |w|^2} = \phi_1 + i\phi_2 .
\end{align}
The complex number $w$ is obtained from $z = t + ir$ by the Cayley  transformation:
\begin{align}
w = \frac{-i\lambda + z}{i\lambda + z} 
= \frac{-i\lambda + t + ir}{i\lambda + t + ir} 
= \frac{t^2 + (r^2 - \lambda^2)-i2\lambda t}{t^2 + (r + \lambda)^2},
\end{align}
which implies
\begin{align}
|w|^2 
= \frac{t^2 + (r - \lambda)^2}{t^2 + (r + \lambda)^2} 
\Rightarrow 1+|w|^2 
=  2\frac{t^2 + r^2 + \lambda^2}{t^2 + (r + \lambda)^2} .
\end{align}
Therefore, $\phi_1,\phi_2$ can be written as functions of $t,r$ as
\begin{align}
\phi_1 = \frac{t^2 + r^2 - \lambda^2}{t^2 + r^2 + \lambda^2} , \
 \phi_2 = -\frac{2\lambda t}{t^2 + r^2 + \lambda^2} .
\end{align}
Then, the magnitude $|\phi|$ of the gauge-invariant scalar field can be found as 
\begin{align}
|\phi|^2 =& \phi_1^2 + \phi_2^2 
= \frac{r^4 + 2r^2(t^2 - \lambda^2) + (t^2 + \lambda^2)^2}{(t^2 + r^2 + \lambda^2)^2} \nonumber\\
=& \frac{[t^2 + (r - \lambda)^2][t^2 + (r +\lambda)^2]}{(t^2 + r^2 + \lambda^2)^2}, 
\nonumber\\
\Rightarrow |\phi| =& \sqrt{\phi_1^2 + \phi_2^2} 
= \frac{\sqrt{[t^2 + (r - \lambda)^2][t^2 + (r +\lambda)^2]}}{t^2 + r^2 + \lambda^2}.
\end{align}
When we introduce $h$ and $\chi$ by
\begin{align}
\phi = e^{\frac{1}{2}h + i\chi} \Rightarrow e^{\frac{1}{2}h} = |\phi| \ , \ e^{i\chi} = \frac{\phi}{|\phi|},
\end{align}
$h$ is written as 
\begin{align}
\frac{h}{2} =& \ln |\phi| =  \ln \frac{\sqrt{(t^2 + \lambda^2)^2 + 2(t^2 - \lambda^2)r^2 + r^4}}{t^2 + r^2 + \lambda^2} \nonumber\\
=&\frac{1}{2}\ln [(t^2 + \lambda^2)^2 + 2(t^2 - \lambda^2)r^2 + r^4] 
\nonumber\\
&- \ln(t^2 + r^2 + \lambda^2).
\end{align}
This leads to the result:
\begin{align}
(\partial_t^2 + \partial_r^2)\frac{h}{2} = -\frac{4\lambda^2}{(t^2 + r^2 + \lambda^2)^2}.
\end{align}
On the other hand, $\chi$ is written as 
\begin{align}
e^{i\chi} =& \frac{t^2 + r^2 - \lambda^2}{\sqrt{(t^2 + \lambda^2)^2 + 2(t^2 - \lambda^2)r^2 + r^4}} 
\nonumber\\
&+ i\frac{-2\lambda t}{\sqrt{(t^2 + \lambda^2)^2 + 2(t^2 - \lambda^2)r^2 + r^4}}
\nonumber\\
\Rightarrow \chi=&\arctan \frac{-2\lambda t}{t^2 + r^2 - \lambda^2} .
\end{align}
Therefore, we calculate the $U(1)$ gauge field
\begin{align}
a_0 =& \partial_1\frac{h}{2} + \partial_0 \chi = \frac{2\lambda(t^2 - r^2 + \lambda^2)}{(t^2 + r^2 + \lambda^2)[t^2 + (r + \lambda)^2]}, \\
a_1 =& -\partial_0\frac{h}{2} + \partial_1 \chi = \frac{4\lambda tr}{(t^2 + r^2 + \lambda^2)[t^2 + (r + \lambda)^2]}.
\end{align}
Therefore, the strength of the $U(1)$ gauge field is
\begin{align}
F_{01}:=&\partial_0a_1 - \partial_1a_0 = \frac{4\lambda^2}{(t^2 + r^2 + \lambda^2)^2} \nonumber\\
=& -(\partial_t^2 + \partial_r^2)\frac{h}{2}.
\end{align}
On the other hand, we find
\begin{align}
1 - |\phi|^2 = \frac{4\lambda^2 r^2}{(t^2 + r^2 + \lambda^2)^2}.
\end{align}
Therefore, it indeed satisfies the vortex equation:
\begin{align}
F_{01} = \frac{1}{r^2}(1 - |\phi|^2).
\end{align}
Here, $\chi$ corresponds to the degree of freedom of the $U(1)$ gauge transformation. Therefore, the gauge-invariant $F_{01}$ depends on $h$, but not on $\chi$.
See Fig.~\ref{F_vortex}. 

The functions $a_0, a_1, \phi_1, \phi_2$ are obtained by directly solving the two-dimensional vortex equations and they are also obtained by dimensional reduction from the four-dimensional superpotential $\Xi_4$. 
Therefore, the need to agree with each other and they may differ by a certain gauge transformation, since they are not gauge invariant. However, since $F_{01}$ and $|\phi|$ are gauge invariants, the two must be the same.

\begin{figure}[htb]
\begin{center}
\includegraphics[scale=0.6]{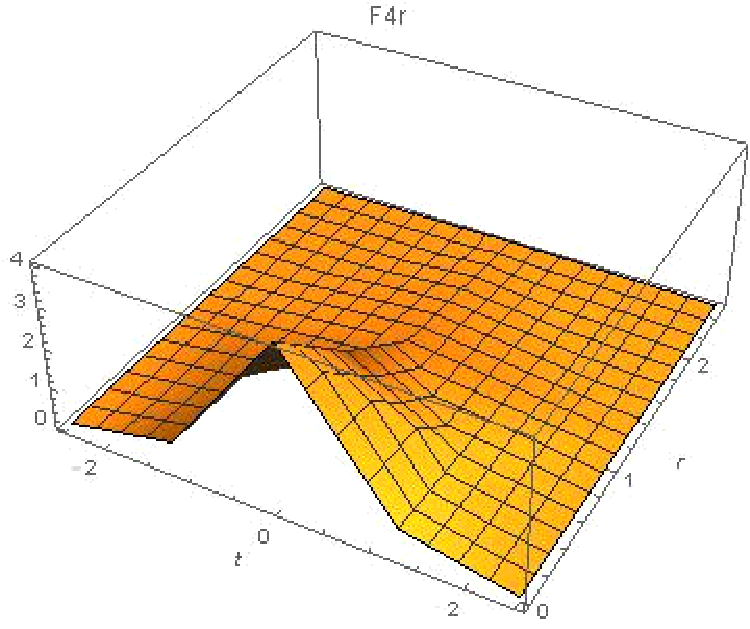}
\includegraphics[scale=0.6]{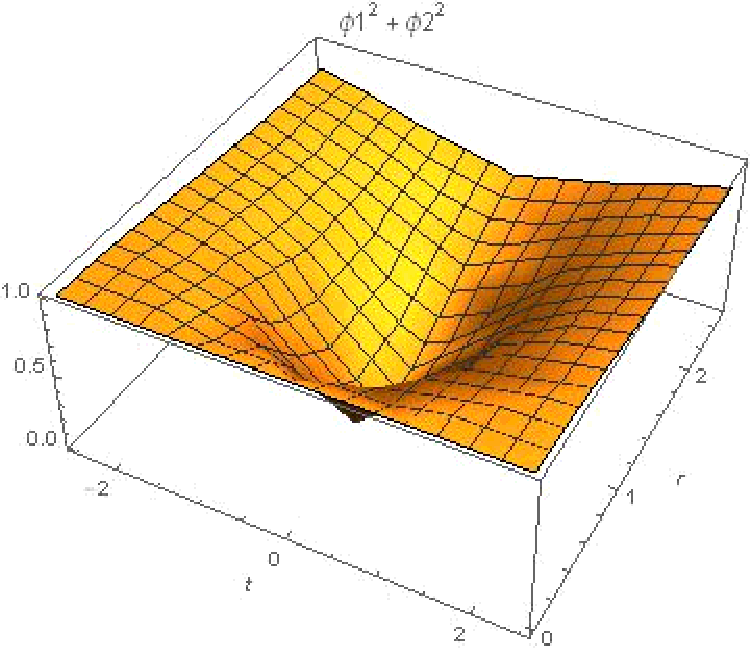}
\end{center}
\caption{
The 1-vortex solution with the center at $(t,r)=(0,1)$ and the size $\lambda=1$. The distribution of gauge-invariant quantities . 
(Left) field strength $F_{01}(t,r)$, (Right) $|\phi(t,r)|^2$. 
}
\label{F_vortex}
\end{figure}

We see that the remaining vortex equations are also satisfied:
\begin{align}
\partial_0\phi_1 + a_0\phi_2 =& \frac{8\lambda^2 tr(r + \lambda)}{(t^2 + r^2 + \lambda^2)^2[t^2 + (r + \lambda)^2]} 
\nonumber\\
=& \partial_1\phi_2 - a_1\phi_1 , 
\\
\partial_1\phi_1 + a_1\phi_2 =& -\frac{4\lambda^2[t^2 - (r + \lambda)^2]r}{(t^2 + r^2 + \lambda^2)^2[t^2 + (r + \lambda)^2]}
\nonumber\\
=& -\partial_0\phi_2 + a_0\phi_1 .
\end{align}
Substituting the vortex solution $\phi$ into the  relationship between the scalar field $\phi(x^4,r)$ of the hyperbolic vortex and the scalar field $\Phi(x^4,x^3,\rho)$ of the hyperbolic magnetic monopole ($r^2 = \rho^2 + (x^3)^2$)
\begin{align}
  \| \Phi(x^4,x^3,\rho) \|^2 = \frac{\rho^2|\phi(x^4,r)|^2 + (x^3)^2}{4r^2} ,
\label{eq:15}
\end{align}
we obtain the equation obtained by directly solving the hyperbolic magnetic monopole equation:
\begin{align}
\| \Phi(x^4,x^3,\rho) \|^2 =& \frac{R^2}{(1 + R^2)^2} , \nonumber\\
 R^2 =& \frac{(x^3)^2 + t^2 + (\rho - \lambda)^2}{(x^3)^2 + t^2 + (\rho + \lambda)^2}.
\end{align}
From \eqref{eq:15}, the zero point of $\Phi$ occurs when $x^3 = 0$ and $\phi(x^4,r) = 0$ $(r = \rho)$. At $x^3 = 0$ (equatorial plane), $\| \Phi\|$ and $|\phi|$ are proportional:
\begin{align}
\| \Phi(t,x^3,\rho) \| = \frac{1}{2}|\phi(t,r)| \ (\text{on} \ x^3 = 0) .
\end{align}
By using another expression
\begin{align}
\| \Phi \|^2 = \frac{1}{4} - \frac{\lambda^2\rho^2}{[t^2 + (x^3)^2 + \rho^2 + \lambda^2]^2}.
\end{align}
it is easy to see that the boundary value of $\Phi$ is given by
\begin{align}
\| \Phi\| \to v =\frac{1}{2} \quad (\rho \to 0).
\end{align}


\section{Holographic principle: bulk and boundary correspondence
}

Magnetic monopoles first emerged from the Dirac work (Dirac(1931)\cite{Dirac}) on quantum electromagnetism, which is a singular solution of the Maxwell equation. The Bogomolny-Prasad-Sommerfield (BPS) magnetic monopole (Prasad-Sommerfield(1975)\cite{PS75}) is a generalization of the Dirac monopole to non-Abelian gauge theories, and is described by a smooth field with no singularities. How does the BPS monopole look at long distances, especially on the sphere at infinity? It is not necessarily true that it should look like the Dirac monopole.

When the metric of $\mathbb{R}^3$ is Euclidean, monopoles look exactly like Dirac monopoles on the sphere at infinity.
Similar to Dirac monopoles, all Euclidean monopoles look the same at infinity, except for their charge. However, the hyperbolic BPS monopoles, in contrast to Euclidean monopoles, take many different values on the sphere at infinity.


\begin{definition}[Anti de-Sitter space]
Consider a hypersurface in $(\nu + 3)$-dimensional space-time $\mathbb{R}^{2,\nu+3}$ with two times 
\begin{equation}
X_0^2 + X_{\nu + 2}^2 = X_1^2 + X_2^2 + \cdots + X_{\nu + 1}^2 - R^2 \label{eq:8c}
\end{equation}
given by the metric
\begin{equation}
(ds)^2 = (dX_0)^2 + (dX_{\nu + 2})^2 - (dX_1)^2 - \cdots -(dX_{\nu+1})^2
\end{equation}
which we call the $(\nu + 2)$-dimensional \textbf{anti de-Sitter space} and write $AdS_{\nu + 2}$, where the constant $R$ represents the radius (size of the space). 

To solve the constraints in \eqref{eq:8c}, we introduce new coordinates $(t,\rho,\Omega_1,\cdots,\Omega_{\nu + 1})$:
\begin{align}
X_0 =&R\cosh\rho\cos t , \ 
X_{\nu + 2} =  R\cosh\rho\sin t , \nonumber\\
X_k =& R\sinh\rho\Omega_k \ (k = 1,2,\cdots,\nu + 1).
\end{align}
Here, $(\Omega_1,\Omega_2,\cdots,\Omega_{\nu + 1})$ represent the coordinates of the $\nu$-dimensional unit sphere $S^\nu$. In this case, the metric is expressed as
\begin{align}
(ds)^2 = R^2(\cosh^2\rho(dt)^2 - (d\rho)^2 - \sinh^2\rho d\Omega_\nu^2)
\end{align}
Here, $d\Omega_\nu^2$ is the metric of the $\nu$-dimensional unit sphere. When time $x_0$ is Euclideanized, it becomes the hyperboloid $\mathbb{H}_{\nu + 2}$.
\end{definition}

The \textbf{Poincar\'{e} coordinate}\index{Poincar\'{e} coordinate}  $(x_0,z,x_1,\cdots,x_\nu)$ defined by
\begin{align}
X_0 &= \frac{z}{2}\left(1 + \frac{R^2 + \bm{x}^2 - x_0^2}{z^2}\right) , \nonumber\\
 X_j &=  R\frac{x_j}{z}  \ (j = 1,2,\cdots,\nu), 
\nonumber\\
X_{\nu + 1} &= \frac{z}{2}\left(1 - \frac{R^2 - \bm{x}^2 + x_0^2}{z^2}\right) , \nonumber\\
 X_{\nu + 2} &=  R\frac{x_0}{z}.
\end{align}
gives the metric:
\begin{align}
(ds)^2 = R^2\frac{(dz)^2 - (dx_0)^2 + \sum_{j = 1}^\nu(dx_j)^2}{z^2} .
\end{align}
Euclideanization is formally possible by replacing $x_0 \to ix_4$ to give the metric of the hyperbolic space $\mathbb{H}^{\nu+1}$:
\begin{align}
(ds)^2 = R^2\frac{(dz)^2 + (dx_4)^2 + \sum_{j = 1}^\nu(dx_j)^2}{z^2}.
\end{align}

\begin{example}[$AdS_3$]
Fore $\nu = 1$, $AdS_3$ is characterized by the metric with the constraint:
\begin{align}
&(ds)^2 = (dX_0)^2 + (dX_3)^2 - (dX_1)^2 - (dX_2)^2, \nonumber\\
&X_0^2 + X_3^2 = X_1^2 + X_2^2 - R^2.
\end{align}
The constraint is solved by introducing the new coordinates:
\begin{align}
&X_0 = R\cosh\rho \ \cos t , X_3 = R\cosh\rho  \ \sin t, \nonumber\\
&X_1 = R\sinh\rho \ \Omega_1 , X_2 = R\sinh\rho \ \Omega_2, 
\end{align}
and the metric is written as
\begin{align}
&(ds)^2 = R^2[\cosh^2\rho (dt)^2 - (d\rho)^2 - \sinh^2\rho \ d\Omega^2]. 
\end{align}
The Poincar\'{e} coordinate 
\begin{align}
&X_0 = \frac{z}{2}\left(1 + \frac{R^2 + x_1^2 -x_0^2}{z^2}\right) , X_1 = R\frac{x_1}{z}, \nonumber\\
&X_2 = \frac{z}{2}\left(1 - \frac{R^2 - x_1^2 +x_0^2}{z^2}\right) , X_3 = R\frac{x_0}{z}, 
\end{align}
leads to the metric:
\begin{align}
(ds)^2=R^2\frac{(dz)^2 - (dx_0)^2 + (dx_1)^2}{z^2}.
\end{align}
The metric of $\mathbb{H}^{3}$ is obtained by the Euclidean rotation:
\begin{align}
(ds)^2=R^2\frac{(dz)^2 + (dx_0)^2 + (dx_1)^2}{z^2}.
\end{align}

\end{example}%

\begin{example}[$AdS_2$]
Fore $\nu = 0$, $AdS_2$ is characterized by the metric with the constraint:
\begin{align}
&(ds)^2 = (dX_0)^2 + (dX_2)^2 - (dX_1)^2, \nonumber\\
&X_0^2 + X_2^2 = X_1^2 - R^2.
\end{align}
The constraint is solved by introducing the new coordinates:
\begin{align}
&X_0 = R\cosh\rho \cos t , X_1 = R\sinh\rho \Omega_1, 
\nonumber\\
&X_2 = R\cosh\rho \sin t,
\end{align}
and the metric is written as
\begin{align}
&(ds)^2 = R^2[\cosh^2\rho(dt)^2 - (d\rho)^2].
\end{align}
The Poincar\'{e} coordinate 
\begin{align}
&X_0 = \frac{z}{2}\left(1 + \frac{R^2 -x_0^2}{z^2}\right) , 
X_1 = \frac{z}{2}\left(1 - \frac{R^2 +x_0^2}{z^2}\right) , 
\nonumber\\
&X_2 = R\frac{x_0}{z}
\end{align}
leads to the metric:
\begin{align}
(ds)^2=R^2\frac{(dz)^2 - (dx_0)^2}{z^2}.
\end{align}
The metric of $\mathbb{H}^{2}$ is obtained by the Euclidean rotation:
\begin{align}
(ds)^2 = R^2\frac{(dz)^2 + (dx_4)^2}{z^2}.
\end{align}

\end{example}

Before presenting the general result, it is useful to see the difference between the hyperbolic and the Euclidean Dirac magnetic monopoles. 

\begin{example}[Hyperbolic Dirac monopoles]
Euclidean BPS monopoles all look like Dirac monopoles at infinity. Hyperbolic monopoles, on the other hand, detect the difference between BPS monopoles and Dirac monopoles. 
For the Euclidean metric, the magnetic field $\bm{B}(\bm{r})$ at $\bm{r}\in\mathbb{R}^3$ induced by a magnetic charge $Q_m$ at the origin  is given by
\begin{align}
\bm{B}(\bm{r}) = \frac{Q_m}{4\pi}\frac{\hat{r}}{r^2} = \frac{Q_m}{4\pi}\frac{\bm{r}}{r^3}.
\end{align}
More generally, the magnetic field $\bm{B}(\bm{r})$ at $\bm{r}$ generated  by a magnetic charge $Q_m$ at a point $\bm{a}\in\mathbb{R}^3$ is given by
\begin{align}
\bm{B}(\bm{r}) = \frac{Q_m}{4\pi}\frac{\bm{r} - \bm{a}}{|\bm{r}-\bm{a}|^3} .
\end{align}
Considering the point at infinity $r = |\bm{r}| \to \infty$, we have
\begin{align}
r^2\bm{B}(\bm{r}) = \frac{Q_m}{4\pi}\frac{r^2}{|\bm{r}-\bm{a}|^2}\frac{\bm{r} - \bm{a}}{|\bm{r}-\bm{a}|} \xrightarrow[r\to\infty]{} \frac{Q_m}{4\pi}\hat{\bm{r}} .
\end{align}
This implies that
\begin{align}
\bm{B}(\bm{r})\xrightarrow[r\to\infty]{} \frac{Q_m}{4\pi}\frac{\hat{\bm{r}}}{r^2} ,
\end{align}
 does not depend on $\bm{a}$. Therefore, monopoles appear to be symmetrically distributed on the infinitely distant sphere $S_\infty^2$. A collection of $N$ monopoles gives a magnetic field
\begin{align}
\bm{B}(\bm{r}) \to N\frac{Q_m}{4\pi}\frac{\hat{\bm{r}}}{r^2} .
\end{align}
From the above, Euclidean Dirac monopoles cannot be distinguished from a distance.
\end{example}

\begin{prop}[Hyperbolic Dirac monopole]
Hyperbolic Dirac monopoles are determined by their asymptotic values.
\end{prop}

\begin{proof}
This result was obtained by Norbury(1999)\cite{Norbury99}.
Let $d(x,a)$ be the hyperbolic distance between point $x\in\mathbb{H}^3$ and a given point $a$.
The magnetic field $\bm{B}$ at point $r = d(x,0)$ of a hyperbolic Dirac monopole generated by a magnetic charge $Q_m$ at the origin  is given by
\begin{align}
\bm{B}(x) = \frac{Q_m}{4\pi}\frac{\hat{\bm{r}}}{\sinh^2(r)} .
\end{align}
When a magnetic charge $Q_m$ is put at the point $x = a$, the generated magnetic field $\bm{B}$ at point $x$ of  the hyperbolic Dirac monopole  is given by
\begin{align}
\bm{B}(x) = \frac{Q_m}{4\pi}\frac{\hat{\nu}}{\sinh^2(d(x,a))}.
\end{align}
Here  $\hat{\nu}$ is the unit vector pointing away from $a$ along the geodesics connecting $x$ and $a$. The asymptotic value of each of these monopoles is given by the outward normal vector of $S_\infty^2$ scaled by 
\begin{align}
\lim_{r \to \infty}\frac{\sinh^2 r}{\sinh^2 d(x,a)}.
\end{align}
It uniquely determines the monopole, since it simply gives the symmetric measure transformed by a conformal transformation that brings $0$ of $S_\infty^2$ to $a$, induced by the isometry of $\mathbb{H}^3$.
\end{proof}

It was rigorously shown that \textbf{the holographic principle applies to hyperbolic magnetic monopoles in the hyperbolic space $\mathbb{H}^3$}. In contrast, it should be noted that \textbf{the holographic principle does not apply to magnetic monopoles in the flat Euclidean space $\mathbb{E}^3$, even though they are regarded as the infinite mass limit of magnetic monopoles in $\mathbb{H}^3$}. In flat space, the holographic images of any two magnetic monopoles with the same magnetic charge on $S^2$ are identical and indistinguishable.
(`t Hooft (1993)\cite{Hooft93}, Susskind (1995)\cite{Susskind95})

\begin{prop}[Bulk/boundary correspondence of $\mathbb{H}^3$]
A magnetic monopole on hyperbolic space $\mathbb{H}^3$ is completely determined by its asymptotic boundary value (the value of the boundary at infinity $\partial\mathbb{H}^3$), apart from the gauge equivalence. 
This situation is in sharp contrast with the Euclidean case in which all monopole have the same boundary values. 
\end{prop}

\begin{proof}
This result was proved by Peter Braam and David Austin (Braam and Austin (1990)\cite{BA90}).
A magnetic monopole on hyperbolic space can be regarded as an $S^1$-invariant instanton on $S^4$, therefore the ADHM construction of instantons can be used (Atiyah-Drinfeld-Hitchin-Manin(1978)\cite{ADHM78}). In fact, this proposition was proved using the ADHM construction in the presence of group actions, i.e., the $S^1$-equivariant ADHM construction.
See Appendix B and C for the details on the general ADHM construction and $S^1$-equivariant ADHM construction.

This proposition was first shown in the $SU(2)$ case by introducing \textit{discrete Nahm data} using algebraic geometry techniques when the ``mass'' $v := ||\Phi||_\infty$ takes half-integer values. 
The ADHM matrices satisfy certain equations and in the presence of the circle action these equations break up into difference equations for matrices of size specified by the monopole charge and labelled by an index specified by the mass. 
These equations are discretization of the Nahm equation, discrete Nahm equation (Nahm(1982)\cite{Nahm82}). 

Later, this approach was generalized to $SU(N)$ case by Michael Murray and Michael Singer (Murray and Singer (1996))\cite{MS96}.

See also Donaldson(1984)\cite{Donaldson84},
Hitchin(1983)\cite{Hitchin83}, 
Ward(1998)\cite{Ward98},
Jaffe and Taubes(1980)\cite{JT80} for the background and a recent review  
(Chan(2017))\cite{Chan17}.

\end{proof}

\begin{prop}[Holography of hyperbolic magnetic monopoles]
On the conformal boundary 2-sphere $S_\infty^2$ in the hyperbolic space $\mathbb{H}^3$, BPS hyperbolic magnetic monopoles take many different values.  
In other words, they can be distinguished from each other by their asymptotic values. 
A hyperbolic monopole in $\mathbb{H}^3$ are completely determined by its holographic image on the conformal boundary two-sphere $S_\infty^2$. 
This is different from the case of Euclidean monopoles.

\end{prop}

\begin{proof}
First, this result was shown by Austin and Braam (1990)\cite{BA90} for $SU(2)$ hyperbolic magnetic monopoles when the ``mass'' ($v$ is the norm of the scalar scalar field on $S_\infty^2$) $2v$ is an integer $(2v \in  \mathbb{Z})$. They used methods of algebraic geometry and introduced the discrete Nahm equation. 
Next, in the same case $(2v \in \mathbb{Z})$, it was generalized by Murray and Singer (1999)\cite{MS96} to $SU(N)$ hyperbolic magnetic monopoles. These monopoles are determined by their asymptotic values. 

Then, for any gauge group $G$ and $v$ being non-integer, Norbury (1999)\cite{Norbury99} showed it by using the method of holomorphic maps. When the  $2v$ is not necessarily an integer, Murray, Norbury and Singer(2001)\cite{MNS01} investigated it by using twistor theory and showed that the $SU(2)$ hyperbolic magnetic monopole is determined by its asymptotic value plus some extra information.
When $2v$ is any positive real number, the holography was proved by Norbury (2001)\cite{Norbury01}. 


\end{proof}

A proof of the holography for hyperbolic magnetic monopole is given based on the following framework. 

\begin{definition}[Scattering equation and spectral curve]
Consider the scattering equation involving monopoles $(A,\Phi)$ defined for a local section $s$ of a vector bundle $E$ along a geodesic with parameter $t$ in $\mathbb{R}^3$:
\begin{align}
(\partial_t^A-i\Phi)s = 0 \Leftrightarrow (D_t[\mathscr{A}]-i\Phi)s = 0.
\end{align}
This scattering equation is called the Hitchin equation (Hitchin(1982)\cite{Hitchin82}). 
Among the geodesics, a geodesic for which there exists an $L^2$ solution of this equation is called a \textbf{spectral curve}\index{spectral curve}. It forms a compact algebraic curve in $T\bm{CP}^1$.
\end{definition}

\begin{rem} 
The magnetic monopole and the spectral curve were investigated by e.g., Murry(1983)\cite{Murry83}, Hitchin and Murry(1988)\cite{HM88}, Norbury and Romao(2007)\cite{NR07}.
\end{rem}

\begin{definition}[$n$-point functions defined for a sequence of points on the boundary]
Consider an ordered set of points $\{z_1,z_2,\cdots,z_n\}$ on the conformal boundary 2-sphere $S_\infty^2$ of $\mathbb{H}^3$.
Then consider the set of geodesics in $\mathbb{H}^3$ that run from $z_1$ to $z_2$, those that run from $z_2$ to $z_3$, and so on, those that run from $z_n$ to $z_1$.
By solving the scattering equation involving monopoles $(A,\Phi)$ along the geodesics, we can measure the interaction between the $n$ points $\{z_1,\cdots,z_n\}$ on the conformal boundary $S_\infty^2$.

The $n$-point function $\braket{P_{z_1}\cdots P_{z_n}}$ defined for a given monopole and $n$-points $\{z_1,\cdots,z_n\}$ is a complex number corresponding to a sequence of geodesics, and is continuously differentiable with respect to its variables $(z_1,\cdots,z_n)$.
The $n$-point function $\braket{P_{z_1}\cdots P_{z_n}}$ is a notation that takes into account the construction of an algebra with an expectation value given by $\braket{P_{z_1}\cdots P_{z_n}}$.
Note that the three-dimensional hyperbolic space $\mathbb{H}^3$ with negative constant curvature is equivalent to the three-dimensional Euclidean anti-de Sitter space $AdS_3$ ($\mathbb{H}^3 = AdS_3^E)$.

\end{definition}

\begin{definition}[$n$-point function on the boundary]
Given two points $z_1, z_2$ on $S_\infty^2$, define $s_+(t)$ as the solution of the scattering equation that decays at $t \to \infty$ along the geodesic connecting them, and $r_+(t)$ as the solution that decays along the same geodesic in the opposite direction, i.e., the solution that decays at $t \to -\infty$:
\begin{align}
(\partial_t^A-i\Phi)s_+(t) = 0 \ , \ (\partial_t^A+i\Phi)r_+(t) = 0.
\end{align}
The inner product $(r(t),s(t))$ of any two solutions of these scattering equations does not depend on $t$: $\partial_t(r(t),s(t))=0.$
In fact, it is as follows:
\begin{align}
\partial_t(r(t),s(t))=& ((\partial_t^A+i\Phi)r(t),s(t)) \nonumber\\
&+(r(t),(\partial_t^A-i\Phi)s(t))= 0.
\end{align}

If we normalize $r_+(t)$ and $s_+(t)$ by
\begin{align}
\lim_{t \to \infty}e^{mt}||s_+|| = 1 \ , \ \lim_{t \to -\infty}e^{-mt}||r_+|| = 1.
\end{align}
The decay solution is well defined except for the phase, and $|(r_+,s_+)|^2$ depends only on the geodesic and the monopole $(A,\Phi)$. Therefore, we define the two-point function $\braket{P_{z_1}P_{z_2}}$ as \begin{align}
\braket{P_{z_1}P_{z_2}}:=|(r_+,s_+)|^2
\end{align}

\end{definition}

Similarly, the n-point function $\braket{P_{z_1}\cdots P_{z_n}}$ is defined using the decay solution of the scattering equation along the geodesic that runs between the ordered $n$-points on $S_\infty^2$:
\begin{align}
\braket{P_{z_1}\cdots P_{z_n}}:=(r_{12},s_{12})(r_{23},s_{23})\cdots(r_{n1},s_{n1}).
\end{align}
This depends only on $(A,\Phi)$ and the oriented geodesic that passes through $z_1,z_2,\cdots,z_n,z_1$ in order. In fact, if we set $n = 2$,
\begin{align}
\braket{P_{z_1}P_{z_2}} = (r_{12},s_{12})(r_{21},s_{21}) 
= |(r_{12},s_{12})|^2,
\end{align}
which is consistent with the two-point function already mentioned. 
Here, $r_{jk},s_{jk}$ represent the solutions $r_+,s_+$ along the geodesic that runs from $z_j$ to $z_k$. The phase is for each $r_{j,j+1},s_{j-1,j}$
\begin{align}
\lim_{t \to \infty}e^{mt}s_{j-1,j} =c \lim_{t \to -\infty}e^{-mt}r_{j,j+1} \ (c \in \mathbb{C}).
\end{align}


The calculation of $n$-point functions using solutions to scattering equations along geodesics in $\mathbb{H}^3$ is similar to the approximation to the calculation of correlation functions using path integrals that appear in the AdS/CFT correspondence. 
(Maldacena(1998)\cite{Maldacena98}, Witten (1998)\cite{Witten98}, Aharony et al. (2000)\cite{Aharony00})

However, the proof is too complicated to give a review here, although we hope to give the supplementary materials in near future. 
Rather, we use this result to calculate the Wilson loop average to show quark confinement in the following form: 

\begin{prop}[Abelian dominance  and magnetic monopole dominance on the conformal boundary $\partial\mathbb{H}^3$]\label{prop:monopole-dominance}
On the conformal boundary $\partial\mathbb{H}^3 \simeq S^2$ of the upper half-space coordinates of  $\mathbb{H}^3(\rho,x^3,x^4)$, that is, $\rho = 0$: $x^4$-$x^3$ plane, the $SU(2)$ Yang-Mills field and the $SU(2)$ scalar field converges as $\rho \to 0$ to
\begin{align}
\mathscr{A}_4^G(\rho,x^3,x^4) &= \frac{\sigma_3}{2}\frac{x^3}{r}a_t \to \frac{\sigma_3}{2}a_t(x^4,x^3) , \nonumber\\
\mathscr{A}_3^G(\rho,x^3,x^4) &= \frac{\sigma_3}{2}\frac{(x^3)^2}{r^2}a_r \to \frac{\sigma_3}{2}a_r(x^4,x^3) , \nonumber\\
\mathscr{A}_\rho^G(\rho,x^3,x^4) &= \frac{\sigma_1}{2}\frac{(x^3)^2}{r^3}\phi_1 + \frac{\sigma_2}{2}\frac{x^3}{r^2}(1 + \phi_2) \nonumber\\
&\to \frac{\sigma_1}{2}\frac{1}{r}\phi_1(x^4,x^3) + \frac{\sigma_2}{2}\frac{1}{r} [1 + \phi_2(x^4,x^3)], \nonumber\\
\Phi(\rho,x^3,x^4) &\to \frac{\sigma_3}{2}(-1) \ \left( ||\Phi || \to v=\frac12 \right) ,
\end{align}
where $r = \sqrt{\rho^2 + (x^3)^2} \to |x^3|$ for $\rho \to 0$. 
Therefore, the $SU(2)$ gauge field $\mathscr{A}_4^G(\rho,x^3,x^4), \mathscr{A}_3^G(\rho,x^3,x^4)$ on the boundary $\rho=0$ has only the diagonal components $a_t(x^4,x^3), a_r(x^4,x^3)$, while the gauge field $\mathscr{A}_\rho^G(\rho,x^3,x^4)$ is dominated by the off-diagonal components $\frac{\sigma_1}{2}\frac{1}{r}\phi_1(x^4,x^3) + \frac{\sigma_2}{2}\frac{1}{r} [1 + \phi_2(x^4,x^3)]$.
Thus, the $SU(2)$ field strength on the boundary has only the maximal torus $U(1)$ component: 
\begin{align}
\mathscr{F}_{43}^G(\rho,x^3,x^4) :=&\partial_4\mathscr{A}_3^G - \partial_3\mathscr{A}_4^G - ig[\mathscr{A}_4^G, \mathscr{A}_3^G] \nonumber\\
\to&\frac{\sigma_3}{2}\frac{(x^3)^2}{r^2}\partial_4a_r - \frac{\sigma_3}{2}\partial_3\left(\frac{x^3}{r}a_t\right) \nonumber\\
=&\frac{\sigma_3}{2}(\partial_4a_r - \partial_ra_t) 
=\frac{\sigma_3}{2} F_{4r}(x^4,x^3).
\end{align}
This fact is regarded as the (infrared) Abelian dominance and the the magnetic monopole dominance.  

\end{prop}

\begin{proof}
This follows immediately from the result of Proposition~\ref{prop:monopole-vortex}.
\end{proof}

\section{CFtHW Ansatz for instantons and the superpotential}

\begin{definition}[self-dual equation]
In 4-dimensional Euclidean space, the self-dual equation for Yang-Mills field strength $\mathscr{F}_{\mu \nu} (x) = T_A \mathscr{F}_{\mu \nu}^A (x) \ (\mu , \nu = 1 , 2 , 3 , 4)$ is given by
\begin{align}
 \mathscr{F}_{\mu \nu} (x) = \pm {}^* \mathscr{F}_{\mu \nu} (x) , \quad
 {}^* \mathscr{F}_{\mu \nu} (x) := \frac{1}{2} \varepsilon_{\mu \nu \rho \sigma} \mathscr{F}_{\rho \sigma} (x) .
\label{eq:Tss-1}
\end{align}
Here ${}^* \mathscr{F}_{\mu \nu}$ is the Hodge dual  of $\mathscr{F}_{\mu \nu}$.
The right-hand side with a plus $+$ is called a \textbf{self-dual equation}\index{self-dual equation}, and the negative $-$ is called an \textbf{anti-self-dual equation}\index{anti-self-dual equation}.

\end{definition}

\begin{prop}[Specific form of self-dual solution]
In the following, we restrict the gauge group to $SU(2)$.
To find the solution to the self-dual equation, we adopt the Ansatz called the Corrigan-Fairlie-'t~Hooft-Wilczek (CFtHW)  Ansatz \cite{CFtHW}:
\begin{align}
\mathscr{A}_\mu (x) = T_A \mathscr{A}_\mu^A (x) =T_A \bar{\eta}^A_{\mu \nu} \partial_\nu \ln \Xi (x) .
\label{eq:Tss-2}
\end{align}
where $\Xi$ is called the \textbf{superpotential}. 
Here $\bar{\eta}^A_{\mu \nu} \ (A = 1 , 2 , 3)$ is the 't Hooft symbol defined by
\begin{align}
\bar{\eta}^A_{\mu \nu} = \varepsilon_{4 A \mu \nu} - \delta_{A\mu } \delta_{\nu 4} + \delta_{\mu 4} \delta_{A \nu} .
\end{align}

When $\Xi (x)$ is nonsingular, $\Xi (x)$ is a constant and the gauge field is trivial: $\mathscr{A}_\mu (x) \equiv 0$.
On the other hand, if $\Xi (x)$ has a singular point, the solution is given by 
\begin{align}
\Xi (x) = 1 + \sum_{n = 1}^N \frac{\lambda_n^2}{(x - a_n)^2} .
\label{eq:Tss-5}
\end{align}
This solution can be written in a more general form (called the JNR form):
\begin{align}
\Xi (x) = \sum_{n = 0}^{N} \frac{\lambda_n^2}{(x - a_n)^2} ,
\label{eq:Tss-6}
\end{align}
where $a_n = (a_{n \mu})$ is any 4-real vector, and $\lambda_n \ (n = 0, 1 , 2 , \cdots , N)$ is any real constant. (\ref{eq:Tss-6}) reduces to (\ref{eq:Tss-5}) by fixing the ratio and taking the limits $a_{0} \rightarrow \ \infty$ and $\lambda_{0} \rightarrow \ \infty$: $\lambda_{0}^2 / a_{0}^2 = 1$.
\end{prop}

\begin{proof}
For the gauge field of the form \eqref{eq:Tss-2} to be a solution of the self-dual equation (\ref{eq:Tss-1}), the superpotential $\Xi$ must satisfy the following equation:
\begin{align}
\Xi(x)^{-1} \Box \Xi(x) = 0 , \quad \Box := \partial_\mu \partial_\mu .
\label{eq:Tss-4}
\end{align}
Here, $\Box$ is the (four-dimensional) Laplacian defined by $\Box := \partial_\mu \partial_\mu$.
See e.g., Rajaraman(1989)\cite{Rajaraman89} for the details of the calculations.
\end{proof}

\begin{prop}[superpotential for the vortex equation]

\noindent
(1) 
Let $\Xi_4$ be the superpotential on the Euclidean space $\mathbb{E}^4$.
Let $\Xi_2$ be the superpotential on the hyperbolic upper half-plane $\mathbb{H}^2$.
The two superpotentials follow the relationship:
\begin{align}
\Xi_4=r^{-1}\Xi_2 . 
\end{align}

\noindent
(2) 
Using the Witten Ansatz, the $U(1)$ gauge field $a_\mu = (a_t,a_r)$ is obtained from the superpotential:
\begin{align}
a_t &= \frac{\partial}{\partial r}\log \Xi_4 = \frac{\partial}{\partial r}\log \Xi_2 - \frac{1}{r}, \nonumber\\
a_r &= -\frac{\partial}{\partial t}\log \Xi_4 = -\frac{\partial}{\partial t}\log \Xi_2.
\end{align}
Then the field strength of the $U(1)$ gauge field reads
\begin{align}
F_A 
=&(\partial_r a_t - \partial_t a_r) dt \wedge dr \nonumber\\
=& \left[\left(\frac{\partial^2}{\partial t^2} + \frac{\partial^2}{\partial r^2}\right)\log \Xi_2 + \frac{1}{r^2}\right]dt \wedge dr .
\label{eq:8}
\end{align}
The complex scalar field $\phi = \phi_1 + i\phi_2$ is obtained from the superpotential:
\begin{align}
\phi_1 &= -r \frac{\partial}{\partial t}\log \Xi_4 = -r \frac{\partial}{\partial t}\log \Xi_2 , \nonumber\\
\phi_2 &= -1-r \frac{\partial}{\partial r}\log \Xi_4 = -r \frac{\partial}{\partial r}\log \Xi_2 .
\end{align}

\noindent
(3) 
Introducing the complex number $z := t + ir$,  the $U(1)$ connection $D_\mu = \partial_\mu + ia_\mu$ is written as
\begin{align}
\bar{\partial}_A = \bar{\partial} + \bar{\partial}\log \Xi_2 +\frac{1}{z - \bar{z}} . 
\end{align}
The complex scalar field $\phi$ is written as
\begin{align}
\phi = i(z - \bar{z})\frac{\partial}{\partial z}\log \Xi_2.
\end{align}
A necessary and sufficient condition for these to be a solution of the \textbf{vortex equation}\index{vortex equation}
\begin{align}
&\bar{\partial}_A \phi =0 ,\\
& 
F_A = *_{\mathbb{H}^2}(1 - |\phi|^2) \ (*_{\mathbb{H}^2}\mathbf{1}:=r^{-2}dt  \wedge dr),
\end{align}
is that $\Xi_2$ is \textbf{harmonic function}\index{harmonic function}
:
\begin{align}
\Delta_h \Xi_2 \equiv -r^2\left(\frac{\partial^2}{\partial t^2} + \frac{\partial^2}{\partial r^2}\right) \Xi_2 = 0.
\end{align}

\end{prop}

\begin{proof}
This result was derived by Landweber(2005)\cite{Landweber05}. 

\noindent
(1) 
Suppose that the 4-dimensional superpotential $\Xi_4$ depends only on $t$ and $r$, and does not depend on $\theta$ or $\varphi$. 
Let the 3-dimensional Laplacian be written in polar coordinates $(r, \theta, \varphi)$. 
Then we can write 
\begin{align}
\Delta_3 =& \frac{\partial^2}{\partial x^2} + \frac{\partial^2}{\partial y^2} + \frac{\partial^2}{\partial z^2} = \frac{1}{r}\frac{\partial^2}{\partial r^2}r + \frac{\Lambda(\theta,\varphi)}{r^2}, 
\nonumber\\
\Lambda(\theta,\varphi):=&\frac{1}{\sin \theta}\frac{\partial}{\partial \theta}\left(\sin \theta\frac{\partial}{\partial \theta}\right) + \frac{1}{\sin^2 \theta}\frac{\partial^2}{\partial \varphi^2},
\end{align}
which means 
\begin{align}
& \Box\Xi_4(t,r) = \left(\frac{\partial^2}{\partial t^2}+\Delta_3\right)\Xi_4(t,r) \nonumber\\
=& \frac{1}{r}\left(\frac{\partial^2}{\partial t^2} + \frac{\partial^2}{\partial r^2}\right)r\Xi_4(t,r) 
 + \frac{1}{r^2}\Lambda(\theta,\varphi)\Xi_4(t,r) \nonumber\\
=&\frac{1}{r}\left(\frac{\partial^2}{\partial t^2} + \frac{\partial^2}{\partial r^2}\right)\Xi_2(t,r) .
\end{align}
Thus we obtain
\begin{align}
\Xi_2(t,r) = r\Xi_4(t,r), \quad
\left(\frac{\partial^2}{\partial t^2} + \frac{\partial^2}{\partial r^2}\right)\Xi_2(t,r) =& 0 .
\end{align}

\noindent
(2) 
The \textbf{inverse Witten transformation} is given by 
\begin{align}
a_0(r,t) =& \frac{x^A}{r}\mathscr{A}_4^A(x) \ , \ a_1(r,t) = \frac{x^j}{r}\frac{x^A}{r}\mathscr{A}_j^A(x), 
\nonumber\\
\phi_1(r,t) =& \frac{\delta^{jA}r^2 - x^jx^A}{2r}\mathscr{A}_j^A(x) , 
\nonumber\\
 \phi_2(r,t) =& -\epsilon_{Ajk}\frac{x^k}{2}\mathscr{A}_j^A(x) - 1.
 \label{field2}
\end{align}
The Ansatz for the instanton reads 
\begin{align}
& \mathscr{A}_\mu^A(x) = \bar{\eta}_{\mu\nu}^A\partial_\nu \ln \Xi_4(x) , 
\nonumber\\ & \bar{\eta}_{\mu\nu}^A = \epsilon_{4A\mu\nu} - \delta_{A\mu}\delta_{\nu 4} + \delta_{A\nu}\delta_{\mu 4} \nonumber\\
\Rightarrow &
\begin{cases}
 \mathscr{A}_4^A(x) = \bar{\eta}_{4\nu}^A\partial_\nu\ln\Xi_4(x) 
 = \partial_A\ln \Xi_4(x), \\
 \mathscr{A}_j^A(x) = \epsilon_{Ajk}\partial_k\ln\Xi_4(x) - \delta_{Aj}\partial_4 \ln \Xi_4(x).
\end{cases}
\label{field4}
\end{align}
Therefore, the dimensionally reduced fields are obtained 
by substituting (\ref{field4}) into (\ref{field2}):
\begin{align}
a_0(r,t) =& \frac{x^A}{r}\partial_A\ln \Xi_4(x) = \partial_r \ln \Xi_4(x), \nonumber\\
a_1(r,t) =& -\frac{x^j}{r}\frac{x^A}{r}\delta_{Aj}\partial_4\ln\Xi_4(x) = -\partial_4\ln\Xi_4(x), \nonumber\\
\phi_1(r,t) =& \frac{-\delta^{jA}\delta_{Aj}r^2 + x^jx^A\delta_{Aj}}{2r}\partial_4\ln\Xi_4(x) \nonumber\\
=& -r\partial_4\ln\Xi_4(x), 
\nonumber\\
\phi_2(r,t) =& -1-\epsilon_{Ajk}\frac{x^k}{2}\epsilon_{Aj\ell}\partial_\ell \ln\Xi_4(x) \nonumber\\
=& -1-x^k\partial_k \ln \Xi_4(x) 
= -1 -r\partial_r\ln\Xi_4(x).
\end{align}

\noindent
(3) 
To move to the complex representation, we find the relations $z,\bar{z}$ and $t,r$:
\begin{align}
z:=& t + ir \ , \ \bar{z} = t -ir, \nonumber\\
\Rightarrow t =& \frac{z + \bar{z}}{2} \ , \ r = \frac{z - \bar{z}}{2i} \ , \ \frac{1}{r} = \frac{2i}{z - \bar{z}}.
\end{align}
The differential operators are related as 
\begin{align}
\partial_z =& \frac{\partial}{\partial z} 
= \frac{\partial t}{\partial z}\frac{\partial}{\partial t} + \frac{\partial r}{\partial z}\frac{\partial}{\partial r} 
= \frac{1}{2}\left(\frac{\partial}{\partial t} - i\frac{\partial}{\partial r}\right), 
\nonumber\\
\partial_{\bar{z}} =& \frac{\partial}{\partial \bar{z}} 
= \frac{\partial t}{\partial \bar{z}}\frac{\partial}{\partial t} + \frac{\partial r}{\partial \bar{z}}\frac{\partial}{\partial r} 
= \frac{1}{2}\left(\frac{\partial}{\partial t} + i\frac{\partial}{\partial r}\right).
\end{align}
For the complex scalar field, 
\begin{align}
\phi_1 =& -r\frac{\partial}{\partial t}\ln \Xi_2 \ , \ \phi_2 = -r\frac{\partial}{\partial r}\ln \Xi_2 , \nonumber\\
\Rightarrow \phi =& \phi_1 + i\phi_2 = -r\left(\frac{\partial}{\partial t} + i\frac{\partial}{\partial r}\right)\ln \Xi_2 
\nonumber\\
=& i(z - \bar{z})\partial_{\bar{z}}\ln \Xi_2 .
\end{align}
For the $U(1)$ gauge field, 
\begin{align}
a_t \equiv a_0=& \frac{\partial}{\partial r}\ln \Xi_2 - \frac{1}{r} \ , \ a_r \equiv a_1
= -\frac{\partial}{\partial t}\ln \Xi_2, 
\nonumber\\
\Rightarrow a_r - ia_t 
=& -2\partial_{\bar{z}}\ln\Xi_2 + \frac{i}{r}.
\end{align}
For the covariant derivatives, 
\begin{align}
(D_t + iD_r)\phi =& (\partial_t \pm ia_t)\phi + i(\partial_r \pm ia_r)\phi \nonumber\\
=& (\partial_t + i\partial_r)\phi \mp (a_r - ia_t)\phi \nonumber\\
=&2\partial_{\bar{z}}\phi \pm \left(2\partial_{\bar{z}}\ln\Xi_2 - i\frac{1}{r}\right)\phi .
\end{align}
Therefore, we have 
\begin{align}
\frac{1}{2}(D_t + iD_r)\phi = \left(\partial_{\bar{z}} \pm \partial_{\bar{z}}\ln\Xi_2 \pm \frac{1}{z - \bar{z}}\right)\phi .
\end{align}

\end{proof}

\begin{prop}[1st Chern class]
The Chern class $c_1$ is given by
\begin{align}
c_1 = \frac{1}{2\pi}\int_{\mathbb{H}^2} F_A
= \frac{1}{2\pi}\int_{\mathbb{H}^2} *_{\mathbb{H}^2}(1 - |\phi|^2)
\end{align}
Using the superpotential $\Xi_2$, it is written as 
\begin{align}
c_1=& \frac{1}{2\pi}\int_{\mathbb{R}_+^2} (1 - \Delta_h\log \Xi_2)r^{-2}dt \wedge dr \nonumber\\
=& \frac{1}{2\pi}\int_{\mathbb{R}_+^2} \left(\frac{1}{r^2} - 4\left|\frac{\partial}{\partial z}\log \Xi_2\right|^2 \right)dt \wedge dr .
\end{align}

\end{prop}

\begin{proof}
This result was derived by Landweber(2005)\cite{Landweber05}. 

\end{proof}

\begin{example}[1-vortex on the upper half model]
The superpotential for a 1-instanton in 4 dimensions is
\begin{align}
\Xi_4 = 1 + \frac{\lambda^2}{|x|^2} = 1 + \frac{\lambda^2}{t^2 + r^2}
\end{align}

In general, if we use translation $x\to x-a$ to indicate the position and dilation $x \to \lambda^{-1}$ to indicate the size, then $x \to \lambda^{-1}(x-a)$ gives us $\Xi_4 = 1 + \frac{\lambda^2}{|x-a|^2}$.

For the four-dimensional superpotential $\Xi_4$ for 1-instanton:
\begin{align}
\Xi_4 = 1 + \frac{\lambda^2}{x^2} = \frac{x^2 + \lambda^2}{x^2} .
\end{align}
the two-dimensional superpotential $\Xi_2$ is given by
\begin{align}
\Xi_2 =& r\Xi_4 = r + \frac{r\lambda^2}{t^2 + r^2} =r\frac{t^2 + r^2 + \lambda^2}{t^2 + r^2} 
\nonumber\\
\Rightarrow  \ln \Xi_2 =& \ln r + \ln(t^2 + r^2 + \lambda^2) -\ln(t^2 + r^2).
\end{align}
Therefore, the derivatives are given by
\begin{align}
&\frac{\partial}{\partial r}\ln\Xi_2 
= \frac{1}{r} - \frac{2r\lambda^2}{(t^2 + r^2 + \lambda^2)(t^2 + r^2)}, 
\nonumber\\
&\frac{\partial}{\partial t}\ln \Xi_2 
= \frac{-2t\lambda^2}{(t^2 + r^2)(t^2 + r^2 + \lambda^2)}.
\end{align}

The complex scalar field is written as
\begin{align}
\phi_1 =& -r\frac{\partial}{\partial t}\ln \Xi_2 = \frac{2\lambda^2 tr}{(t^2 + r^2)(t^2 + r^2 + \lambda^2)} , \nonumber\\
\phi_2 =& -r\frac{\partial}{\partial r}\ln \Xi_2 = -1 + \frac{2\lambda^2 r^2}{(t^2 + r^2)(t^2 + r^2 + \lambda^2)} , \nonumber\\
|\phi|^2 =& \phi_1^2 + \phi_2^2 = \frac{r^4 + 2r^2(t^2 - \lambda^2) + (t^2 + \lambda^2)^2}{(t^2 + r^2 + \lambda^2)^2} \nonumber\\
=& \frac{[t^2 + (r -\lambda)^2][t^2 + (r + \lambda)^2]}{(t^2 + r^2 + \lambda^2)^2} , \nonumber\\
1 - |\phi|^2 =& 1 - (\phi_1^2 + \phi_2^2) = \frac{4r^2\lambda^2}{(t^2 + r^2 + \lambda^2)^2} .
\end{align}
The center of the vortex is given by the zero of $\phi$: $\phi(r,t) = 0$ is on the $(t,r) = (0,\lambda)$.

The $U(1)$ gauge field is written as
\begin{align}
a_t =& \frac{\partial}{\partial r}\ln \Xi_2 - \frac{1}{r} = - \frac{2\lambda^2 r}{(t^2 + r^2)(t^2 + r^2 + \lambda^2)}, \\
a_r =& -\frac{\partial}{\partial t}\ln \Xi_2 = \frac{2\lambda^2 t}{(t^2 + r^2)(t^2 + r^2 + \lambda^2)}. 
\end{align}
Then the field strength of the $U(1)$ gauge field reads
\begin{align}
F_{r4} := \partial_ra_t - \partial_ta_r = \frac{4\lambda^2}{(t^2 + r^2 + \lambda^2)^2} = \frac{4\lambda^2}{(|z|^2 + \lambda^2)^2}.
\end{align}
Therefore, they indeed satisfy the vortex equation:
\begin{align}
F_{r4} = \frac{1}{r^2}(1 - |\phi|^2).
\end{align}

From this, the superpotential for the hyperbolic vortex on $\mathbb{H}^2$ is $z = t + ir$
\begin{align}
\Xi_2 =& r\Xi_4 = r + \frac{\lambda^2 r}{t^2 + r^2} \nonumber\\
=& \operatorname{Im}\left(z - \frac{\lambda^2}{z}\right) = \frac{(z - \bar{z})(\lambda^2 + z\bar{z})}{2iz\bar{z}} .
\end{align}
Using
\begin{align}
& \phi
= i\frac{\bar{z}(\lambda^2 + z^2)}{z(\lambda^2 + z\bar{z})} , \nonumber\\ \quad 
&|\phi|^2 = \frac{(\lambda^2 + z^2)(\lambda^2 + \bar{z}^2)}{(\lambda^2 + z\bar{z})(\lambda^2 + z\bar{z})} \to 1 \ (r \to 0),
\end{align}
$\phi$ has a zero $\phi = 0$ at $z = i$ $(t=0,r=1)$.
The complex scalar field $\phi$ satisfies the boundary condition $\phi \to i$ $(\phi_1 \to 0,\phi_2 \to 1)$ when $z$ approaches the real axis ($r \to 0$).

The field strength of the $U(1)$ gauge field is
\begin{align}
F_A = \frac{4\lambda^2}{(t^2+r^2+\lambda^2 )^2}dt \wedge dr = F_{01}dt \wedge dr ,
\end{align}
and the topological charge or the vortex number is obtained:
\begin{align}
N_v = \frac{1}{2\pi}\int^{+\infty}_{-\infty}dt \ \int^{\infty}_{0}dr \ F_{01} = c_1=1.
\label{Nv}
\end{align}

\end{example}

Next, we move on to another complex variable. Let $z$ be the  complex coordinate of the upper half-plane $\mathbb{H}^2$, and $w$ be the complex coordinate of the unit disk $D^2$. $z$ and $w$ are related by a conformal transformation or \textbf{Cayley transformation}\index{Cayley transformation}:
\begin{align}
w = \frac{i -z}{i + z} \Leftrightarrow z = i\frac{1 -w}{1 + w}
\label{eq:9cc}
\end{align}

See Figure\ref{H2_conf_transf}. 
This mapping maps the upper half-plane $\mathbb{H}^2$ to the interior of the unit disk $D^2$. 
This disk $D^2$ is called \textbf{Poincar\'{e} disk}\index{Poincar\'{e} disk}.
The vortex equation is invariant under this coordinate transformation \eqref{eq:9cc}. The equation $\bar{\partial}\phi = 0$ holds because the transformation \eqref{eq:9cc} is regular. The equation $*_hF_A = 1 - |\phi|^2$ is invariant because it is a relation between scalar quantities that are invariant under coordinate transformations.

In the Poincar\'{e} disk model, the first Chern number is written as 
\begin{align}
c_1(A) = \frac{1}{2\pi}\int_{D^2} dS \left(\frac{4}{(1 - |w|^2)^2} - 4\left|\frac{\partial}{\partial w}\log \Xi_2\right|^2 \right)
\end{align}
where $dS$ is a surface element on the disk.

\begin{example} [1-vortex on the disk model]

For the hyperbolic vortex obtained by dimensional reduction from the $c_2 = 1$ instanton, $\Xi_2$ can be written using $w$:
\begin{align}
\boxed{\Xi_2 = 2\frac{1 - |w|^4}{|1 - w^2|^2}}.
\end{align}
Using this, the $U(1)$ connection and the scalar field are written as
\begin{align}
\bar{\partial}_A = \bar{\partial} + \frac{d\bar{w}}{1 + |w|^2}\frac{1 + w}{1 - w} , \ \phi= -\frac{2iw}{1 + |w|^2}\frac{1 - \bar{w}}{1 - w}
\end{align}
The scalar field $\phi$ has a first-order zero only at the origin $w = 0$.

We obtain the first Chern number for the hyperbolic vortex obtained by dimensional reduction from an instanton with the second Chern number $c_2(A)=1$:
\begin{align}
c_1(A)=1
\end{align}

\end{example}

\section{Singular instantons}

\begin{prop}[Singular gauge transformation]
Consider the gauge transformation $\mathscr{A}^U_\mu(x)$ with $U(x)\in SU(2)$ of gauge field $\mathscr{A}_\mu(x)$:
\begin{align}
\mathscr{A}_\mu &\to \mathscr{A}^U_\mu = U\mathscr{A}_\mu U^\dagger + iU\partial_\mu U^\dagger , \nonumber\\
U &= \exp\left(i\frac{\bm{v}}{2}\varphi\right) = \exp\left(i\hat{\bm{v}}|\bm{v}|\frac{\varphi}{2}\right) \nonumber\\
&=\cos\frac{|\bm{v}|\varphi}{2}\hat{\bm{1}} + i\hat{\bm{v}}\sin\frac{|\bm{v}|\varphi}{2},
\end{align}
where $\bm{v}$ is $su(2)$-valued function, $|\bm{v}|$ denotes its magnitude, and $\hat{\bm{v}}$ denotes the direction:
\begin{align}
\bm{v} = |\bm{v}|\hat{\bm{v}} , \ |\bm{v}|:=\sqrt{\frac{1}{2}\operatorname{tr}(\bm{v}^2)}>0 \ , \ \hat{\bm{v}}:=\frac{\bm{v}}{|\bm{v}|} \in su(2).
\end{align}

Thus we obtain the result:
\begin{align}
\mathscr{A}_\mu^U =& \cos(|\bm{v}|\varphi)\mathscr{A}_\mu + [1 - \cos(|\bm{v}|\varphi)](\mathscr{A}_\mu\cdot\hat{\bm{v}})\hat{\bm{v}} 
\nonumber\\
&+ \frac{1}{2}\sin(|\bm{v}|\varphi)(\mathscr{A}_\mu\times\hat{\bm{v}}) + iU\partial_\mu U^\dagger.
\end{align}
with
\begin{align}
iU\partial_\mu U^\dagger = 
\begin{cases}
\bm{v} & (\mu = \varphi) \\
0 & (\mu \neq \varphi)
\end{cases} .
\end{align}

\end{prop}

\begin{proof}
This result was obtained by Nash (1986)\cite{Nash86}. 

\end{proof}

In the case of $|\bm{v}| \notin \mathbb{Z}$, At $\varphi =0$, the gauge field starts from a certain $\mathscr{A}_\mu$. After the gauge transformation $\varphi = 2\pi$ (going around the axis $x^3$), it does not return to the initial value, and is no longer a \textit{single-valued} function. 
Therefore, if we request the single-valuedness of rotating $2\pi$ around $\hat{\bm{v}}$ and returning back, it must be $|\bm{v}| \in \mathbb{Z}$.  For  $|\bm{v}| \notin \mathbb{Z}$, $\mathscr{A}_\mu$ has a \textit{branch singularity} on the axis $\mathbb{R}^2$. 
Even in the case of $|\bm{v}| \notin \mathbb{Z}$, the hyperbolic magnetic monopole $(A,\Phi)$ exists, but the corresponding ordinary instanton does not exist. Therefore, hyperbolic magnetic monopoles are wider than $S^1$ (axisymmetric) symmetric instantons.

If such an instanton exists, it would have a non-integer $c_2 \ (c_2 \notin \mathbb{Z})$. In fact, such an example was constructed by Forgacs-Horvath-Palla (1981)\cite{FHP81}.  
Instantons, which generally correspond to the hyperbolic magnetic monopoles of $|\bm{v}| \notin \mathbb{Z}$, have some kind of singularity. 
Along the ``axis''$\mathbb{R}^2$, $|\Phi| = |\Phi|_\infty = C - 1$. In order for instanton to have a finite action, there must be an element $g$ of the gauge group such that $\mathscr{A}_\mu \to g^{-1}\partial_\mu g$ at infinity. 
This suggests that there exists a gauge that can reach $\mathscr{A}_\varphi \to 0$ at infinity when approaching infinity along $\mathbb{R}^2$.

\begin{example}[singular 1-instantons]
To be concrete, furthermore, to investigate singular instantons, we consider a general superpotential:
\begin{align}
\boxed{\Xi_2 = 2\frac{1 - |w|^{2c}}{|1 - w^c|^2}}, \ (c \in \mathbb{R}, \ c \not= 0) .
\end{align}
This is a harmonic function. Therefore it gives an instanton solution. 
In fact, we find
\begin{align}
& \frac{\partial}{\partial \bar{w}} \Xi_2 
= \frac{\partial}{\partial \bar{w}} 2\frac{1 - w^c\bar{w}^c}{(1 - w^c)(1 - \bar{w}^c)}
= 2 \frac{c\bar{w}^{c-1}}{(1 - \bar{w}^c)^2} \nonumber\\
\Rightarrow & \frac{\partial}{\partial w} \frac{\partial}{\partial \bar{w}} \Xi_2 = 0 .
\end{align}
$c=1$ is the vacuum solution.
$c=2$ is the usual case of a one-instanton.
When $c$ is an integer $c \in \mathbb{Z}$, it is fine, but when $c$ is a non-integer $c \notin \mathbb{Z}$, $\Xi_2$ is not well defined over the entire disk. In this case, we restrict it to a simply connected disk with a  cut. If we go around the origin in the positive direction, we cross the cut and $\Xi_2$ is given by
\begin{align}
\boxed{\Xi_2^\prime = 2\frac{1 - |w|^{2c}}{|1 - \varepsilon w^c|^2}}, \ (c \in \mathbb{R}, \ c \not= 0) , \ \varepsilon = e^{2\pi i c} .
\end{align}
$\Xi_2$ and $\Xi_2^\prime$ are zero on the unit circle $|w|=1$, except at the roots of $w^c=1$ or $w^c=\bar{\varepsilon}$ (where they have simple poles).

\end{example}

If the hyperbolic vortex solutions obtained from $\Xi_2$ and $\Xi_2^\prime$ are $(a,\phi)$ and $(a^\prime,\phi^\prime)$, respectively, then to construct one hyperbolic vortex on the full disk, we need to find a gauge transformation $G$ that connects two vortices.
Therefore, we need to find $g\in \mathbb{C}$ that satisfies 
\begin{align}
&\frac{\partial}{\partial w} \log \Xi_2^\prime = g \frac{\partial}{\partial w} \log \Xi_2, \nonumber\\
&\frac{\partial}{\partial \bar{w}} \log \Xi_2^\prime=\frac{\partial}{\partial \bar{w}} \log \Xi_2- \frac{\partial}{\partial \bar{w}} \log g .
\end{align}
Indeed, the group element $g \in U(1)$ is given by:
\begin{equation}
g = \varepsilon \frac{1-\bar{\varepsilon}\bar{w}^{c}}{1-\varepsilon w^{c}} \frac{1-w^{c}}{1-\bar{w}^{c}} \in U(1) .
\end{equation}
Therefore, the two are gauge equivalent, and we obtain a well-defined vortex on the entire punctured disk.
Calculating the first Chern class of this vortex gives 
\begin{align}
c_1(a)=c-1.
\end{align}
Therefore, for $c=1$, $c_1(a)=0$, which is a vacuum solution.
For $c=2$, $c_1(a)=1$, which is a standard hyperbolic vortex.

For $c=5/2$, $c_1(a)=3/2$, which reproduces Forgacs-Horvath-Palla(1981)\cite{FHP81}.
It has been shown that this is a singular symmetric instanton with a singular point on the 2-sphere ($w=0 \Leftrightarrow z=i \Leftrightarrow t=0, \ r:=\sqrt{(x^1)^2+(x^2)^2+(x^3)^2}=1$).
The contribution of the connection $a$ to the holonomy around the circle $|w|=r$ centered at the origin is $\oint_{|w|=r}a= -2r^c \sin (2\pi ci)\to 0 \ (r \to 0, c>0)$. Hence, all holonomies around the origin arise from gauge transformations, and $g\cong \varepsilon=e^{2\pi ic}$.
Any value of $c \in \mathbb{R}$ is possible.
The description here follows (Landweber(2005)\cite{Landweber05}).

\section{Area law and quark confinement}

Finally, we define the Wilson loop operator for a closed loop $C$, and calculate its expectation value using the dilute gas approximation of hyperbolic magnetic monopoles and hyperbolic vortices to show the area law.  
Therefore, confinement is understood in the sense of a linear potential for quark-antiquark static potential.

\begin{definition}[Wilson loop operator]
Let $\mathscr{A}$ be a Lie algebra valued \textbf{connection 1-form}:
\begin{equation}
\mathscr{A}(x) := \mathscr{A}_\mu(x) dx^\mu = \mathscr{A}_\mu^A(x) T_A dx^\mu.
\end{equation}
For a given loop, i.e., a closed path $C$, the \textbf{Wilson loop operator} $W_{\rm C}[\mathscr{A}]$ in the representation $R$ is defined by 
\begin{equation}
W_C[\mathscr{A}] := \mathcal{N}^{-1} {\rm tr}_{R} \left\{ \mathscr{P} \exp \left[ ig_{{}_{\rm YM}} \oint_C \mathscr{A} \right] \right\} ,
\end{equation}
where $\mathscr{P}$ represents the \textbf{path ordered product}, and the normalization factor $\mathcal{N}$ is equal to the dimension $d_R$ of the representation $R$ to which the Wilson loop probe belongs, leading to $W_C[0]=1$.
The Yang-Mills coupling constant $g_{{}_{\rm YM}}$ 
can be eliminated by scaling the field: $\mathscr{A}  \to g_{{}_{\rm YM}}^{-1}\mathscr{A}$.

\end{definition}

\begin{prop}[non-Abelian Stokes theorem for the Wilson loop operator]

The $SU(2)$  Wilson loop operator in any representation characterized by a half-integer single index  $J=\frac12, 1, \frac32, 2, \frac52, \cdots$ obeys the \textbf{non-Abelian Stokes theorem}. 
We introduce the unit vector field $n^A(x)$ ($n^A(x)n^A(x)=1$) called the \textbf{color direction field} defined by
\begin{align}
 n^A(x) \sigma_A   =& U(x)  \sigma_3 U^\dagger(x) , \ U(x) \in SU(2)  
  ,
\end{align}
with the third Pauli matrix $\sigma_3$.  
Then the $SU(2)$ Wilson loop operator is rewritten in the form of the area integral over any surface $\Sigma$ bounded by the loop $C$:
\begin{equation}
  W_C[\mathscr{A}]  
 = \int  [d \mu(U)]_\Sigma \exp \left\{ i  g_{{}_{\rm YM}} J \int_{\Sigma: \partial \Sigma=C} dS^{\mu\nu} F_{\mu\nu}^U \right\} , 
\end{equation}
where $F_{\mu\nu}^U$ is the gauge-invariant field strength defined by
\begin{align}
  F_{\mu\nu}^U(x) :=& \partial_\mu [ {n}^A(x)   \mathscr{A}^A_\nu(x) ] -  \partial_\nu [{n}^A(x)  \mathscr{A}^A_\mu(x)  ]   
\nonumber\\ &
  - g_{{}_{\rm YM}}^{-1} \epsilon^{ABC}  {n}^A(x)  \partial_\mu  {n}^B(x)  \partial_\nu  {n}^C(x) 
 ,
\end{align}
and $ [d \mu(U)]_\Sigma$ is the product measure of an invariant measure on $SU(2)/U(1)$ over $\Sigma$: 
\begin{align}
 [d \mu(U)]_\Sigma :=&\prod_{x \in \Sigma}d\mu(\bm{n}(x)) ,
\nonumber\\  
  d\mu(\bm{n}(x)) =& \frac{2J+1}{4\pi} \delta(\bm{n}^A(x)   \bm{n}^A(x)-1) d^3 \bm{n}(x) .  
\end{align}
This version of the {non-Abelian Stokes theorem} was investigated by (Diakonov and Petrov (1989)\cite{DP89}, Diakonov and Petrov (1996)\cite{DP96}, Kondo (1998)\cite{KondoIV}, Kondo and Taira (2000)\cite{KT00}, Kondo (2008)\cite{Kondo08}).

\end{prop}

\begin{proof}
See Kondo (2008)\cite{Kondo08} or Kondo et al. (2015)\cite{PR}.
\end{proof}

\begin{figure}[htb]
\begin{center}
\includegraphics[scale=0.5]{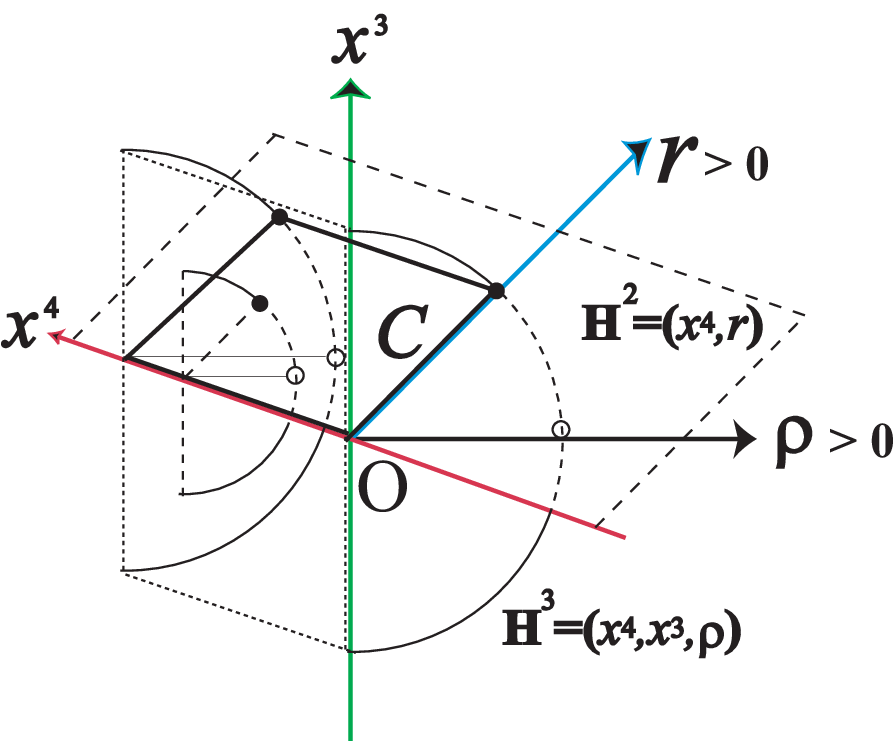}
\end{center}
\caption{
(Left) The relationship between Wilson loop $C$ and the hyperbolic vortex (black circle) on $\mathbb{H}^2$ and the hyperbolic magnetic monopole (white circle) on $\mathbb{H}^3$.
(Right) The dilute gas approximation. 
}
\label{Wilson_vortex_monopole}
\end{figure}

\noindent
(I) Quark confinement due to hyperbolic vortices on $\mathbb{H}^2$: 

The Witten transformation corresponds to choosing the color direction field  as (see Appendix A for more details.)
\begin{equation}
 n^A(x) = \frac{x^A}{r} \quad (r:=\sqrt{x^Ax^A}) .
\end{equation}
Then the Abelian-like field defined by
\begin{equation}
 c_\mu (x) := n^A(x) \mathscr{A}_\mu^A(x) 
\end{equation}
is rewritten by using the Witten transformation into 
\begin{equation}
 c_\mu (x) =
 \begin{cases}
  c_4 (x)=  \frac{x^A}{r} \mathscr{A}_4^A(x)  = a_0(r,t) & (\mu=4) 
  \\
  c_j (x)=  \frac{x^A}{r} \mathscr{A}_j^A(x)  = \frac{x^j}{r} a_1(r,t) & (\mu=j) 
  \end{cases}
  .
\end{equation}
If we consider the loop $C$ on the $(t,r)$ plane, i.e., $\mu=4, \nu=r$, 
the second term vanishes: $- g_{{}_{\rm YM}}^{-1} \epsilon^{ABC}  {n}^A(x)  \partial_\mu  {n}^B(x)  \partial_\nu  {n}^C(x)=0$. 
Therefore we find
\begin{align}
 F_{4r}^U(x) =& \partial_4 c_r(x) - \partial_r c_4(x) 
 = \partial_4 \left( \frac{x^j}{r} c_j(x) \right) - \partial_r c_4(x) \nonumber\\
 =& \partial_4 a_1(r,t) - \partial_r a_0(r,t) := F_{4r}(t,r)  .
\end{align}
In this setting, the Wilson loop operator for a rectangular loop $C$ with the size $T \times L$ is expressed as 
\begin{align}
  W_{C=T \times L}[\mathscr{A}]  
 = \exp \left\{ i  J \int_{-T/2}^{T/2} dt \int_{0}^{L} dr F_{4r}(t,r) \right\} . 
 \label{Wilson_loop_A}
\end{align}
If the rectangular loop $C$ is very large $L,T \to \infty$ so that a vortex is located inside of $C$, the integral becomes equal to the topological charge $N_v=c_1$ according to (\ref{Nv}): 
\begin{align}
& \int_{-T/2}^{T/2} dt \int_{0}^{L} dr F_{4r}(t,r)  (L,T \to \infty) \nonumber\\
& \to  \int_{-\infty}^{\infty} dt \int_{0}^{\infty} dr F_{4r}(t,r) = 2\pi N_v  .
\end{align}
Since $2J$ is an integer, we find 
\begin{align}
  & W_{C=T \times L}[\mathscr{A}]  
 \to   \exp \left\{ i  2 \pi Jc_1 \right\} 
 = \exp (i\pi)^{2Jc_1} \nonumber\\ 
 =& (-1)^{2Jc_1}  
 =
 \begin{cases}
  (-1)^{c_1} & (J=\frac12, \frac32, ...) \\
  (+1)^{c_1}  & (J=1, 2, ... ) 
 \end{cases} . 
\end{align}
For a 1-vortex with $c_1=1$,  we find $W_{C=T \times L} \to \pm \in Z(2)$. Therefore, this vortex is regarded as the \textbf{center vortex}, since the center of $SU(2)$ is $Z(2)$.

Now we evaluate the Wilson loop expectation value to obtain the
static potential for two widely separated color charges
 in a $\theta$ vacuum. 
Note that the integrand of the Wilson loop operator shown above is the density of the instanton number, which means that in this theory, the Wilson loop $W_C[\mathscr{A}]$ counts the number of instantons-antiinstantons (or vortices-antivortices) that exist within the region $\Sigma$ enclosed by the loop $C$. The expectation value of the Wilson loop, including the topological term $i\theta Q$, is expressed as 
\begin{equation}
\langle \theta | W_C[\mathscr{A}] | \theta \rangle_{\rm{GS}} =
{\int \mathcal{D}A \mathcal{D}\phi
e^{-S_{\rm{GS}}+i \theta Q} W_C[\mathscr{A}]
\over 
\int \mathcal{D}A \mathcal{D}\phi
e^{-S_{\rm{GS}}+i \theta Q} }
=: {I_2 \over I_1},
\label{wle}
\end{equation}
Note that a nonzero $\theta$ is not required to show the area law of the Wilson loop below. We can set $\theta=0$ in the final result.
Including the topological term $i \theta Q$ in the action is equivalent to defining the $\theta$ vacuum as follows:
\begin{align}
|\theta \rangle := \sum_{n=-\infty}^{+\infty} e^{in\theta}
|n \rangle .
\end{align}

\par
In the following, we calculate the Wilson loop expectation value (\ref{wle}) using the dilute instanton gas approximation.
This method is well known, see for example Chapter 11 of Rajaraman(1989)\cite{Rajaraman89} or Chapter 7 of Coleman(1985) \cite{Coleman85}.

\begin{prop}[area law of the Wilson loop average]
In the dilute instanton gas approximation, the expectation value of the Wilson loop operator (\ref{Wilson_loop_A}) obeys the area law: 
\begin{align}
& \langle \theta | W_C[\mathscr{A}] | \theta \rangle = e^{- \sigma A(C)} ,
\nonumber\\
&\sigma :=  2K e^{-S_1/\hbar} \left[ \cos (\theta c_2) - \cos
\left(\theta c_2 + 2\pi J c_1\right) \right] ,
\label{wler}
\end{align}
where $c_1$ and $c_2$ are the first and second Chern numbers respectively. 
Taking the rectangular Wilson loop leads to the static quark potential:
\begin{equation}
V(R) = \sigma R ,
\quad
\sigma = 2K e^{-S_1/\hbar} \left[ \cos (\theta c_2) - \cos
\left(\theta c_2 + 2\pi J c_1\right) \right] .
\label{pot}
\end{equation}
When $J$ is an integer, the vacuum is periodic with respect to $\theta$ with period $2\pi$, so the potential is zero.
This means that the integral charges are screened by the formation of neutral bound states.
When $J$ is not an integer,
the static quark potential $V(R)$ is given by a linear potential with string tension $\sigma$ as the proportionality coefficient.

\end{prop}

\begin{proof}
Let $n_{+}$ and $n_{-}$ be the total number of instantons and anti-instantons in the volume (area)$V$, respectively. 
The action of a 1-instanton and a 1-anti-instanton is equal, and we call this $S_1$.
Then the action  is expressed as follows:
\begin{align}
S_{\rm{GS}} 
= (n_{+}+n_{-})S_1, \  
S_1=  {4\pi^2 \over g_{{}_{\rm YM}}^2} .
\end{align}
(\ref{wle}) is considered to be the average number of the total instanton charge $Q$ within the volume (area) $V$ of all instanton-antiinstanton ensembles generated by the action of the theory.

First, we classify the field configurations that contribute to the instanton tunneling amplitude according to the number of sufficiently separated instantons $n_{+}$ and antiinstantons $n_{-}$ where 
\begin{align} 
Q=nc_2=(n_{+}-n_{-}) c_2 .
\end{align}

Next, we sum over all configurations that contain sufficiently separated $n_+$ instantons and $n_-$ antiinstantons.
In the dilute gas approximation, the calculation of the tunneling amplitude is reduced to the calculation of a single instanton (or antiinstanton) contribution: 
$n \rightarrow n+1$ (or $n \rightarrow n-1$).
The terms $n_{+}=1, n_{-}=0$ (or $n_{+}=0, n_{-}=1$) are given by
\begin{align}
& \langle n=\pm 1|e^{-HT} | 0 \rangle \nonumber\\
=& \int d\mu(\lambda) \int_{V} d^2 x \exp (-S_1/\hbar) \exp (\pm i\theta c_2)
\nonumber\\
=& KV \exp (-S_1/\hbar) \exp (\pm i\theta c_2).
\end{align}
Here, the coefficients $KV$ can be in principle obtained from integrals of the collective coordinates, i.e., the size $\lambda$ and position $x$ of the instanton:
\begin{align}
\int d\mu(\lambda) \int_{V} d^2 x = K V .
\end{align}
Here, $V$ is the volume (area) of a finite but sufficiently large two-dimensional space.
$K$ is a constant. To know the exact form of $K$, it is necessary to determine the measure $\mu(\lambda)$ of the collective coordinate $\lambda$, but this is omitted.

In the dilute gas approximation, the denominator of the Wilson loop expectation value, i.e., the partition function, can be calculated as 
\begin{align}
I_1
=& \sum_{n_{+},n_{-}=0}^{\infty} {(KV)^{n_{+}+n_{-}} \over
n_{+}!n_{-}!}  \nonumber\\ &
\times \exp \left[-(n_{+}+n_{-})S_1/\hbar 
+ i \theta c_2 (n_{+}-n_{-}) \right] 
\nonumber\\ 
=& \sum_{n_{+}=0}^{\infty} 
{(KV e^{-S_1/\hbar + i\theta c_2})^{n_{+}} \over n_{+}!} 
\sum_{n_{-}=0}^{\infty}  { (KV e^{-S_1/\hbar - i\theta c_2})^{n_{-}} \over n_{-}!}  
\nonumber\\
=& \exp [ KV e^{-S_1/\hbar + i\theta c_2}
+ KV e^{-S_1/\hbar - i\theta c_2} ]
\nonumber\\
=& \exp \left[2 KV \cos (\theta c_2) e^{-S_1/\hbar} \right] .
\end{align}
Here, there are no restrictions on integers $n_+$ or $n_-$, because we sum over $Q = n_+ - n_-$.
This sum is exactly the grand partition function for a classical perfect gas, i.e., non-interacting particles.

To calculate the numerator of the Wilson loop expectation value, let $n_{+}^{in}$ and $n_{+}^{out}$ be the numbers of instantons inside and outside the loop $C$, respectively.
Similarly, let $n_{-}^{in}$ and $n_{-}^{out}$ be the numbers of anti-instantons inside and outside the loop $C$, respectively. 
Let $A(C)$ be the area on the plane enclosed by loop $C$.
The dilute gau approximation only makes sense if the loop $C$ is large enough so that the size of the instantons is negligibly small compared to the size of the loop $C$, and the overlap of instantons and anti-instantons with the loop is ignored (this is equivalent to ignoring the perimeter part of the Wilson loop):
\begin{align}
W_C[\mathscr{A}] = \exp \left[ 2\pi i J c_1 (n_{+}^{in}-n_{-}^{in})  \right] .
\end{align}
The numerator of the Wilson loop expectation value is
\begin{align}
I_2
  =& \sum_{n_{+}^{in},n_{-}^{in}=0}^{\infty}
{(K A(C))^{n_{+}^{in}+n_{-}^{in}}
\over n_{+}^{in}!n_{-}^{in}!}  
\nonumber\\
&  e^{-(n_{+}^{in}+n_{-}^{in})S_1/\hbar
   + i \theta c_2 (n_{+}^{in}-n_{-}^{in}) } e^{ 2\pi i J c_1 (n_{+}^{in}-n_{-}^{in}) }
  \nonumber\\ 
  & \times \sum_{n_{+}^{out},n_{-}^{out}=0}^{\infty}
{(K(V-A(C)))^{n_{+}^{out}+n_{-}^{out}}
\over n_{+}^{out}!n_{-}^{out}!}  
\nonumber\\
&  e^{-(n_{+}^{out}+n_{-}^{out})S_1/\hbar 
  + i \theta c_2 (n_{+}^{out}-n_{-}^{out}) }   \nonumber\\
  =& \sum_{n_{+}^{in}=0}^{\infty} \frac{1}{n_{+}^{in}!} \left( KA(C) e^{-S_1/\hbar} e^{i\theta c_2 } e^{2\pi i J c_1} \right)^{n_{+}^{in}} 
\nonumber\\
& \times \sum_{n_{-}^{in}=0}^{\infty} \frac{1}{n_{-}^{in}!} \left( KA(C) e^{-S_1/\hbar} e^{-i\theta c_2 } e^{-2\pi i J c_1} \right)^{n_{-}^{in}}
  \nonumber\\ 
  & \times \sum_{n_{+}^{out}=0}^{\infty} \frac{1}{n_{+}^{out}!} \left( K(V-A(C)) e^{-S_1/\hbar} e^{i\theta c_2 } \right)^{n_{+}^{out}} 
\nonumber\\
& \times \sum_{n_{-}^{out}=0}^{\infty} \frac{1}{n_{-}^{out}!} \left( K(V-A(C)) e^{-S_1/\hbar} e^{i\theta c_2 } \right)^{n_{-}^{out}} 
  \nonumber\\
  =& \exp \Big\{2 K e^{-S_1/\hbar} \Big[ A(C)\cos \left(\theta c_2+ 2\pi J c_1 \right) 
\nonumber\\
&  \quad\quad\quad\quad\quad\quad\quad\quad + (V-A(C))\cos \theta c_2 \Big]
 \Big\} . 
\end{align}
Therefore, the Wilson loop expectation value is given as the ratio $I_2/I_1$:
\begin{align}
& \langle \theta | W_C[\mathscr{A}] | \theta \rangle = e^{- \sigma A(C)} ,
\nonumber\\
& \sigma := 2K e^{-S_1/\hbar} \left[ \cos (\theta c_2) - \cos
\left(\theta c_2 + 2\pi J c_1\right) \right] .
\label{wler}
\end{align}
We note here that the volume dependence disappears by taking the ratio $I_2^\theta/I_1^\theta$.
Thus, the Wilson loop expectation value in a vacuum with the topological angle $\theta$ shows the area law.

\end{proof}

\begin{rem} 
It should be noted that the confining potential is a non-perturbative quantum effect caused by instantons, because the linear potential 
has a factor $e^{-S_1/\hbar}$, which is exponentially small in $\hbar$ and vanishes at $\hbar \rightarrow 0$. This is the crucial difference between the linear potential in two dimensions (\ref{pot}) and the linear Coulomb potential. The latter arises simply because one-dimensional space leaves no room for fluxes to be spread out.

\end{rem}

\noindent
(II) Quark confinement due to hyperbolic magnetic monopoles on $\mathbb{H}^3$ and holography: 

Next, we calculate the the Wilson loop average in a different setting. 
We take the Wilson loop operator on the conformal boundary $\partial\mathbb{H}^3$ of $\mathbb{H}^3$ by taking the limit $\rho \to 0$. 

\begin{figure}[htb]
\begin{center}
\includegraphics[scale=0.5]{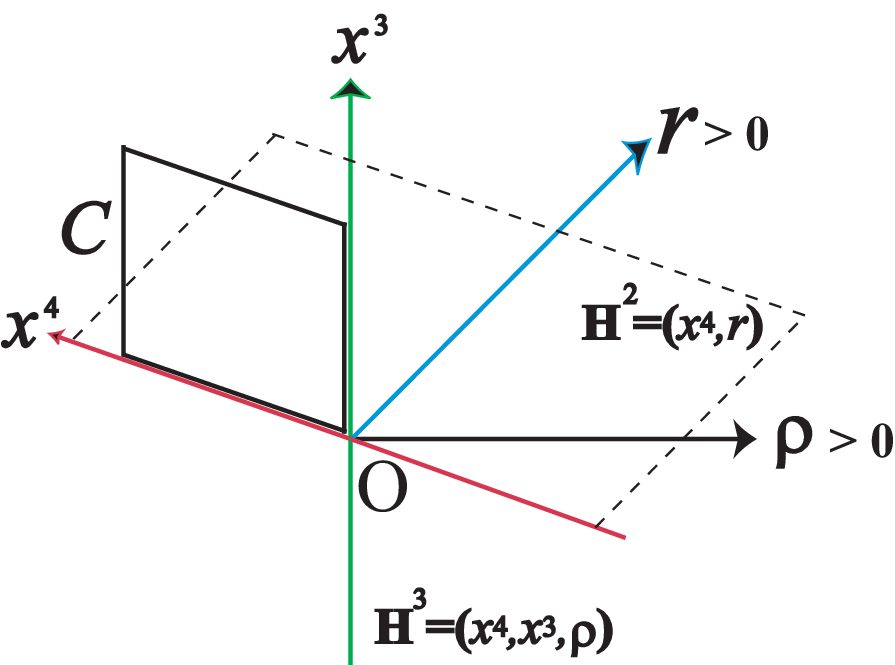}
\end{center}
\caption{
The Wilson loop $C$ located on the conformal boundary $\partial \mathbb{H}^3$, i.e., $x^3-x^4$ plane of the hyperbolic space $\mathbb{H}^3$.
}
\label{Wilson_vortex_monopole}
\end{figure}


\begin{prop}[Wilson loop operator on the conformal boundary $\partial \mathbb{H}^3$]
If the loop $C$ lies on the conformal boundary $\partial \mathbb{H}^3$, i.e., $x^3-x^4$ of $\mathbb{H}^3$, the Wilson loop operator in the fundamental representation $F$ defined for the $S^1$-invariant $SU(2)$ Yang-Mills field $\mathscr{A}_\mu^G$ takes the simple Abelian form:
\begin{align}
W_C[\mathscr{A}] 
=& \frac12  {\rm tr}_{F} \left\{ \exp \left[ i \frac{\sigma_3}{2} \oint_C dx^\mu a_\mu(t,x^3) \right] \right\} \nonumber\\
=&  \frac12  {\rm tr}_{F} \left\{ \exp \left[ i \frac{\sigma_3}{2} \int_{\Sigma: \partial \Sigma=C} dt dx^3  F_{4r}(t,x^3) \right] \right\} 
.
\end{align}
This Wilson loop operator is the same as the previous one (\ref{Wilson_loop_A}). 
Adopting the dilute instanton gas approximation, the Wilson loop average obeys the area law (\ref{wler}). 

\end{prop}

\begin{proof}
This is an immediate consequence of Prop.~\ref{prop:monopole-dominance}.

\end{proof}

Therefore, the Yang-Mills field approaches the diagonal Abelian field on the conformal boundary $x^3-x^4$.  This is regarded as the Abelian dominance and the magnetic monopole dominance.

In the ordinary flat Euclidean case, the (infrared) Abelian dominance and the the magnetic monopole dominance in quark confinement have been confirmed by numerical simulations and also supported by analytical investigations, see e.g., \cite{PR} for the details, although there are no rigorous proofs.

\begin{rem} 
This result obtained in the dilute gas approximation is consistent with the `t Hooft anomaly, which can be shown in the way similar to that given by Tanizaki and \"Unsal (2022)\cite{TU22}.

\end{rem} 

\begin{rem} 
In order to discuss deconfinement at finite temperature, we need the caloron solution. 
The hyperbolic caloron was studied by Harland(2008)\cite{Harland08} and Sibner-Sibner-Yang(2015)\cite{SSY15} where  the caloron of the Harrrington-Schepard type was found.  
However, the caloron of the KvBLL type with the non-trivial holonomy has never been found to the knowledge of the author. 

\end{rem}

\section{The viewpoint of the Cho-Duan-Ge-Faddeev-Niemi decomposition}

It is shown that the Witten transformation is a special case of the $SU(2)$ gauge-covariant Cho-Duan-Ge-Faddeev-Niemi decomposition of the gauge field (Cho(1980)\cite{Cho80}, Duan and Ge(1979)\cite{DG79}, Faddeev and Niemi(FN99)\cite{FN99,FN99a,FN99b}, Shabanov(1999)\cite{Shabanov99}, Kondo, Murakami and Shinohara(2006)\cite{KMS06}) 
Therefore, the decomposition can potentially give more general point of view for rewriting and reinterpreting the Yang-Mills theory. 
For example, it enables us to understood the reason why the dimensionally reduced gauge theory has the just Abelian $U(1)$ gauge invariance which is reduced from $SU(2)$ gauge invariance of the original non-Abelian gauge theory.

\begin{definition}[Cho-Duan-Ge-Faddeev-Niemi decomposition of Yang-Mills field]

The Cho-Duan-Ge-Faddeev-Niemi (CDGFN) decomposition of the SU(2) gauge field $\mathbf{A}_\mu(x)$ can be written in vector notation as follows:
\begin{align}
\mathbf{A}_\mu(x)=& \mathbf{V}_\mu(x) + \mathbf{X}_\mu(x), \nonumber\\
\mathbf{V}_\mu(x):=& c_\mu(x)\mathbf{n}(x) + g^{-1}\partial_\mu \mathbf{n}(x) \times\mathbf{n}(x), \nonumber\\
\mathbf{X}_\mu(x):=& \phi_1(x)\partial_\mu\mathbf{n}(x) + \phi_2(x)\partial_\mu\mathbf{n}(x)\times\mathbf{n}(x).
\end{align}
Here, $\mathbf{n}(x)$ is a unit vector field with three components:
\begin{align}
\mathbf{n}(x)\cdot\mathbf{n}(x) = n^A(x)n^A(x) = 1.
\end{align}
$c_\mu(x)$ is an Abelian vector field. $\phi_1(x)$ and $\phi_2(x)$ are real scalar fields. The field $\mathbf{n}(x)$ has two independent degrees of freedom, the field $c_\mu(x)$ has two independent (transverse polarization) degrees of freedom, and the fields $\phi_1(x)$ and $\phi_2(x)$ each have one degree of freedom. This is equivalent to the $2\times3=6$ polarization degrees of freedom of the SU(2) Yang-Mills field $\mathbf{A}_\mu(x)$ after gauge fixing.

Note that $\mathbf{n}$, $\partial_\mu\mathbf{n}$, and $\partial_\mu\mathbf{n}\times \mathbf{n}$ constitute a basis for $SU(2)$ space for any choice of $\mathbf{n}(x)$.  
See the left panel of Fig.~\ref{fig:SU2_n_basis}.
\end{definition}

\begin{figure}[htb]
\begin{center}
\includegraphics[scale=0.20]{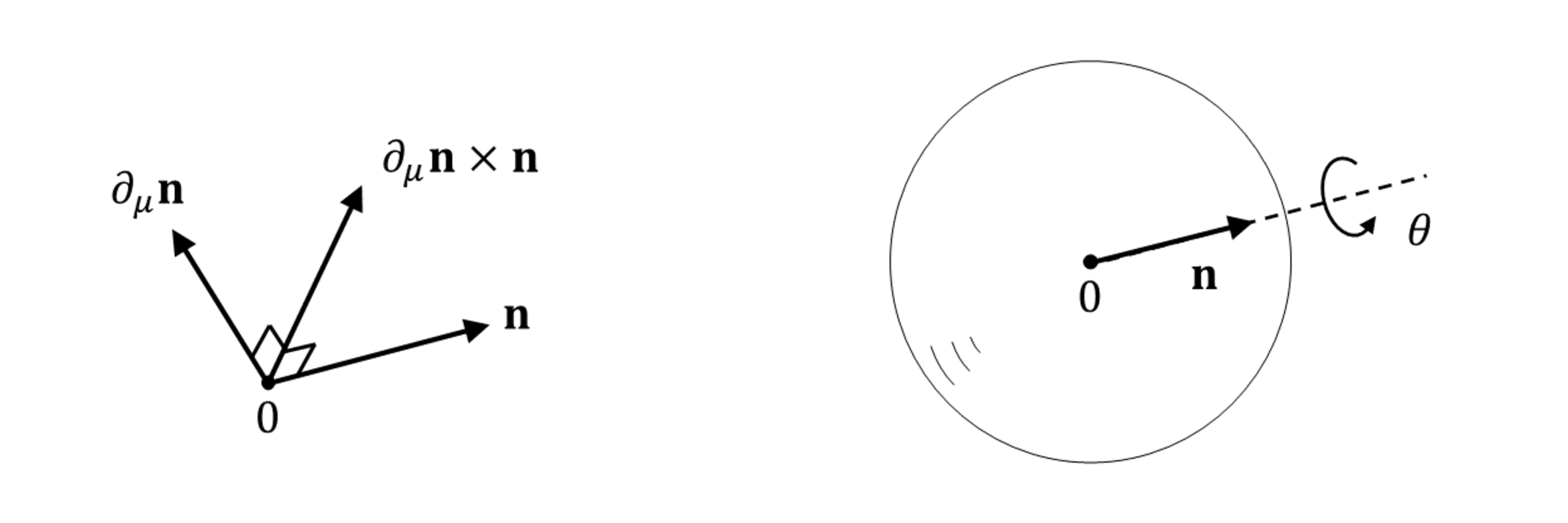}
\end{center}
\vskip -0.5cm
\caption{
(Left) $\mathbf{n}, \partial_\mu\mathbf{n}, \partial_\mu\mathbf{n}\times \mathbf{n}$ as a basis for $SU(2)$ space,
(Right) Rotation by angle $\theta$ around the $\mathbf{n}$-direction axis.
}
\label{fig:SU2_n_basis}
\end{figure}

\begin{rem} 
The independent field degrees of freedom in both sides of the Cho-Duan-Ge-Faddeev-Niemi decomposition of the $D=4$ and $SU(2)$ Yang-Mills field have the correct matching as follows. 

The decomposition is considered on-shell: the $D=4$ and $SU(2)$ Yang-Mills field $\mathscr{A}_\mu^{A}$ has $2 \times 3=6$ (2 transverse components for $\mu$, 3 components for $A$) on the left hand side, which agrees with $2+2+2=6$ (2 transverse components for $c_\mu$, 2 components for $\mathbf{n}$, 2 components for $\phi_1, \phi_2$) on the right hand side. 
See Faddeev-Niemi (1999a)\cite{FN99a}

For the decomposition to be valid for off-shell with the gauge fixing conditions (gauge-fixed off-shell), we need to introduce another pair of the scalar field $\tilde\phi_1, \tilde\phi_2$: the $D=4$ and $SU(2)$ Yang-Mills field $\mathscr{A}_\mu^{A}$ has $4 \times 3-3=9$ (4 components for $\mu$, 3  for $A$, 3 gauge fixing conditions, e.g., $\partial^\mu \mathscr{A}_\mu^{A}=0$) on the left hand side, which agrees with $4+2+2+2-1=9$ (4 components for $c_\mu$, 2 components for $\mathbf{n}$, 2 components for $\phi_1, \phi_2$, 2 components for $\tilde\phi_1, \tilde\phi_2$, 1 gauge fixing condition, e.g., $\partial^\mu c_\mu=0$) on the right hand side. 
Without gauge fixing, both sides do not agree for the independent fields. 
See Faddeev-Niemi (1999b)\cite{FN99b}
This fact is important at the level of a quantum theory which goes beyond obtaining the solution of the field equation.  

\end{rem}

\begin{example}[Witten Ansatz and BPST instanton]
Witten Ansatz
\begin{align}
\mathscr{A}_0^A(x) =& \frac{x^A}{r}a_0(r,t), \\
\mathscr{A}_j^A(x) =&\frac{x^jx^A}{r^2}a_1(r,t) + \frac{\delta^{jA}r^2 - x^jx^A}{r^3}\phi_1(r,t) \nonumber\\
&+ \frac{\epsilon^{jAk}x^k}{r^2}[1 + \phi_2(r,t)]
\end{align}
is reproduced from the CDGFN decomposition:
\begin{align}
\mathscr{A}_0^A(x) =& a_0(r,t)\mathbf{n}^A(x),\\
\mathscr{A}_j^A(x) =& \frac{x^j}{r}a_1(r,t)\mathbf{n}^A(x) + \partial_j \mathbf{n}^A \phi_1(r,t) \nonumber\\
&+ (\partial_j \mathbf{n} \times \mathbf{n})^A[1 + \phi_2(r,t)] ,
\end{align}
with
\begin{align}
n^A(x) =& \frac{x^A}{r}, \ c_0(x) = a_0(r,t) , \ c_j(x) = \frac{x^j}{r}a_1(r,t) , \nonumber\\
\phi_1(x) =& \phi_1(r,t), \ \phi_2(x) = \phi_2(r,t), \nonumber\\
& r:=|\mathbf{x}| = \sqrt{x_1^2 +x_2^2 +x_3^2 }.
\end{align}
In particular, the one-instanton solution of 
Belavin-Polyakov-Shwarts-Tyupkin (BPST) \cite{BPST} corresponds to the following choice of Ansatz function:
\begin{align}
A_0 =& F(r,t)r, \ A_1 = -F(r,t)t, \nonumber\\
\phi_1 =& -F(r,t)rt, \ \phi_2 = F(r,t)r^2 - 1, \nonumber\\
& F(r,t) = \frac{2}{r^2 + t^2 + \lambda^2}.
\end{align}
\end{example}

\begin{example}[Wu-Yang magnetic monopole]
The Wu-Yang magnetic monopole configuration (Wu and Yang(1975)\cite{WY75})
\begin{align}
\mathscr{A}_0^A(x) = 0, \ \mathscr{A}_j^A(x) = \epsilon_{Ajk}\frac{x_k}{r^2}
\end{align}
is reproduced from the CDGFN decomposition:
\begin{align}
n^A(x) =\frac{x^A}{r}, \ c_\mu(x) = 0, \ \phi_1(x) = 0, \ \phi_2(x) = 0.
\end{align}
\end{example}

\begin{example}['tHooft-Polyakov magnetic monopole]
The 'tHooft-Polyakov magnetic monopole configuration (`t Hooft, Polyakov(1974)\cite{tHP74})
\begin{align}
\mathscr{A}_0^A(x) = 0, \ \mathscr{A}_j^A(x) = \epsilon^{Ajk}\frac{x^k}{r^2}f(r) .
\end{align}
This is reproduced from the CDGFN decomposition:
\begin{align}
& n^A(x) =\frac{x^A}{r}, \ c_\mu(x) = 0, \nonumber\\
& \phi_1(x) = 0, \ \phi_2(x) = -1 + f(r).
\end{align}
\end{example}

\begin{prop}[Gauge covariance of Cho-Duan-Ge-Faddeev-Niemi decomposed Yang-Mills fields]\label{FNd:covariance}

Consider the SU(2) finite gauge transformation of the SU(2) Yang-Mills field $\mathbf{A}_\mu(x)$:
\begin{align}
\mathbf{A}_\mu(x) \rightarrow& \mathbf{A}_\mu^U(x) \nonumber\\
=& U(x)\mathbf{A}_\mu(x)U(x)^\dagger 
+ ig^{-1}U(x)\partial_\mu U(x)^\dagger, \nonumber\\
& U(x)  \in {\rm{SU}}(2) .
\label{SU2-U1_transf}
\end{align}
Under this transformation, the Cho-Duan-Ge-Faddeev-Niemi decomposition is covariant. That is, the SU(2) Yang-Mills field $\mathbf{A}_\mu^U(x)$ after gauge transformation has the same Cho-Duan-Ge-Faddeev-Niemi decomposition as the original gauge field $\mathbf{A}_\mu(x)$:
\begin{align}
\mathbf{A}_\mu^U(x) =& \mathbf{V}_\mu^U(x) + \mathbf{X}_\mu^U(x), \nonumber\\
\mathbf{V}_\mu^U(x):=& U(x)\mathbf{V}_\mu(x)U(x)^\dagger + ig^{-1}U(x)\partial_\mu U(x)^\dagger \nonumber\\
=& c_\mu^U(x)\mathbf{n}(x) + g^{-1}\partial_\mu\mathbf{n}(x)\times\mathbf{n}(x), \nonumber\\
\mathbf{X}_\mu^U(x):=& U(x)\mathbf{X}_\mu(x)U(x)^\dagger \nonumber\\
=& \phi_1^U(x)\partial_\mu\mathbf{n}(x) + \phi_2^U(x)\partial_\mu\mathbf{n}(x)\times\mathbf{n}(x), \nonumber\\
\label{SU2_transf1}
\end{align}
provided that the color direction field transforms according to the adjoint representation:
\begin{align}
\mathbf{n}^U(x) =& U(x)\mathbf{n}(x) U(x)^\dagger .
\end{align}

\end{prop}

\begin{proof}
See Appendix A.

\end{proof}

In particular, we consider the special case. 

\begin{prop}[Maximal torus gauge covariance of Cho-Duan-Ge-Faddeev-Niemi decomposed Yang-Mills fields]\label{FNd:covariance}

Consider the SU(2) finite gauge transformation of the SU(2) Yang-Mills field $\mathbf{A}_\mu(x)$ with 
\begin{align}
U(x) =  \exp\left(i \theta(x)\mathbf{n}(x)\cdot\frac{\boldsymbol{\sigma}}{2}\right) \in {\rm{SU}}(2) .
\label{SU2-U1_transf}
\end{align}
See the right panel of Fig.~\ref{fig:SU2_n_basis}.
Under this local transformation (\ref{SU2-U1_transf}) that keeps the color direction field $\mathbf{n}(x)$ invariant:
\begin{align}
\mathbf{n}^U(x) =&\mathbf{n}(x), 
\end{align}
the other fields transform according to the maximal $U(1)$ gauge transformation:
\begin{align}
c_\mu^U(x) =& c_\mu(x) + g^{-1}\partial_\mu\theta(x), \nonumber\\
\begin{pmatrix}
\phi_1^U(x)\\
\phi_2^U(x)
\end{pmatrix}
=&\begin{pmatrix}
\cos \theta(x)&-\sin \theta(x)\\
\sin \theta(x)&\cos \theta(x)
\end{pmatrix}
\begin{pmatrix}
\phi_1(x)\\
\phi_2(x)
\end{pmatrix}.
\end{align}
\end{prop}

\begin{proof}
This result is shown by explicitly calculating the gauge transformation (\ref{SU2_transf1}) using  (\ref{SU2-U1_transf}). 
See Appendix A.

\end{proof}

\begin{rem}
Note that the transformation (\ref{SU2-U1_transf}) includes the color direction field $\mathbf{n}(x)$, and in quantum theory, it is a quantum gauge transformation.
\end{rem}

\begin{rem}
The Witten Ansatz is a special case of the decomposition in which the color direction field ${\bf n}$ is fixed in the sense that it does not change under the $SU(2)$ gauge transformation. 
As shown in the above Proposition, this is equivalent to restricting the gauge transformation to the $U(1)$ gauge transformation (\ref{SU2-U1_transf}) with the variable $\theta(x)$ because $\mathbf{n}(x)$ is fixed: $\mathbf{n}^A(x)=x^A/r$. 
Therefore, the Witten Ansatz is gauge covariant under the $U(1)$ gauge transformation where $SU(2)/U(1)$ part is fixed by the specific choice of the color field ${\bf n}(x)$. 
\end{rem}


\begin{figure}[tbp]
\begin{center}
\includegraphics[scale=0.30]{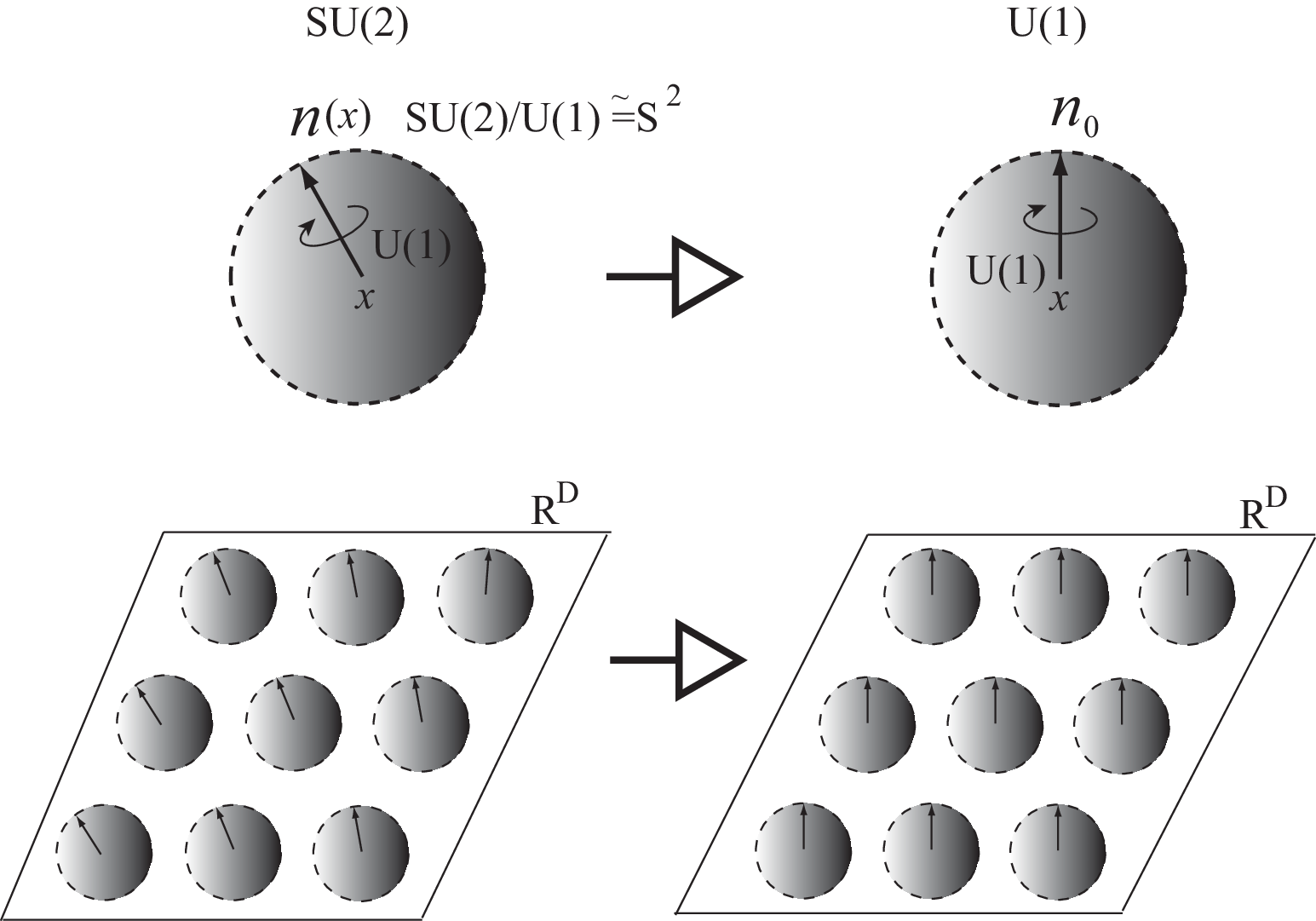}
\caption{ 
The color field and symmetry breaking.
(Left panel)
The original local $SU(2)$ gauge symmetry is preserved by a local embedding of the Abelian directions with the color direction field $\bm{n}(x)$: 
$SU(2) \simeq SU(2)/U(1) \times U(1) \simeq S^2 \times U(1)$.
(Right panel)
Partial gauge fixing $SU(2) \to U(1)$ is performed by fixing the color direction field globally, setting $\bm{n}(x) \equiv n_0$ for any $x \in \mathbb{R}^D$.
Only the local $U(1)$ symmetry corresponding to a local rotation around a fixed Abelian direction or a uniformly distributed color direction field remains.
}
\label{fig:color-field-fixing}
\end{center}
\end{figure}


\begin{prop}[Faddeev-Niemi decomposition of the Yang-Mills Lagrangian density]\label{FNd:YMLagrangian}

\noindent
(1) The Cho-Duan-Ge- decomposition for the field strength $\mathbf{F}_{\mu\nu} = \partial_\mu\mathbf{A}_{\nu} - \partial_\nu\mathbf{A}_{\mu} + g\mathbf{A}_{\mu}\times\mathbf{A}_{\nu}$ is
\begin{align}
\mathbf{F}_{\mu\nu} =& \{G_{\mu\nu} - [g^2(\phi_1^2 + \phi_2^2)-1]H_{\mu\nu}\}\mathbf{n} \nonumber\\
&+(D_\mu\phi_1\partial_\nu\mathbf{n} - D_\nu\phi_1\partial_\mu\mathbf{n}) \nonumber\\
&+ (D_\mu\phi_2\partial_\nu\mathbf{n}\times\mathbf{n} - D_\nu\phi_2\partial_\mu\mathbf{n}\times\mathbf{n}).
\end{align}
where
\begin{align}
G_{\mu\nu}:=&\partial_\mu c_\nu - \partial_\nu c_\mu ,
\ H_{\mu\nu}:=-g^{-1}\mathbf{n}\cdot(\partial_\mu\mathbf{n} \times\partial_\nu\mathbf{n}), \nonumber\\
D_\mu\phi_a:=&\partial_\mu\phi_a + \epsilon_{ab}gc_\mu\phi_b .
\end{align}
\\
\noindent
(2)
The Yang-Mills Lagrangian density $\mathscr{L}_{{\rm{YM}}}$ is
It can be calculated from $\mathscr{L}_{{\rm{YM}}} = \frac{1}{4}\mathbf{F}_{\mu\nu}\cdot\mathbf{F}^{\mu\nu}$:
\begin{align}
\mathscr{L}_{{\rm{YM}}} =& \pm\frac{1}{4}[G_{\mu\nu} - (g^2|\phi|^2-1)H_{\mu\nu}]^2 \nonumber\\
&\pm\frac{1}{4}\{[\delta_{\mu\nu}(\partial_\rho\mathbf{n})^2 - \partial_\mu\mathbf{n}\cdot\partial_\nu\mathbf{n}+igH_{\mu\nu}]\nonumber\\
& \times (D_\mu\phi)^*(D_\nu\phi) + {\rm{h.c.}} \}.
\end{align}
Here, ${\rm{h.c.}}$ represents the Hermite conjugate of the previous term. $D_\mu$ is the covariant derivative of U(1):
\begin{align}
D_\mu\phi = \partial_\mu\phi - igc_\mu\phi \, \ \phi = \phi_1 + i\phi_2 \in \mathbb{C}.
\end{align}
Note that the fundamental variables of this theory are $\mathbf{n}(x), c_\mu(x), \phi_1(x), \phi_2(x)$ without gauge fixing.
\end{prop}

\begin{proof}
This was shown in Faddeev-Niemi (1999)\cite{FN99}.
See Appendix A for the proof. 
\end{proof}

\begin{prop}[Effective theory written in new field variables]
From the Yang-Mills theory, we obtain an effective theory written in terms of  new field variables.

\noindent
(1) Integrating $c_\mu$ and $\phi$: For composite operators written in $c_\mu$ and $\phi$
\begin{align}
&\langle |\partial_\lambda\phi|^2 \delta_{\mu\nu} - \partial_\mu\phi^*\partial_\nu\phi\rangle = m^2\delta_{\mu\nu} 
\nonumber\\
&\Rightarrow \langle (D_\mu\phi)^*(D_\nu\phi)\rangle = m^2\delta_{\mu\nu} .
\end{align}
Then, we reduce to the Faddeev-Skyrme model:
\begin{align}
S = \int d^4x \ \left[m^2(\partial_\mu\mathbf{n})^2 + \frac{1}{g^2}(\mathbf{n}\cdot(\partial_\mu\mathbf{n}\times\partial_\nu\mathbf{n}))^2\right] .
\end{align}
Here we used $H_{\mu\nu}\delta_{\mu\nu}=0$, which follows from $H_{\mu\nu} = -H_{\nu\mu}$. Also, $v_{\mu\nu} = v_{\nu\mu}$, and $v_{\mu\mu}\propto(\partial_\mu\mathbf{n})^2$ was used.

\noindent
(2) When integrating $\mathbf{n}$: For the compound operator written in $\mathbf{n}$,
\begin{align}
& \langle H_{\mu\nu}\rangle = 0 , \ \langle H_{\mu\nu}^2\rangle = \lambda , \nonumber\\
& \langle \delta_{\mu\nu}(\partial_\rho\mathbf{n})^2 - \partial_\mu\mathbf{n}\cdot\partial_\nu\mathbf{n} \rangle = v^2\delta_{\mu\nu} .
\end{align}
This results in a U(1) gauge scalar model:
\begin{align}
S = \int d^4x \ \left\{ \frac{1}{4}G_{\mu\nu}^2 + \frac{1}{4}v^2|D_\mu\phi|^2 + \frac{\lambda}{4}(|\phi|^2 - 1)^2 \right\} .
\end{align}
Here, the coupling constant $g$ is absorbed into the field.

\end{prop}

\begin{proof}
This was shown in Faddeev-Niemi (1999)\cite{FN99}.
\end{proof}

\section{Conclusions and discussions}

\noindent
Conclusion:

\noindent
$\bullet$ 
In this review, we considered the space and time symmetric instantons as solutions of the self-dual Yang-Mills equation with conformal symmetry in the $SU(2)$ Yang-Mills theory in the four-dimensional Euclidean space $\mathbb{E}^4$. 

\noindent
$\bullet$ 
The instanton with time translation symmetry (time independence or time translation invariance) is the well-known Prasad-Sommerfield (PS) magnetic monopole with a finite energy.  However, the PS solution gives an infinite four-dimensional action.  Therefore it gives no contribution in  the path integral and plays no role in the quantum Yang-Mills theory. 

\noindent
$\bullet$ 
On the other hand, instantons with spatial rotation symmetries give a finite four-dimensional action and hence can contribute in the quantum Yang-Mills theory. 
\\
For the spatial symmetry $SO(2) \simeq U(1) \simeq S^1$, the instanton is reduced to a hyperbolic magnetic monopole (of Atiyah) living in the three-dimensional hyperbolic space $\mathbb{H}^3$.
\\
For the spatial symmetry $SO(3) \simeq SU(2)$, the instanton is reduced to a hyperbolic vortex (of Witten-Manton) living in the two-dimensional hyperbolic space$\mathbb{H}^2$. 

By requiring the spatial symmetry $SO(2)$ or $SO(3)$ for instantons,  the four-dimensional Euclidean space $\mathbb{E}^4$ in which instantons live is inevitably transformed to the curved spacetime $\mathbb{H}^3 \times S^1$ or $\mathbb{H}^2 \times S^2$  by maintaining the conformally equivalence. 
The spatial symmetry $SO(2)$ or $SO(3)$ acts on the compact space $S^1$ or $S^2$.



\noindent
$\bullet$
Three-dimensional hyperbolic magnetic monopoles and two-dimensional hyperbolic vortices can be connected through conformal equivalence with the explicit relationship between the magnetic monopole field and the vortex field has been obtained.
This allows magnetic monopoles and vortices can be treated in a unified manner.

\noindent
$\bullet$
Both $\mathbb{H}^3$ and $\mathbb{H}^3$ are curved space $AdS_3$ and $AdS_2$ with a constant negative curvature. 
A hyperbolic monopole in $\mathbb{H}^3$ are completely determined by its holographic image on the conformal boundary two-sphere $S_\infty^2$. 
(This is different from the case of Euclidean monopoles.)
This fact enable us to reduce the non-Abelian Wilson loop operator to the Abelian Wilson loop defined by the Abelian gauge field of the vortex. 

\noindent
$\bullet$
Furthermore, by considering a symmetric instanton with a singularity in a compact subspace of spacetime, a symmetric instanton with a non-integral topological charge can be obtained, and then by dimensional reduction, a hyperbolic magnetic monopole and a hyperbolic vortex with a non-integral topological charge have been obtained.

\noindent
$\bullet$
Using the hyperbolic magnetic monopole and hyperbolic vortex obtained in this way, quark confinement was shown to be realized in the sense of Wilson's area law within the dilute gas approximation.
This is a semi-classical quark confinement mechanism originating from the unified hyperbolic magnetic monopole and hyperbolic vortex.

\noindent
Discussion:

\noindent
$\bullet$
Why does the space-time obtained by dimensional reduction have negative curvature? Is there no case where it has positive curvature?
cf: The 4-dimensional standard model can be obtained by dimensional reduction of 6-dimensional Yang-Mills theory to 4!
[Manton(1981)]

\noindent
$\bullet$
How does the gauge group change due to dimensional reduction?

\noindent
$\bullet$
How can it be extended to a large gauge group $SU(N)$?

\noindent
$\bullet$
What happens when a matter field is introduced? For example, can QCD be analyzed in the same way?

\noindent
$\bullet$
How do we incorporate quantum effects that do not maintain conformal invariance?

\section*{Acknowledgements}
This work was supported by Grant-in-Aid for Scientific Research, JSPS KAKENHI Grant
Number (C) No.23K03406.

\appendix

\section{Cho-Duan-Ge-Faddeev-Niemi decompsition}

\begin{proof} [Proof of Proposition~\ref{FNd:covariance}]

First, find the transformations of the three bases of SU(2):
\begin{align}
&U(x)\mathbf{n}(x)U(x)^\dagger = \mathbf{n}(x), \nonumber\\
&U(x)\partial_\mu\mathbf{n}(x)U(x)^\dagger = \cos\theta(x)\partial_\mu\mathbf{n}(x) \nonumber\\ 
&+ \sin\theta(x)\partial_\mu\mathbf{n}(x)\times\mathbf{n}(x), \nonumber\\
&U(x)(\partial_\mu\mathbf{n}(x)\times\mathbf{n}(x))U(x)^\dagger =  -\sin\theta(x)\partial_\mu\mathbf{n}(x) \nonumber\\ 
&+ \cos\theta(x)\partial_\mu\mathbf{n}(x)\times\mathbf{n}(x) .
\end{align}
These results are:
$
U = \exp(i\theta\mathbf{n}\cdot\frac{\boldsymbol{\sigma}}{2}) = \cos\theta\mathbf{1} + i\mathbf{n}\cdot\boldsymbol{\sigma}\sin\frac{\theta}{2}
$
using,
$
(\mathbf{a}\cdot\boldsymbol{\sigma})(\mathbf{b}\cdot\boldsymbol{\sigma}) = (\mathbf{a}\cdot\mathbf{b})\mathbf{1} + i(\mathbf{a}\times\mathbf{b})\cdot\boldsymbol{\sigma}$, $\mathbf{n}\cdot\partial_\mu\mathbf{n}=0
$,
$
\mathbf{n}\cdot(\partial_\mu\mathbf{n}\times\mathbf{n}) = 0
$,
$
It can be obtained by considering \mathbf{n}\times(\partial_\mu\mathbf{n}\times\mathbf{n})=\partial_\mu\mathbf{n}$,
$\mathbf{n}$, $\partial_\mu\mathbf{n}$, Considering that $\partial_\mu\mathbf{n}\times\mathbf{n}$
is an orthogonal basis and $U$ is a rotation of angle $\theta$ around $\mathbf{n}$, this is correct because $\mathbf{n}$ does not change and $\partial_\mu\mathbf{n}$ and $\partial_\mu\mathbf{n}\times\mathbf{n}$ in a plane perpendicular to $\mathbf{n}$ correspond to a rotation of $\theta$. See the figure.
\begin{align}
&\Omega_\mu := ig^{-1}U\partial_\mu U^\dagger \nonumber\\
=& ig^{-1}\left(\cos\frac{\theta}{2}\mathbf{1} + i\mathbf{n}\cdot\boldsymbol{\sigma}\sin\frac{\theta}{2}\right) \nonumber\\ 
& \times  \partial_\mu\left(\cos\frac{\theta}{2}\mathbf{1} - i\mathbf{n}\cdot\boldsymbol{\sigma}\sin\frac{\theta}{2}\right) \nonumber\\
=& \frac{\boldsymbol{\sigma}}{2}\cdot g^{-1}\left[\partial_\mu\theta\mathbf{n} + 2\sin\frac{\theta}{2}\cos\frac{\theta}{2}\partial_\mu\mathbf{n} + 2\sin^2\frac{\theta}{2}\partial_\mu\mathbf{n}\times\mathbf{n}\right] \nonumber\\
=& \frac{\boldsymbol{\sigma}}{2}\cdot g^{-1}\left[\partial_\mu\theta\mathbf{n} + \sin\theta\partial_\mu\mathbf{n} + (1-\cos\theta)\partial_\mu\mathbf{n}\times\mathbf{n}\right] .
\end{align}
Therefore, for $\mathbf{V}_\mu$,
\begin{align}
 U\mathbf{V}_\mu U^\dagger + \Omega_\mu 
 =& c_\mu\mathbf{n} - g^{-1}\sin\theta\partial_\mu\mathbf{n} \nonumber\\
& + g^{-1}\cos\theta\partial_\mu\mathbf{n}\times\mathbf{n} \nonumber\\
&+g^{-1}\partial_\mu\theta\mathbf{n} + g^{-1}\sin\theta\partial_\mu\mathbf{n} \nonumber\\ 
&+ g^{-1}(1-\cos\theta)\partial_\mu\mathbf{n}\times\mathbf{n} \nonumber\\
=&(c_\mu + g^{-1}\partial_\mu\theta)\mathbf{n} + g^{-1}\partial_\mu\mathbf{n}\times\mathbf{n}.
\end{align}
On the other hand, for $\mathbf{X}_\mu$,
\begin{align}
U\mathbf{X}_\mu U^\dagger =& \phi_1U\partial_\mu\mathbf{n}U^\dagger + \phi_2U(\partial_\mu\mathbf{n}\times\mathbf{n})U^\dagger \nonumber\\
=&\phi_1(\cos\theta\partial_\mu\mathbf{n} + \sin\theta\partial_\mu\mathbf{n}\times\mathbf{n}) \nonumber\\
&+ \phi_2(-\sin\theta\partial_\mu\mathbf{n} + \cos\theta\partial_\mu\mathbf{n}\times\mathbf{n}) \nonumber\\
=&(\phi_1\cos\theta - \phi_2\sin\theta)\partial_\mu\mathbf{n}\nonumber\\
& + (\phi_1\sin\theta + \phi_2\cos\theta)\partial_\mu\mathbf{n}\times\mathbf{n} .
\end{align}
twist, $\mathbf{X}_\mu$ is a vector in a plane perpendicular to $\mathbf{n}$, and it rotates in that plane under the action of $U$.
\end{proof}

\begin{proof}[Proof of Proposition~\ref{FNd:YMLagrangian}]

\noindent
(1) For $\mathbf{X}_\mu = \phi_1\partial_\mu\mathbf{n} + \phi_2\partial_\mu\mathbf{n}\times\mathbf{n}$, $\mathbf{X}_\mu\times\mathbf{X}_\nu$ becomes:
\begin{align}
\mathbf{X}_\mu\times\mathbf{X}_\nu =& (\phi_1\partial_\mu\mathbf{n} + \phi_2\partial_\mu\mathbf{n}\times\mathbf{n}) \nonumber\\
&\times(\phi_1\partial_\nu\mathbf{n} + \phi_2\partial_\nu\mathbf{n}\times\mathbf{n}) \nonumber\\
=& \phi_1^2\partial_\mu\mathbf{n}\times\partial_\nu\mathbf{n} + \phi_1\phi_2\partial_\mu\mathbf{n}\times(\partial_\nu\mathbf{n}\times\mathbf{n}) \nonumber\\
&+\phi_1\phi_2(\partial_\mu\mathbf{n}\times\mathbf{n})\times\partial_\nu\mathbf{n} \nonumber\\
&+ \phi_2^2(\partial_\mu\mathbf{n}\times\mathbf{n})\times(\partial_\nu\mathbf{n}\times\mathbf{n}) .
\end{align}
vector formula
\begin{align}
\mathbf{A}\times(\mathbf{B}\times\mathbf{C})=&(\mathbf{A}\cdot\mathbf{C})\mathbf{B} - (\mathbf{A}\cdot\mathbf{B})\mathbf{C} ,
\nonumber\\
(\mathbf{A}\times\mathbf{B})\times\mathbf{C}=&(\mathbf{A}\cdot\mathbf{C})\mathbf{B} - (\mathbf{B}\cdot\mathbf{C})\mathbf{A} ,
\nonumber\\
(\mathbf{A}\times\mathbf{B})\cdot(\mathbf{C}\times\mathbf{D}) =& (\mathbf{A}\cdot\mathbf{C})(\mathbf{B}\cdot\mathbf{D}) - (\mathbf{A}\cdot\mathbf{D})(\mathbf{B}\cdot\mathbf{C}) ,
\nonumber\\
(\mathbf{A}\times\mathbf{B})\times(\mathbf{C}\times\mathbf{D}) =&(\mathbf{A}\cdot(\mathbf{C}\times\mathbf{D}))\mathbf{B} - (\mathbf{B}\cdot(\mathbf{C}\times\mathbf{D}))\mathbf{A}
\nonumber\\
=&(\mathbf{A}\cdot(\mathbf{B}\times\mathbf{D}))\mathbf{C} - (\mathbf{A}\cdot(\mathbf{B}\times\mathbf{C}))\mathbf{D},
\end{align}
Using , we get the following relationship:
\begin{align}
\partial_\mu\mathbf{n}\times(\partial_\nu\mathbf{n}\times\mathbf{n}) &= (\partial_\mu\mathbf{n}\cdot\mathbf{n})\partial_\nu\mathbf{n} - (\partial_\mu\mathbf{n}\cdot\partial_\nu\mathbf{n})\mathbf{n}
\nonumber\\
&= -(\partial_\mu\mathbf{n}\cdot\partial_\nu\mathbf{n})\mathbf{n} , \nonumber\\
(\partial_\mu\mathbf{n}\times\mathbf{n})\times\partial_\nu\mathbf{n} &=(\partial_\mu\mathbf{n}\cdot\partial_\nu\mathbf{n})\mathbf{n} - (\mathbf{n}\cdot\partial_\nu\mathbf{n})\partial_\mu\mathbf{n}
\nonumber\\
&= (\partial_\mu\mathbf{n}\cdot\partial_\nu\mathbf{n})\mathbf{n} , \nonumber\\
(\partial_\mu\mathbf{n}\times\mathbf{n})\times(\partial_\nu\mathbf{n}\times\mathbf{n})
&=(\partial_\mu\mathbf{n}\cdot(\partial_\nu\mathbf{n}\times\mathbf{n}))\mathbf{n} \nonumber\\
&- (\mathbf{n}\cdot(\partial_\nu\mathbf{n}\times\mathbf{n}))\partial_\mu\mathbf{n} \nonumber\\
&=(\mathbf{n}\cdot(\partial_\mu\mathbf{n}\times\partial_\nu\mathbf{n}))\mathbf{n} .
\end{align}
Here we used $\partial_\mu\mathbf{n}\cdot\mathbf{n}=0=\mathbf{n}\cdot\partial_\nu\mathbf{n}$ following from $\mathbf{n}\cdot\mathbf{n}=1$.
Therefore we get:
\begin{align}
g\mathbf{X}_\mu\times\mathbf{X}_\nu &= g(\phi_1^2 + \phi_2^2)(\mathbf{n}\cdot(\partial_\mu\mathbf{n}\times\partial_\nu\mathbf{n}))\mathbf{n} \nonumber\\
&= -g^2(\phi_1^2 + \phi_2^2)H_{\mu\nu}\mathbf{n} .
\end{align}
Here we used $H_{\mu\nu} = -g^{-1}\mathbf{n}\cdot(\partial_\mu\mathbf{n}\times\partial_\nu\mathbf{n})$. The field strength can be calculated as follows.
\begin{align}
\mathbf{F}_{\mu\nu}[\mathbf{A}] :=& \partial_\mu\mathbf{A}_{\nu} - \partial_\nu\mathbf{A}_{\mu} + g\mathbf{A}_{\mu}\times\mathbf{A}_{\nu} \nonumber\\
=&\mathbf{F}_{\mu\nu}[\mathbf{V}] + g\mathbf{X}_\mu\times\mathbf{X}_\nu + D_\mu[\mathbf{V}]\mathbf{X}_\nu - D_\nu[\mathbf{V}]\mathbf{X}_\mu ,
\\
\mathbf{F}_{\mu\nu}[\mathbf{V}] := & \partial_\mu\mathbf{V}_{\nu} - \partial_\nu\mathbf{V}_{\mu} + g\mathbf{V}_{\mu}\times\mathbf{V}_{\nu} \nonumber\\
=&\mathbf{n}(\partial_\mu C_\nu - \partial_\nu c_\mu) - g^{-1}[\mathbf{n}\cdot(\partial_\mu\mathbf{n}\times\partial_\nu\mathbf{n})]\mathbf{n}.
\end{align}
where 
\begin{align}
&D_\mu[\mathbf{V}]\mathbf{X}_\nu - D_\nu[\mathbf{V}]\mathbf{X}_\mu \nonumber\\
=& \partial_\mu\mathbf{X}_\nu + g\mathbf{V}_\mu\times\mathbf{X}_\nu - (\mu\leftrightarrow\nu) \nonumber\\
=& \partial_\mu(\phi_1\partial_\nu\mathbf{n} + \phi_2\partial_\nu\mathbf{n}\times\mathbf{n}) 
\nonumber\\ 
&+ g(c_\mu\mathbf{n} + g^{-1}\partial_\mu\mathbf{n}\times\mathbf{n})\times(\phi_1\partial_\nu\mathbf{n} + \phi_2\partial_\nu\mathbf{n}\times\mathbf{n})\nonumber\\
&-(\mu\leftrightarrow\nu) \nonumber\\
=&\partial_\mu\phi_1\partial_\nu\mathbf{n} + \phi_1\partial_\mu\partial_\nu\mathbf{n} + \partial_\mu\phi_2\partial_\nu\mathbf{n}\times\mathbf{n} 
\nonumber\\ 
&+ \phi_2\partial_\mu\partial_\nu\mathbf{n}\times\mathbf{n} + \phi_2\partial_\nu\mathbf{n}\times\partial_\mu\mathbf{n} \nonumber\\
&+gc_\mu\phi_1\mathbf{n}\times\partial_\nu\mathbf{n} + gc_\mu\phi_2\mathbf{n}\times(\partial_\nu\mathbf{n}\times\mathbf{n}) \nonumber\\ 
&+(\partial_\mu\mathbf{n}\times\mathbf{n})\times\partial_\nu\mathbf{n}\phi_1 \nonumber\\
&+(\partial_\mu\mathbf{n}\times\mathbf{n})\times(\partial_\nu\mathbf{n}\times\mathbf{n})\phi_2-(\mu\leftrightarrow\nu) \nonumber\\
=&\partial_\mu\phi_1\partial_\nu\mathbf{n} + \partial_\mu\phi_2\partial_\nu\mathbf{n}\times\mathbf{n} - \phi_2\partial_\mu\mathbf{n}\times\partial_\nu\mathbf{n} \nonumber\\
&+gc_\mu\phi_1\mathbf{n}\times\partial_\nu\mathbf{n} + gc_\mu\phi_2\partial_\nu\mathbf{n} + \phi_1(\partial_\mu\mathbf{n}\cdot\partial_\nu\mathbf{n})\mathbf{n} \nonumber\\ 
&- \phi_1(\mathbf{n}\cdot\partial_\nu\mathbf{n})\partial_\mu\mathbf{n}\phi_1 +\phi_2(\mathbf{n}\cdot(\partial_\mu\mathbf{n}\times\partial_\nu\mathbf{n}))\mathbf{n}\nonumber\\
=&\partial_\mu\phi_1\partial_\nu\mathbf{n} + gc_\mu\phi_2\partial_\nu\mathbf{n} -\partial_\nu\phi_1\partial_\mu\mathbf{n} - gC_\nu\phi_2\partial_\mu\mathbf{n} \nonumber\\
&+\partial_\mu\phi_2\partial_\nu\mathbf{n}\times\mathbf{n} + gc_\mu\phi_1\mathbf{n}\times\partial_\nu\mathbf{n} - \partial_\nu\phi_2\partial_\mu\mathbf{n}\times\mathbf{n} \nonumber\\ 
& -gC_\nu\phi_1\mathbf{n}\times\partial_\mu\mathbf{n} \nonumber\\
=&(\partial_\mu\phi_1+gc_\mu\phi_2)\partial_\nu\mathbf{n} -(\partial_\nu\phi_1+gC_\nu\phi_2)\partial_\mu\mathbf{n} \nonumber\\
&+(\partial_\mu\phi_2-gc_\mu\phi_1)\partial_\nu\mathbf{n} \times \mathbf{n} \nonumber\\ 
&-(\partial_\nu\phi_2-gC_\nu\phi_1)\partial_\mu\mathbf{n} \times \mathbf{n} .
\end{align}
Also, if we set $\phi = \phi_1 + i\phi_2$, we get:
\begin{align}
D_\mu\phi\equiv&\partial_\mu\phi -igc_\mu\phi = \partial_\mu\phi_1+i\partial_\mu\phi_2 - gc_\mu(i\phi_1 - \phi_2) \nonumber\\
=&(\partial_\mu\phi_1+gc_\mu\phi_2) + i(\partial_\mu\phi_2 - gc_\mu\phi_1) \nonumber\\
=&(\partial_\mu\phi_1-igc_\mu\phi_1) + i(\partial_\mu\phi_2-igc_\mu\phi_2) .
\end{align}
First of all, 
\begin{align}
&(D_\mu\phi_1\partial_\nu\mathbf{n} - D_\nu\phi_1\partial_\mu\mathbf{n})\cdot(D_\mu\phi_1\partial_\nu\mathbf{n} - D_\nu\phi_1\partial_\mu\mathbf{n}) \nonumber\\
=&(D_\mu\phi_1)^2(\partial_\nu\mathbf{n})^2 - D_\mu\phi_1D_\nu\phi_1\partial_\mu\mathbf{n}\cdot\partial_\nu\mathbf{n} \nonumber\\ 
&  - D_\mu\phi_1D_\nu\phi_1\partial_\mu\mathbf{n}\cdot\partial_\nu\mathbf{n} +(D_\nu\phi_1)^2(\partial_\mu\mathbf{n})^2 \nonumber\\
=&2(D_\mu\phi_1)^2(\partial_\nu\mathbf{n})^2 - 2(D_\mu\phi_1)(D_\nu\phi_1)\partial_\mu\mathbf{n}\cdot\partial_\nu\mathbf{n}  .
\end{align}
Next, 
\begin{align}
&(D_\mu\phi_2\partial_\nu\mathbf{n}\times \mathbf{n} -D_\nu\phi_2\partial_\mu\mathbf{n}\times \mathbf{n})\nonumber\\ 
& \cdot (D_\mu\phi_2\partial_\nu\mathbf{n}\times \mathbf{n} -D_\nu\phi_2\partial_\mu\mathbf{n}\times \mathbf{n}) \nonumber\\
=&(D_\mu\phi_2)^2(\partial_\nu\mathbf{n}\times \mathbf{n})\cdot(\partial_\nu\mathbf{n}\times \mathbf{n}) \nonumber\\ 
& -D_\mu\phi_2D_\nu\phi_2(\partial_\mu\mathbf{n}\times \mathbf{n})\cdot(\partial_\nu\mathbf{n}\times \mathbf{n}) \nonumber\\
&-(D_\mu\phi_2)(D_\nu\phi_2)(\partial_\mu\mathbf{n}\times \mathbf{n})\cdot(\partial_\nu\mathbf{n}\times \mathbf{n}) \nonumber\\
&+(D_\nu\phi_2)^2(\partial_\mu\mathbf{n}\times \mathbf{n})\cdot(\partial_\mu\mathbf{n}\times \mathbf{n}) \nonumber\\
=&2(D_\mu\phi_2)^2((\partial_\nu\mathbf{n}\cdot\partial_\nu\mathbf{n})(\mathbf{n}\cdot\mathbf{n})-(\partial_\nu\mathbf{n}\cdot\mathbf{n})(\mathbf{n}\cdot\partial_\nu\mathbf{n})) \nonumber\\
&-2D_\mu\phi_2D_\nu\phi_2((\partial_\mu\mathbf{n}\cdot\partial_\nu\mathbf{n})(\mathbf{n}\times \mathbf{n})
\nonumber\\ 
&  -(\partial_\mu\mathbf{n}\cdot \mathbf{n})(\mathbf{n}\cdot\partial_\nu\mathbf{n})) \nonumber\\
=&2(D_\mu\phi_2)^2(\partial_\nu\mathbf{n})^2 -2(D_\mu\phi_2)(D_\nu\phi_2)(\partial_\mu\mathbf{n}\cdot\partial_\nu\mathbf{n}) .
\end{align}
Putting the two equations together leads to 
\begin{align}
&2((D_\mu\phi_1)^2 + (D_\mu\phi_2)^2)(\partial_\rho\mathbf{n})^2 \nonumber\\
&-2((D_\mu\phi_1)(D_\nu\phi_1) + (D_\mu\phi_2)(D_\nu\phi_2))
(\partial_\mu\mathbf{n}\cdot\partial_\nu\mathbf{n}) \nonumber\\
=&2[\delta_{\mu\nu}(\partial_\rho\mathbf{n})^2 - \partial_\mu\mathbf{n}\cdot\partial_\nu\mathbf{n}](D_\mu\phi_1D_\nu\phi_1 + D_\mu\phi_2D_\nu\phi_2) \nonumber\\
=&[\delta_{\mu\nu}(\partial_\rho\mathbf{n})^2 - \partial_\mu\mathbf{n}\cdot\partial_\nu\mathbf{n}][(D_\mu\phi)^*(D_\nu\phi) + {\rm{h.c.}}] .
\end{align}
where we have used
\begin{align}
&(D_\mu\phi)^*(D_\nu\phi) + {\rm{h.c.}} 
\nonumber\\ 
&= 2[(D_\mu\phi_1)(D_\nu\phi_1) + (D_\mu\phi_2)(D_\nu\phi_2)] .
\end{align}
Moreover, 
\begin{align}
&2(D_\mu\phi_1\partial_\nu\mathbf{n} - D_\nu\phi_1\partial_\mu\mathbf{n})\nonumber\\ 
& \cdot (D_\mu\phi_2\partial_\nu\mathbf{n}\times\mathbf{n} - D_\nu\phi_2\partial_\mu\mathbf{n}\times\mathbf{n} ) \nonumber\\
=&2(D_\mu\phi_1)(D_\mu\phi_2)\partial_\nu\mathbf{n}\cdot(\partial_\nu\mathbf{n}\times\mathbf{n}) 
\nonumber\\ 
&-2(D_\mu\phi_1)(D_\nu\phi_2)\partial_\nu\mathbf{n}\cdot(\partial_\mu\mathbf{n}\times\mathbf{n}) \nonumber\\
&-2(D_\nu\phi_1)(D_\mu\phi_2)\partial_\mu\mathbf{n}\cdot(\partial_\nu\mathbf{n}\times\mathbf{n})\nonumber\\ 
& +2D_\nu\phi_1D_\nu\phi_2\partial_\mu\mathbf{n}\cdot(\mathbf{n}\times\mathbf{n}) \nonumber\\
=&4(D_\mu\phi_1)(D_\nu\phi_2)\mathbf{n}\cdot(\partial_\mu\mathbf{n}\times\partial_\nu) \nonumber\\
=&-4gH_{\mu\nu}(D_\mu\phi_1)(D_\nu\phi_2) .
\end{align}
where by taking into account 
\begin{align}
&(D_\mu\phi)^*(D_\nu\phi) \nonumber\\
=&(\partial_\mu\phi - igc_\mu\phi)^*(\partial_\nu\phi - igC_\nu\phi) \nonumber\\
=&(\partial_\mu\phi_1 + i\partial_\mu\phi_2 - igc_\mu(\phi_1+i\phi_2))^*\nonumber\\ 
& \cdot (\partial_\nu\phi_1 + i\partial_\nu\phi_2 - igC_\nu(\phi_1+i\phi_2)) \nonumber\\
=&((\partial_\mu\phi_1 + gc_\mu\phi_2)+i(\partial_\mu\phi_2 - gc_\mu\phi_1))^* \nonumber\\ 
& \cdot (\partial_\nu\phi_1 + gC_\nu\phi_2+i(\partial_\nu\phi_2 - gC_\nu\phi_1)) \nonumber\\
=&(D_\mu\phi_1 -iD_\mu\phi_2)(D_\nu\phi_1 +iD_\nu\phi_2)\nonumber\\
=&(D_\mu\phi_1)(D_\nu\phi_1)+(D_\mu\phi_2)(D_\nu\phi_2)
\nonumber\\ 
&  +i(D_\mu\phi_1)(D_\nu\phi_2) -i(D_\nu\phi_1)(D_\mu\phi_2)
\end{align}
yields
\begin{align}
-i(D_\mu\phi)^*(D_\nu\phi)
=&(D_\mu\phi_1)(D_\nu\phi_2)  -(D_\nu\phi_1)(D_\mu\phi_2) \nonumber\\ 
&-i(D_\mu\phi_a)(D_\nu\phi_a) ,
\end{align}
and
\begin{align}
-iH_{\mu\nu}(D_\mu\phi)^*(D_\nu\phi) = 2H_{\mu\nu}(D_\mu\phi_1)(D_\nu\phi_2) ,
\end{align}
where it should be noted that 
\begin{align}
&-4gH_{\mu\nu}(D_\mu\phi_1)(D_\nu\phi_2) 
= 2igH_{\mu\nu}(D_\mu\phi)^*(D_\nu\phi)\nonumber\\
=&igH_{\mu\nu}(D_\mu\phi)^*(D_\nu\phi) - igH_{\mu\nu}(D_\mu\phi)(D_\nu\phi)^*\nonumber\\
=&igH_{\mu\nu}(D_\mu\phi)^*(D_\nu\phi) + {\rm{h.c.}} .
\end{align}

\end{proof}

\section{ADHM construction for instantons}

Since the self-dual (and anti-self-dual) Yang-Mills equations are nonlinear partial differential equations, they are very difficult to solve.
However, fortunately, a method is known that gives all instanton solutions using a matrix written in quaternions that satisfies algebraic constraints.


First, we define quaternions.

\begin{definition}[Quaternions]
A \textbf{quaternion}\index{quaternion} can be written as 
\begin{align}
q=q_\mu e_\mu \ (\mu = 1, 2, 3, 4) ,  \ q_\mu\in\mathbb{R} .
\end{align}
Here, $e_4$ represents the unit component: $e_4 = 1$ and $e_1, e_2, e_3$ satisfy the relations:
\begin{align}
e_1^2 = e_2^2 =e_3^2 = 1 , \ e_1e_2 =e_3=-e_2e_1 , \ etc.
\end{align}
We use the $2\times2$ Pauli matrix representation of the quaternion:
\begin{align}
e_j = -i\sigma_j (j = 1, 2, 3), \ e_4 = 1.
\end{align}
Hence, the quaternion $q = q_je_j$ with $q_4 = 0$ (called a pure quaternion) can be identified with an element of the Lie algebra $su(2)$.
The quaternion $q_j=0$ $(j=1,2,3)$ (called the real quaternion), $q = q_4e_4=q_4$ can be identified with the real number $q_4$.
For a matrix $M$, $M^\dagger$ represents the complex conjugate transpose, meaning that $q^\dagger = -q_je_j + q_4e_4$ for the quaternion $q = q_\mu e_\mu = q_je_j + q_4e_4$.

\end{definition}

We consider finding the second Chern class $c_2 = n$ (number of instantons) in the self-dual $SU(2)$ connection.

\begin{definition}[ADHM data]

First, we introduce the matrix $\Delta$, which is a $(n+1)\times n$-dimensional matrix whose elements take values in the quaternion $\mathbb{H}$:
\begin{align}
\Delta(x) = A + Bx , \ A, B\in M(n + 1, n; \mathbb{H})
\end{align}
These matrices $A, B$ are called \textbf{ADHM data}\index{ADHM data}. ADHM is an abbreviation for Atiyah, Drinfeld, Hitchin and Manin (Atiyah-Drinfeld-Hitchin-Manin(1978)\cite{ADHM78}).
Here, $A_{ab}$, $B_{ab}$ and $x$ take values in $SU(2)$ and 
$x$ on the right hand side of $\Delta(x)$ is the quaternion corresponding to the point $(x^\mu)$ in $\mathbb{R}^4$: ($\sigma_A$ are the Pauli matrices.)
\begin{align}
x = x^\mu e_\mu = x^0 - i\sigma_Ax^A .
\end{align}
Here we have defined 
\begin{align}
e_\mu := (1, -i\sigma_A)(A = 1, 2, 3).
\end{align}

Next, the row vector (or column vector) $V(x)$ with $(n + 1)$ components is determined so that it satisfies the following property.

(i) [Orthogonality of $\Delta$ and $V$]
\begin{align}
V^\dagger(x)\Delta(x) = 0 \Leftrightarrow \Delta^\dagger(x)V(x) = 0.
\end{align}

(ii) [Normalization and existence of the projection operator]
\begin{align}
V^\dagger(x)V(x) =& \bm{1} , \\
V(x)V^\dagger(x) =& \bm{1} - \Delta(x)\frac{1}{\Delta^\dagger(x)\Delta(x)}\Delta^\dagger(x):=P(x).
\end{align}

(iii) [Invertibility condition] or [Real non-singular condition]
The columns of $\Delta(x)$ span an $N$-dimensional quaternion space for all $x$. In other words, there exists an $n\times n$-matrix $R(x)$ of real quaternions that has an inverse for all $x$:
\begin{align}
\Delta(x)^\dagger \Delta(x) = R(x).
\end{align}

\end{definition}

\begin{prop}[ADHM construction of instanton solutions]
\noindent
(1)  
For $\Delta(x)$ with given ADHM data, if properties (i), (ii), and (iii) are satisfied, then the $su(2)$ gauge field as a pure quaternion
\begin{align}
\mathscr{A}_\mu(x) = V(x)^\dagger\partial_\mu V(x) \in su(2)
\end{align}
becomes a self-dual $SU(2)$ connection and gives the Chern class $c_2 = n$.
That is, the field strength $\mathscr{F}_{\mu\nu}(x)$ constructed from $\mathscr{A}_\mu(x)$ becomes self-dual $\mathscr{F}_{\mu\nu}^* = \mathscr{F}_{\mu\nu}$ and has instanton number $n$.

\noindent
(2)  
The number of real parameters in the ADHM data is equal to $8n$ (which is required to generate a general $n$-instanton solution).
\end{prop}

\begin{proof}
\noindent
(1) 
Using $V^\dagger V = 1$ and $\mathscr{A}_\mu = iV^\dagger \partial_\mu V$, we have
\begin{align}
\mathscr{F}_{\mu\nu} =& \partial_\mu \mathscr{A}_\nu - \partial_\nu \mathscr{A}_\mu - i\mathscr{A}_\mu\mathscr{A}_\nu + i\mathscr{A}_\nu\mathscr{A}_\mu \nonumber\\
=&i\partial_\mu(V^\dagger \partial_\nu V) - i\partial_\nu(V^\dagger \partial_\mu V) \nonumber\\
&+iV^\dagger\partial_\mu V V^\dagger\partial_\nu V - iV^\dagger\partial_\nu V V^\dagger\partial_\mu V. \nonumber\\
=&i\partial_\mu V^\dagger \partial_\nu V - i\partial_\nu V^\dagger\partial_\mu V - i\partial_\mu V^\dagger VV^\dagger \partial_\nu V \nonumber\\ &+ i\partial_\nu V^\dagger VV^\dagger \partial_\mu V \nonumber\\
=&i\partial_\mu V^\dagger(1 - VV^\dagger)\partial_\nu V - i\partial_\nu V^\dagger(1 - VV^\dagger)\partial_\mu V.
\end{align}
Here the operator $1 - VV^\dagger$ is a projection operator onto the subspace perpendicular to $V$. Using the definition of $R$: $\Delta^\dagger\Delta = R$ and $V^\dagger\Delta = 0$, we can write
\begin{align}
1 - P = 1 - VV^\dagger = \Delta R^{-1}\Delta^\dagger ,
\end{align}
and therefore obtain:
\begin{align}
\mathscr{F}_{\mu\nu} = i\partial_\mu V^\dagger \Delta R^{-1} \Delta^\dagger \partial_\nu V - i\partial_\nu V^\dagger \Delta R^{-1} \Delta^\dagger \partial_\mu V .
\end{align}
Using $\partial_\mu V^\dagger \Delta = -V^\dagger \partial_\mu\Delta$ following from the differentiation of $V^\dagger\Delta = 0$ and $\Delta^\dagger\partial_\mu V = -\partial_\mu\Delta^\dagger V$ following from its conjugate, we have
\begin{align}
\mathscr{F}_{\mu\nu} = iV^\dagger\partial_\mu\Delta R^{-1}\partial_\nu \Delta^\dagger V - iV^\dagger\partial_\nu\Delta R^{-1}\partial_\mu \Delta^\dagger V.
\end{align}
The relation holds
\begin{align}
\partial_\mu\Delta(x) = -\tilde{1}_ne_\mu .
\end{align}
Here $\tilde{1}_n$ is a constant real $(n + 1)\times n$ matrix whose first row is zero and whose remaining $n \times n$ block is  the identity matrix.  
For
\begin{align}
\Delta(x) =
\begin{pmatrix} 
L\\ 
M - x {1}_n
\end{pmatrix}
=
\begin{pmatrix}
L\\
M
\end{pmatrix}
-x^\rho
\begin{pmatrix}
0\\
e_\rho {1}_n
\end{pmatrix} ,
\end{align}
we have
\begin{align}
\partial_\mu \Delta(x) =-
\begin{pmatrix}
0\\
e_\mu {1}_n
\end{pmatrix}.
\end{align}
Therefore, we find
\begin{align}
\mathscr{F}_{\mu\nu} = iV^\dagger\tilde{1}_nR^{-1}\eta_{\mu\nu}\tilde{1}_n^\dagger V , \ \eta_{\mu\nu}:=e_\mu e_\nu^\dagger - e_\nu e_\mu^\dagger .
\end{align}
Here we can see that $\eta_{\mu\nu}$ is self-dual:
\begin{align}
{}^*\eta_{\mu\nu} = \eta_{\mu\nu}
\end{align}
Hence, we have 
\begin{align}
{}^*\mathscr{F}_{\mu\nu} = \mathscr{F}_{\mu\nu}
\end{align}
If we adopt $\Delta(x) = A + Bx = A + Be_\rho x_\rho$, we have 
\begin{align}
\partial_\mu\Delta(x) = Be_\mu , \ \partial_\mu\Delta^\dagger(x) = e_\mu B^\dagger ,
\end{align}
which yields
\begin{align}
\mathscr{F}_{\mu\nu} = iV^\dagger BR^{-1}\eta_{\mu\nu}B^\dagger V.
\end{align}

\noindent
(2)  
In fact, the ADHM matrix
$\mathscr{M}:=\begin{pmatrix}
L\\
M
\end{pmatrix}$ has $4n$ real parameters in the column vector (or row vector) $L$ and $2n(n+1)$ real parameters in the symmetric matrix $M$. The constraint $M^\dagger M = R_0$ removes $\frac{3}{2}n(n-1)$ degrees of freedom. This is because the pure quaternion part is set to 0, resulting in 3 for each of the upper triangular elements of the matrix $R$.

Additionally, the ADHM data has extra degrees of freedom corresponding to the transformation:
\begin{align}
\Delta(x) \to
\begin{pmatrix}
q&0\\
0&\mathscr{O}
\end{pmatrix}
\Delta(x)\mathscr{O}^{-1} ,
\end{align}
where $\mathscr{O}$ is a constant $n \times n$ real orthogonal matrix, $q$ is a constant unit quaternion, and the decomposition into blocks is $\Delta(x)=\begin{pmatrix}
L\\
M - x {1}_N
\end{pmatrix}$. This transformation rotates the components of the vector $V$ as defined by $V(x)^\dagger\Delta(x) = 0$, but does not change $\mathscr{A}$, which follows from $\mathscr{A} = V^\dagger dV$. The matrix $\mathscr{O}\in O(n)$ has $\frac{1}{2}n(n-1)$ parameters, and although the unit quaternion $q$ has three parameters, we do not subtract the last three, since they balance the three coming from the $SU(2)$ rotation. Therefore, we obtain the expected result:
\begin{align}
4n + 2n(n + 1) -\frac{3}{2}n(n-1) - \frac{1}{2}n(n-1) = 8n.
\end{align}

\end{proof}


\begin{rem}
The construction of $V$ differs slightly depending on the choice of gauge group $G$. The above is for $G = Sp(N)$: $Sp(1) = SU(2)$. In the case of $Sp(N)$, the $(n+1)\times n$ dimension is changed to $(n+N)\times n$. The case of $G = U(N)$ is different. 
\end{rem}

\begin{rem}
At first glance, $\mathscr{A}_\mu = V + \partial_\mu V$ looks like a \textbf{pure gauge}\index{pure gauge} form, but it is not the case. This is because $V$ is not a square matrix. If $\mathscr{A}_\mu$ were a pure gauge, the field strength would be zero and an instanton could not be constructed.
\end{rem}

By calculating $\mathscr{F} = d\mathscr{A} + \mathscr{A} \wedge \mathscr{A}$, we obtain $\mathscr{A} = V^\dagger dV$:
\begin{align}
d\mathscr{A} &= d(V^\dagger dV) = dV^\dagger \wedge dV + Vd^2V \nonumber\\ &= dV^\dagger \wedge dV , \nonumber\\ 
\mathscr{A} \wedge \mathscr{A} &= V^\dagger dV \wedge V^\dagger dV = -dV^\dagger V \wedge V^\dagger dV \nonumber\\ &
= -dV^\dagger \wedge PdV .
\end{align}
Therefore, we have
\begin{align}
\mathscr{F} = dV^\dagger \wedge dV - dV^\dagger PdV = dV^\dagger(1 - P)dV \neq 0.
\end{align}

\begin{rem}
The gauge transformation $\mathscr{A}(x) \to \mathscr{A}'(x) = U^\dagger(x)\mathscr{A}(x)U(x) + U^\dagger(x)dU(x)$ is realized by multiplying $V(x)$ by the unit quaternion $q(x)$ from the right:
\begin{align}
V(x) \to V'(x) = V(x)q(x) .
\end{align}
In fact, using $q^\dagger(x)q(x) = 1$ and $V^\dagger(x)V(x) = 1$, we find 
\begin{align}
\mathscr{A}' &= {V'}^\dagger dV' = (Vq)^\dagger d(Vq) = q^\dagger V^\dagger dVq + q^\dagger V^\dagger Vdq \nonumber\\
&=q^\dagger\mathscr{A}q + q^\dagger dq ,
\end{align}
where we have used the fact that the unit quaternion can be identified with $SU(2)$.
\end{rem}

\begin{rem}
The projection operator $P(=VV^\dagger)$ is related to the covariant differential operator $D = d + \mathscr{A}(\mathscr{A} = V^\dagger dV)$ as 
\begin{align}
Pd\bar{\Psi} = VD\Psi , \ \bar{\Psi} := V\Psi .
\end{align}
In fact, since $VV^\dagger = 1$, we have:
\begin{align}
Pd\bar{\Psi} &= Pd(V\Psi) = PdV\Psi + PVd\Psi \nonumber\\
&=VV^\dagger dV\Psi + VV^\dagger Vd\Psi
=V(V^\dagger dV + d)\Psi .
\end{align}
\end{rem}

In general, if we adopt
\begin{align}
\Delta =
\begin{pmatrix}
\lambda_1&\lambda_2&\dots&\lambda_n\\
\alpha_1& & &\\
&\alpha_2& &\text{\huge{0}}\\
& &\ddots&\\
\text{\huge{0}}& & &\alpha_n
\end{pmatrix}
+
\begin{pmatrix}
0& &\dots&0\\
1& & &\\
&1& &\text{\huge{0}}\\
& &\ddots&\\
\text{\huge{0}}& & &1
\end{pmatrix} x ,
\end{align}
we can obtain 't~Hooft type $n$-instantons.

\begin{example}[$G =Sp(1) \simeq SU(2)$ $n = 1$(1-instanton)]

Consider the $1$-instanton solution.
For the ADHM operator
\begin{align}
\Delta =& A + Bx =
\begin{pmatrix}
\lambda \\
\alpha
\end{pmatrix}
+
\begin{pmatrix}
0 \\
x
\end{pmatrix}
=
\begin{pmatrix}
\lambda \\
x +\alpha
\end{pmatrix}
\nonumber\\ =&
\begin{pmatrix}
\lambda \\
x_1
\end{pmatrix}
\ (x_1:= x+ \alpha) ,
\end{align}
we have
\begin{align}
\Delta^\dagger\Delta =&
\begin{pmatrix}
\lambda^\dagger & x_1^\dagger
\end{pmatrix}
\begin{pmatrix}
\lambda \\
x_1
\end{pmatrix}
=\lambda^\dagger\lambda + x_1^\dagger x_1\in \mathbb{R}.
\end{align}
For the vector
\begin{align}
V = \Xi^{-1/2}
\begin{pmatrix}
1 \\
v
\end{pmatrix} ,
\end{align}
we find $v$ as
\begin{align}
\Delta^\dagger V =& 0 \Leftrightarrow
\begin{pmatrix}
\lambda^\dagger & x_1^\dagger
\end{pmatrix}
\begin{pmatrix}
1\\
v
\end{pmatrix}
=
0
\Leftrightarrow \lambda^\dagger + x_1^\dagger v = 0
\nonumber\\ 
 \Leftrightarrow & |x_1|^2v = -\lambda^\dagger x_1 \Leftrightarrow v = -\lambda^\dagger\frac{x_1}{|x_1|^2} .
\end{align}
$\Xi$ is found from normalization:
\begin{align}
&1 = V^\dagger V = \Xi^{-1}
(1 \ v^\dagger)
\begin{pmatrix}
1 \\
v
\end{pmatrix}
=\Xi^{-1}(1 + |v|^2) 
\nonumber\\ & \Rightarrow \Xi = 1 + |v|^2 .
\end{align}
By taking $\alpha = a_\mu e_\mu$, we have
\begin{align}
x_1 = x + \alpha = (x_\mu + a_\mu)e_\mu \Rightarrow |x_1|^2 = (x_\mu + a_\mu)^2 .
\end{align}
Therefore, we find
\begin{align}
&V = \Xi^{-1/2}
\begin{pmatrix}
1 \\
v
\end{pmatrix}
=\frac{1}{\sqrt{1 + |v|^2}}
\begin{pmatrix}
1 \\
v
\end{pmatrix}, 
\ v = -\lambda\frac{x_\mu + a_\mu}{(x + a)^2}e_\mu ,
\nonumber\\
&|v|^2 = \frac{|\lambda|^2}{(x + a)^2} , \quad
\Xi = 1 + |v|^2 = 1 + \frac{\lambda^2}{(x + a)^2} .
\end{align}

\end{example}

\begin{example}[$G =Sp(1) \simeq SU(2)$, $n = 2$(2-instanton)]

Consider the $2$-instanton solution.
For the ADHM operator
\begin{align}
\Delta &= A + Bx =
\begin{pmatrix}
\lambda_1& \lambda_2 \\
\alpha&0 \\
0&\beta
\end{pmatrix}
+
\begin{pmatrix}
0&0 \\
x&0 \\
0&x
\end{pmatrix}
=
\begin{pmatrix}
\lambda_1&\lambda_2 \\
x_1&0 \\
0&x_2
\end{pmatrix},
\nonumber\\
& x_1:= x+ \alpha \, \ x_2:= x+ \beta ,
\end{align}
we have
\begin{align}
\Delta^\dagger\Delta &=
\begin{pmatrix}
\lambda_1^\dagger&x_1^\dagger&0 \\
\lambda_2^\dagger&0&x_2^\dagger
\end{pmatrix}
\begin{pmatrix}
\lambda_1&\lambda_2 \\
x_1&0 \\
0&x_2
\end{pmatrix}
\nonumber\\ &=
\begin{pmatrix}
\lambda_1^\dagger\lambda_1 + x_1^\dagger x_1&\lambda_1^\dagger\lambda_2 + x_1^\dagger x_2 \\
\lambda_2^\dagger\lambda_1&\lambda_2^\dagger\lambda_2 + x_2^\dagger x_2
\end{pmatrix} .
\end{align}
For the vector
\begin{align}
V = \Xi^{-1/2}
\begin{pmatrix}
1\\
v_1\\
v_2
\end{pmatrix} ,
\end{align}
we find $v_1,v_2$ as
\begin{align}
& \Delta^\dagger V= 0 \Leftrightarrow
\begin{pmatrix}
\lambda_1^\dagger&x_1^\dagger&0 \\
\lambda_2^\dagger&0&x_2^\dagger
\end{pmatrix}
\begin{pmatrix}
1\\
v_1\\
v_2
\end{pmatrix}
=
\begin{pmatrix}
0\\
0
\end{pmatrix}
\\
\Leftrightarrow&
\left\{\,
\begin{aligned}
&\lambda_1^\dagger + x_1^\dagger v_1 = 0\\
&\lambda_2^\dagger + x_2^\dagger v_2 = 0
\end{aligned}
\right.
\Leftrightarrow
\left\{\,
\begin{aligned}
&|x_1|^2 v_1 = -\lambda_1^\dagger x_1\\
&|x_2|^2 v_2 = -\lambda_2^\dagger x_2
\end{aligned}
\right.
\nonumber\\ 
 \Leftrightarrow&
\left\{\,
\begin{aligned}
&v_1 = -\lambda_1^\dagger \frac{x_1}{|x_1|^2}\\
&v_2 = -\lambda_2^\dagger \frac{x_2}{|x_2|^2}
\end{aligned}
\right. .
\end{align}
$\Xi$ is found from the normalized $V^\dagger V = 1$:
\begin{align}
&1 = \Xi^{-1/2}
\begin{pmatrix}
1&v_1^\dagger&v_2^\dagger
\end{pmatrix}
\begin{pmatrix}
1\\
v_1\\
v_2
\end{pmatrix}
\Xi^{-1/2} 
\nonumber\\ &= \Xi^{-1}(1 + |v_1|^2 + |v_2|^2) 
\nonumber\\ &
\Rightarrow \Xi = 1 + |v_1|^2 + |v_2|^2 , 
\nonumber\\ &
 |v_a|^2 = |\lambda_a|^2\frac{|x_a|^2}{|x_a|^4} = \frac{\lambda_a^2}{|x_a|^2} \ (a = 1, 2) .
\end{align}
By adopting $\alpha = a{1}_2, \beta = b{1}_2$, we have
\begin{align}
x_1 &= x + \alpha = (x_0 + a){1}_2 + x_je_j \Rightarrow |x_1|^2 
\nonumber\\ &= (x_0 + a)^2 + r^2 , \\
x_2 &= x + \beta = (x_0 + b){1}_2 + x_je_j \Rightarrow |x_2|^2 
\nonumber\\ &= (x_0 + b)^2 + r^2 .
\end{align}
Hence, we obtain the superpotential:
\begin{align}
\Xi &= 1 + \frac{\lambda_1^2}{|x_1|^2} + \frac{\lambda_2^2}{|x_2|^2} 
\nonumber\\ &= 1 + \frac{\lambda_1^2}{(x_0 + a)^2 + r^2} + \frac{\lambda_2^2}{(x_0 + b)^2 + r^2} .
\end{align}
This represents a 2-instanton consisting of two 1-instantons located on the $x_0$ axis.
\end{example}

\section{$S^1$-equivariant ADHM construction}

Using the quaternion representation $x = x_\mu e_\mu$ of the coordinates of a point $x\in\mathbb{R}^4$, the conformal group of $\mathbb{R}^4$ acts as the \textbf{M\"{o}biustransform}\index{M\"{o}biustransform}:
\begin{align}
x\mapsto x' = (Ax + B)(Cx + D)^{-1}.
\end{align}
The commuting circle $(S^1)$ action is given by the rotation group:
\begin{align}
\begin{pmatrix}
A&C\\
B&D
\end{pmatrix}
=
\begin{pmatrix}
\cos\frac{\alpha}{2}&\pm\sin\frac{\alpha}{2}\\
\mp\sin\frac{\alpha}{2}&\cos\frac{\alpha}{2}
\end{pmatrix}. \label{eq:6}
\end{align}
This circle action fixes the 2-sphere $S^2$ described by the unit pure quaternion $x$. That is, for $x = x_j e_j$, $x^2 = 1 \Leftrightarrow x_j^2 = 1$. In fact, we can confirm this:
\begin{align}
x^2 &= x_je_jx_ke_k = x_jx_ke_je_k \nonumber\\
&=x_jx_k(\delta_{jk}1 + \epsilon_{jk\ell}e_\ell) = x_jx_j .
\end{align}
This two-sphere $S^2$ is the two-sphere boundary $S_\infty^2$ of $\mathbb{H}^3$ in the ball model.

For the ball model coordinates $X_1, X_2, X_3$, we define a pure quaternion $X = X_je_j(R^2 = |X|^2)$. Using $\varphi$ as the coordinate along the circle, we introduce the \textbf{toroidal coordinates}\index{toroidal coordinates} of $\mathbb{R}^4$. Then $x$ can be written as
\begin{align}
x = \frac{2X + (1 - R^2)\sin\varphi}{1 + R^2 + (1 - R^2)\cos \varphi} . \label{eq:7}
\end{align}
The circular action \eqref{eq:6} corresponds to the rotation $\varphi\mapsto\varphi + \alpha$.

The standard form of ADHM data $\mathscr{M}$ corresponding to instanton number $n$ is a pair of quaternion matrices $L$ and $M$, where $L$ is a row of $n$ quaternions and $M$ is a $n \times n$ symmetric quaternion matrix:
\begin{align}
\mathscr{M} =
\begin{pmatrix}
L\\
M
\end{pmatrix}. 
\label{eq:8b}
\end{align}
This $\mathscr{M}$ is required to satisfy the quadratic constraint
\begin{align}
\mathscr{M}^\dagger \mathscr{M} = R_n ,
\end{align}
where $R_n$ is a real $n \times n$ matrix with an inverse, and the pure quaternion part is required to be zero. Thus, the ADHM operator $\Delta(x)$ is constructed from such $\mathscr{M}$:
\begin{align}
\Delta(x) = \mathscr{M} - \mathscr{U}x =
\begin{pmatrix}
L\\
M
\end{pmatrix}
-
\begin{pmatrix}
0\\
\bm{1}_n
\end{pmatrix}
x.
\end{align}
The \textbf{equivariant}\index{equivariant} ADHM data are obtained by applying the transformation $Q$:
\begin{align}
\mathscr{M} \to Q\mathscr{M} , \ \mathscr{U} \to Q\mathscr{U} , \ Q^\dagger Q = 1_{n+1}.
\end{align}
As a result, the ADHM data are no longer in the standard form. 

Now, following Manton-Sutcliffe(2014)\cite{MS14},  we introduce stronger conditions than \eqref{eq:8} and show that they are enough to show that the ADHM data are invariant under circular actions.

\noindent
(i) The quaternion matrix $M$ with $n \times n$ is pure quaternion and symmetric $(M^\dagger = -M, \ {}^t\! M = M)$.

\noindent
(ii) $\hat{M}^\dagger\hat{M} = {1}_n \Leftrightarrow L^\dagger L + M^\dagger M = {1}_n$.

\noindent
(iii) $LM = \mu L$, $\mu$ is a pure quaternion, $L \neq 0$.

Here $\mu$ is called the left eigenvalue of $M$: $(ML^\dagger = L^\dagger\mu)$.

Using (i), (ii) can be rewritten as

\noindent
(ii)' $L^\dagger L = {1}_n + M^2$.

On the other hand, the relation holds:

\noindent
(iv) $LL^\dagger = 1- |\mu|^2 = 1 + \mu^2$

This is true because by applying $M$ to the right-hand side of (iii) and using (iii) again, we obtain $LM^2 = \mu LM = \mu^2 L = -|\mu|^2L$. By eliminating $M^2$ using (ii)', we obtain $L(L^\dagger L - {1}_n) = -|\mu|^2 L \Leftrightarrow (LL^\dagger - {1}_n + |\mu|^2)L = 0$.

Under the general conformal transformation \eqref{eq:7}, the ADHM data (apart from the total factor on the right-hand side) transform as follows ($\mathscr{M}$ is no longer in the standard form):
\begin{align}
\mathscr{M} \to \mathscr{M}^\prime = \mathscr{M}D - \mathscr{U}B , \ \mathscr{U} \to \mathscr{U}^\prime = \mathscr{U}A - \mathscr{M}C .
\end{align}
Indeed, this can be verified as 
\begin{align}
& \mathscr{M} - \mathscr{U}x' =  \mathscr{M} - \mathscr{U}(Ax + B)(Cx + D)^{-1}  
\nonumber\\
=& [\mathscr{M}(Cx + D) - \mathscr{U}(Ax + B)](Cx + D)^{-1} \nonumber\\
=&[\mathscr{M}D - \mathscr{U}B - (\mathscr{U}A - \mathscr{M}C)x](Cx + D)^{-1} \nonumber\\
=&(\mathscr{M}^\prime-\mathscr{U}^\prime x)(Cx + D)^{-1}.
\end{align}
In the case of circular action \eqref{eq:6}, this transformation reads
\begin{align}
\mathscr{M} = \mathscr{M}\cos\frac{\alpha}{2} + \mathscr{U}\sin\frac{\alpha}{2}, \ \mathscr{U} = \mathscr{U}\cos\frac{\alpha}{2} - \mathscr{M}\sin\frac{\alpha}{2} .
\end{align}
In order to show that \textbf{the constrained ADHM data is circularly invariant}, we need the matrix $Q$ to convert the matrix $Q$ data to a standard form. Using the conditions (i), (iii) and the relation (ii)', we find that such $Q$ is given by 
\begin{align}
Q=
\begin{pmatrix}
\cos\frac{\alpha}{2} + \mu\sin\frac{\alpha}{2}&-L\sin\frac{\alpha}{2}\\
L^\dagger\sin\frac{\alpha}{2}&{1}_n\cos\frac{\alpha}{2}-M\sin\frac{\alpha}{2}
\end{pmatrix}. \label{eq:9}
\end{align}
Indeed, we can check that $Q^\dagger Q = {1}_{n+1}$, and show by straightforward calculations that $QU = U'$ and $Q\hat{M} = \hat{M'}$. As a result, the ADHM data have the required circular invariance and gives a hyperbolic magnetic monopole:
{\footnotesize
\begin{align}
 Q^\dagger Q &=  
\begin{pmatrix}
\cos\frac{\alpha}{2} -\mu\sin\frac{\alpha}{2}&L\sin\frac{\alpha}{2} \\
-L^\dagger\sin\frac{\alpha}{2}&{1}_n\cos\frac{\alpha}{2} + M\sin\frac{\alpha}{2}
\end{pmatrix}
\nonumber\\
\times &
\begin{pmatrix}
\cos\frac{\alpha}{2} +\mu\sin\frac{\alpha}{2}&-L\sin\frac{\alpha}{2} \\
L^\dagger\sin\frac{\alpha}{2}&{1}_n\cos\frac{\alpha}{2} - M\sin\frac{\alpha}{2}
\end{pmatrix}
\nonumber\\
&\stackrel{\rm{(ii), (iii), (iv)}}{=} 
\begin{pmatrix}
1&0\\
0&\bm{1}_n
\end{pmatrix}
=\bm{1}_{n+1}.
\end{align}}
Here, $Q \to Q^\dagger$ corresponds to $\mu \to -\mu, L \to -L, M \to -M$. Since the conditions (i), (ii), (iii), and (iv) are all unchanged under this transformation,  it turns out that $QQ^\dagger = 1_{n+1}$ also holds. This can also be confirmed by direct calculations.
We can also check $\mathscr{M} \to Q\mathscr{M} , \ \mathscr{U} \to Q\mathscr{U}$ by direct calculations:
\begin{align}
Q\mathscr{M}=&
\begin{pmatrix}
\cos\frac{\alpha}{2} + \mu\sin\frac{\alpha}{2}&-L\sin\frac{\alpha}{2} \\
L^\dagger\sin\frac{\alpha}{2}&{1}_n\cos\frac{\alpha}{2} - M\sin\frac{\alpha}{2}
\end{pmatrix}
\begin{pmatrix}
L\\
M
\end{pmatrix} \nonumber\\
=&
\begin{pmatrix}
L\cos\frac{\alpha}{2} + (\mu L - LM)\sin\frac{\alpha}{2}\\
(L^\dagger L - M^2)\sin\frac{\alpha}{2} + M\cos\frac{\alpha}{2}
\end{pmatrix}
\nonumber\\
\stackrel{\rm{(iii), (ii)'}}{=}&
\begin{pmatrix}
L\cos\frac{\alpha}{2}\\
{1}_n\sin\frac{\alpha}{2} + M\cos\frac{\alpha}{2}
\end{pmatrix}\nonumber\\
=&
\begin{pmatrix}
L\\
M
\end{pmatrix}
\cos\frac{\alpha}{2}+
\begin{pmatrix}
0\\
{1}_n
\end{pmatrix}
\sin\frac{\alpha}{2}
\nonumber\\
=& \hat{M}\cos\frac{\alpha}{2} + U\sin\frac{\alpha}{2} = \mathscr{M}^\prime ,
\\
Q\mathscr{U}=&Q
\begin{pmatrix}
0\\
1
\end{pmatrix}
=
\begin{pmatrix}
-L\sin\frac{\alpha}{2}\\
{1}_n\cos\frac{\alpha}{2} - M\sin\frac{\alpha}{2}
\end{pmatrix}
\nonumber\\
=&-
\begin{pmatrix}
L\\
M
\end{pmatrix}
\sin\frac{\alpha}{2}+
\begin{pmatrix}
0\\
{1}_n
\end{pmatrix}
\cos\frac{\alpha}{2} \nonumber\\
=&U\cos\frac{\alpha}{2} - \hat{M}\sin\frac{\alpha}{2} = \mathscr{U}^\prime .
\end{align}
From the transformation equation \eqref{eq:7} between $X$ and $x$, we find  $x = X$ when $\varphi = 0$. Therefore, for circularly invariant data, by setting $\alpha = \varphi$, the following relationship is satisfied at the point $x$ with toroidal coordinates $X$ and $\varphi$:
\begin{align}
\Delta(x) = Q^\dagger\Delta(x).
\end{align}
Here, $Q$ is \eqref{eq:9} with $\alpha = \varphi$. Therefore, the required vector $V$ can be written as $V = Q^\dagger V_0$ using the row vector $V_0 \ (V_0^\dagger V_0 = 1)$ of unit length. 
Note that $V_0$ depends only on the pure quaternion $X$:
\begin{align}
V_0(X)^\dagger \Delta(X) = 0.
\end{align}
As a result, the gauge field $\mathscr{A}_\mu$ does not depend on $\varphi$, which is what was needed to interpret instantons as hyperbolic magnetic monopoles. 
Thus we have
\begin{align}
&\mathscr{A}_\mu(x) = V(x)^\dagger\partial_\mu V(x) = V_0(X)^\dagger Q\partial_\mu(Q^\dagger V_0(X)) \nonumber\\
=&V_0(X)^\dagger Q\partial_\mu Q^\dagger V_0(X) + V_0(X)^\dagger QQ^\dagger\partial_\mu V_0(X) \nonumber\\
=&V_0(X)^\dagger Q\partial_\mu Q^\dagger V_0(X) + V_0(X)^\dagger \partial_\mu V_0(X) .
\end{align}
In particular, the monopole scalar field $\Phi$ is obtained using $\partial_\varphi V_0(X)=0$:
\begin{align}
\Phi = \mathscr{A}_\varphi = V_0^\dagger(Q\partial_\varphi Q^\dagger)V_0 , \ Q\partial_\varphi Q^\dagger = \frac{1}{2}
\begin{pmatrix}
-\mu&L\\
-L^\dagger&M
\end{pmatrix}.
\end{align}
Note that $Q\partial_\varphi Q^\dagger$ does not depend on $\varphi$.
In fact, this can be confirmed as  
{\scriptsize
\begin{align}
\partial_\varphi Q^\dagger =& \partial_\varphi
\begin{pmatrix}
\cos\frac{\varphi}{2}-\mu\sin\frac{\varphi}{2}&L\sin\frac{\varphi}{2}\\
-L^\dagger\sin\frac{\varphi}{2}&{1}_n\cos\frac{\varphi}{2} + M\sin\frac{\varphi}{2}
\end{pmatrix}\nonumber\\
=&
\begin{pmatrix}
-\frac{1}{2}\sin\frac{\varphi}{2}-\frac{1}{2}\mu\cos\frac{\varphi}{2}&\frac{1}{2}L\cos\frac{\varphi}{2}\\
-\frac{1}{2}L^\dagger\cos\frac{\varphi}{2}&-\frac{1}{2}{1}_n\sin\frac{\varphi}{2} + \frac{1}{2}M\cos\frac{\varphi}{2}
\end{pmatrix}, 
\nonumber\\
Q\partial_\varphi Q^\dagger =& \frac{1}{2}
\begin{pmatrix}
\cos\frac{\varphi}{2}+\mu\sin\frac{\varphi}{2}&-L\sin\frac{\varphi}{2}\\
L^\dagger\sin\frac{\varphi}{2}&{1}_n\cos\frac{\varphi}{2} - M\sin\frac{\varphi}{2}
\end{pmatrix}
\nonumber\\
\times &
\begin{pmatrix}
-\sin\frac{\varphi}{2}-\mu\cos\frac{\varphi}{2}&L\cos\frac{\varphi}{2}\\
-L^\dagger\cos\frac{\varphi}{2}&-{1}_n\sin\frac{\varphi}{2} + M\cos\frac{\varphi}{2}
\end{pmatrix}  \nonumber\\
=&\frac{1}{2}
\begin{pmatrix}
-\mu&L \\
-L^\dagger&M
\end{pmatrix}.
\end{align}}
Here, we have used $LL^\dagger = 1 + \mu^2, ML^\dagger = \mu L^\dagger, LM = \mu L$ and $L^\dagger L = 1 + M^2$.

The left eigenvalue $\mu$ has a physical meaning because it is related to the value at the origin of the scalar field. In particular, for $X = 0$, $V_0$ that satisfies $V_0^\dagger \Delta(X) = 0$ is
$V_0 = \begin{pmatrix}
\mu \\
L^\dagger
\end{pmatrix}.
$
for  
$\Delta(X = 0) = 
\begin{pmatrix}
L\\
M
\end{pmatrix}
$. 
In fact, we confirm
$V_0^\dagger \Delta(X = 0) = \begin{pmatrix}
-\mu&L
\end{pmatrix}
\begin{pmatrix}
L \\
M
\end{pmatrix}
=-\mu L + LM = 0$. 
Therefore, we obtain the scalar field $\Phi(0)$ at $X = 0$:
\begin{align}
\Phi(0) =& \mathscr{A}_\varphi(0) = V_0^\dagger\frac{1}{2}
\begin{pmatrix}
-\mu & L\\
-L^\dagger&M
\end{pmatrix}
V_0 
\nonumber\\
=&
\begin{pmatrix}
-\mu & L
\end{pmatrix}
\frac{1}{2}
\begin{pmatrix}
-\mu & L\\
-L^\dagger&M
\end{pmatrix}
\begin{pmatrix}
\mu \\
L^\dagger
\end{pmatrix}
\nonumber\\
=&\frac{1}{2}
\begin{pmatrix}
-\mu & L
\end{pmatrix}
\begin{pmatrix}
-\mu^2 + LL^\dagger \\
-L^\dagger\mu + ML^\dagger
\end{pmatrix}
\nonumber\\
=&\frac{1}{2}
\begin{pmatrix}
-\mu & L
\end{pmatrix}
\begin{pmatrix}
1\\
0
\end{pmatrix}
=-\frac{1}{2}\mu .
\end{align}

\begin{example}[1-magnetic monopole]
For $n = 1$, $\hat{M}$ that satisfies the constraints is given e.g., by
\begin{align}
\hat{M}=
\begin{pmatrix}
\sqrt{1 - a^2}\\
ai
\end{pmatrix}
\ (|a|<1) . 
\label{eq:10} 
\end{align}
In fact, (i) $M = ai$ is a pure quaternion and a pure imaginary number.  
(ii) $\hat{M}^\dagger\hat{M} =
\begin{pmatrix}
\sqrt{1 - a^2}&-ai
\end{pmatrix}
\begin{pmatrix}
\sqrt{1 - a^2}\\
ai
\end{pmatrix}
=1 - a^2 + a^2 =1
$. 
(iii) $L = \sqrt{1-a^2} \neq 0 \ (a^2 \neq 1)$ yields $LM=\mu L$, and therefore $M = \mu = ai$, implying that $\mu$ is a pure imaginary number.

Next, we compute the ADHM operator:
\begin{align}
&\Delta(x) =
\begin{pmatrix}
\sqrt{1 - a^2}\\
ai
\end{pmatrix}
-
\begin{pmatrix}
0\\
1
\end{pmatrix}
x 
\nonumber\\
& (x = x_\mu e_\mu , \ e_j = -i\sigma_j , \ e_4 =1).
\end{align}
Using
$\sigma_1 =
\begin{pmatrix}
0&1\\
1&0
\end{pmatrix}
, \ \sigma_2 =
\begin{pmatrix}
0&-i\\
i&0
\end{pmatrix}
, \ \sigma_3 =
\begin{pmatrix}
1&0\\
0&-1
\end{pmatrix}
$, 
the second term reads
\begin{align}
& x_\mu e_\mu
\begin{pmatrix}
0\\
1
\end{pmatrix}
\nonumber\\
=& -ix_1\sigma_1
\begin{pmatrix}
0\\
1
\end{pmatrix}
-ix_2\sigma_2
\begin{pmatrix}
0\\
1
\end{pmatrix}
-ix_3\sigma_3
\begin{pmatrix}
0\\
1
\end{pmatrix}
+x_4
\begin{pmatrix}
0\\
1
\end{pmatrix} \nonumber\\
=&-ix_1
\begin{pmatrix}
1\\
0
\end{pmatrix}
-ix_2
\begin{pmatrix}
-i\\
0
\end{pmatrix}
-ix_3
\begin{pmatrix}
0\\
-1
\end{pmatrix}
+x_4
\begin{pmatrix}
0\\
1
\end{pmatrix}
\nonumber\\
=&
\begin{pmatrix}
-ix_1 -x_2\\
ix_3 + x_4
\end{pmatrix} .
\end{align}
Thus, we obtain the ADHM operator:
\begin{align}
\Delta(x) =
\begin{pmatrix}
\sqrt{1-a^2}-ix_1 -x_2\\
ai + ix_3 +x_4
\end{pmatrix}
:=
\begin{pmatrix}
\Delta_1\\
\Delta_2
\end{pmatrix}.
\end{align}
Furthermore, we find the vector $V$:
\begin{align}
& V =
\begin{pmatrix}
v_1\\
v_2
\end{pmatrix}
\Rightarrow V^\dagger\Delta(x) =
\begin{pmatrix}
v_1^*&v_2^*
\end{pmatrix}
\begin{pmatrix}
\Delta_1\\
\Delta_2
\end{pmatrix}
=0 \nonumber\\
&
\Leftrightarrow v_1^*\Delta_1 + v_2^*\Delta_2 = 0 \Rightarrow \frac{v_1^*}{v_2^*} = -\frac{\Delta_2}{\Delta_1} \Rightarrow \frac{v_1}{v_2} = -\frac{\Delta_2^*}{\Delta_1^*} .
\end{align}
By taking into account the normalization $V^\dagger V = 1 \Leftrightarrow |v_1|^2 +|v_2|^2 = 1$, we obtain the vector $V$:
\begin{align}
V =&\frac{1}{N}
\begin{pmatrix}
ai + ix_3 - x_4 ,
\\
\sqrt{1 - a^2} + ix_1 - x_2
\end{pmatrix} , \nonumber\\
N^2 := &x_1^2 + (x_2 - \sqrt{1 - a^2})^2 +(x_3 +a)^2 + x_4^2 .
\end{align}
$\hat{M}$ in \eqref{eq:10} gives a hyperbolic 1-monopole centered at $X_1 = (1 - \sqrt{1 - a^2})/a$ along the $X_1$ axis.

The scalar field $\Phi(0)$ at the origin $X = 0$ is obtained:
\begin{align}
|\Phi(0)| = \frac{1}{2}|\mu| = \frac{1}{2}|a| .
\end{align}
The scalar field $\Phi(X)$ at a general point on the unit ball has the absolute value $|\Phi(X)|$:
\begin{align}
|\Phi|^2 =\frac{4R^2 - 4aX_1(1 + R^2) + a^2[4X_1^2 + (1 - R^2)^2]}{4(1 + R^2 -2aX_1)^2}.
\end{align}
Note that the absolute value $|\Phi(X)|$ of the scalar field $\Phi(X)$ is a gauge invariant. The simplest example is the case where $a = 0$
\begin{align}
\hat{M} =
\begin{pmatrix}
1\\
0
\end{pmatrix} ,
\end{align}
which corresponds to a magnetic monopole centered at the origin. In this case, we can easily re-derive the magnitude of the scalar field
\begin{align}
|\Phi(X)|^2 = \frac{R^2}{(1 + R^2)^2} ,
\end{align}
and the energy density
\begin{align}
\epsilon = \frac{3}{2}\left(\frac{1-R^2}{1 + R^2}\right)^4 .
\end{align}
Note that all these monopole examples are spherically symmetric around their centers, but only the $a = 0$ case has the standard $SO(3)$ as its symmetry group.

See Chan(2017)\cite{Chan17} for more details on the hyperbolic magnetic monopole as an $S^1$-invariant instanton. 

\end{example}



\end{document}